\documentclass[11pt]{article}

\usepackage{amssymb}
\usepackage{graphicx}

\usepackage[%
boxruled,
lined,
ruled,
linesnumbered,
noresetcount,
nokwfunc
]%
{algorithm2e}
\DontPrintSemicolon
\SetAlgoSkip{normalskip}
\SetAlgoInsideSkip{medskip}

\usepackage{
amsmath,
amsfonts,%
amssymb,%
enumerate,
hyperref,
wrapfig,
graphicx,
framed,
dsfont,
amsthm
}

\usepackage[text={6.5in,9in},centering]{geometry}
\usepackage{privmath}
\usepackage{epstopdf}
\usepackage{float}

\usepackage{mythesis-macros}

\usepackage{url}
\urldef{\mailsa}\path|{woelfel,apareek}@ucalgary.ca|

\begin{document}


\title{RMR-Efficient Randomized Abortable Mutual Exclusion%
\thanks{This research was supported by a Discovery Grant from the Natural Sciences and Engineering
Research Council of Canada (NSERC)}%
}


%
\author{Philipp Woelfel\\ University of Calgary\\ {\tt woelfel@ucalgary.ca}
\and Abhijeet Pareek\thanks{Supported by Alberta Innovates Technology Futures (AITF)}\\
University of Calgary\\ {\tt abhijeet.ucalgary@gmail.com}
}
%
%


%
%

\maketitle

\sloppy
%
%
\begin{abstract}
Recent research on mutual exclusion for shared-memory systems has focused on \emph{local spin} algorithms.
Performance is measured using the \emph{remote memory references} (RMRs) metric.
As common in recent literature, we consider a standard asynchronous shared memory model with $N$ processes, which allows atomic read, write and compare-and-swap (short: CAS) operations.

In such a model, the asymptotically tight upper and lower bounds on the number of RMRs per passage through the Critical Section is $\Theta(\log N)$ for the optimal deterministic algorithms \cite{YA:YALock,Attiya:MELowerBound}.
Recently, several \emph{randomized} algorithms have been devised that break the $\Omega(\log N)$ barrier and need only $o(\log N)$ RMRs per passage in expectation \cite{Philipp:RLock,PhilippHendler:AdaptiveRLock,BenderGilbert:ME}.
In this paper we present the first randomized \emph{abortable} mutual exclusion algorithm that achieves a sub-logarithmic expected RMR complexity.
More precisely, against a weak adversary (which can make scheduling decisions based on the entire past history, but not the latest coin-flips of each process) every process needs an expected number of $O(\log N/\log\log N)$ RMRs to enter end exit the critical section.
If a process receives an abort-signal, it can abort an attempt to enter the critical section within a finite number of its own steps and by incurring $O(\log N/\log\log N)$ RMRs.
\end{abstract}

\section{Introduction}
\emph{Mutual exclusion}, introduced by Dijkstra \cite{Dijkstra:MEProblem}, is a fundamental and well studied problem.
A mutual exclusion object (or lock) allows processes
to synchronize access to a shared resource.
Each process obtains a lock through a \emph{capture protocol} 
but at any time, at most one process can own the lock.
A process is said to own a lock if it participates in a ``capture'' protocol designed for the object, and completes it.
The owner of the lock can access the shared resource, while all other processes wait in their capture protocol for the owner to ``release'' the lock.
The owner of a lock can execute a release protocol which frees up the lock. 
The capture protocol and release protocol are often denoted \emph{entry} and \emph{exit section}, and a process that owns the lock is in the \emph{critical section}.

In this paper, we consider the standard \emph{cache-coherent} (CC) shared model with $N$ processes that supports atomic read, write, and compare-and-swap (short: \CAS) operations. 
In this model, all shared registers are stored in globally accessible shared memory.
In addition, each process has a local cache and a cache protocol ensures coherency.
A \emph{Remote Memory Reference} (short: RMR) is a shared memory access of a register that cannot be resolved locally (i.e., a cache miss).
Mutual exclusion algorithms require processes to busy-wait, so the traditional step complexity measure, which counts the number of shared memory accesses, is not useful.

Early mutual exclusion locks were designed for uniprocessor systems that supported multitasking and time-sharing.
A comprehensive survey of these locking algorithms is presented in \cite{Raynal:Survey}.
One of the biggest shortcomings of these early locking algorithms is that they did not take into account an important hardware technology trend -- the steadily growing gap between high processor speeds and the low speed/bandwidth of the processor-memory interconnect \cite{Culler:GapInProcessorAndInterconnectSpeeds}.
A memory access that traverses the processor-to-memory interconnect, called a \emph{remote memory reference}, takes much more time than a local memory access.

Recent research \cite{Anderson:SpinLockPerformance,Mellor-Crummey:QueueLock,AK:MEConjecture,AYK:Adaptive,AYK:ImprovedLowerBound,Attiya:MELowerBound,DanekGolab:FCFS,Jayanti:Abortable,KA:AdaptiveLowerBound} on mutual exclusion algorithms therefore focusses on minimizing the number of remote memory references (RMR).
The maximum number of RMRs that any process requires (in any execution) to capture and release a lock is called the RMR complexity of the mutual exclusion algorithm.
RMR complexity is the metric used to analyze the efficiency of mutual exclusion algorithms, as opposed to the traditional metric of counting steps taken by a process (step complexity).
Step complexity is problematic, since for mutual exclusion algorithms, a process may perform an unbounded number of memory accesses (each considered a step) while busy-waiting for another process to release the lock \cite{Taubenfeld:FastMutualExclusion}.

Algorithms that perform all busy-waiting by repeatedly reading \emph{locally accessible} shared variables, achieve bounded RMR complexity and have practical performance benefits \cite{Anderson:SpinLockPerformance}.
Such algorithms are termed  \emph{local spin} algorithms.
A comprehensive survey of these algorithms is presented in \cite{Anderson:Survey}.
Yang and Anderson presented the first $\Order{\log N}$ RMRs mutual exclusion algorithm \cite{YA:YALock} using only reads and writes.
Anderson and Kim \cite{AK:MEConjecture} conjectured that this was optimal, and the conjecture was proved by Attiya, Hendler, and Woelfel \cite{Attiya:MELowerBound}.

Local spin mutual exclusion locks do not meet a critical demand of many systems \cite{Scott:non-blocking}.
Specifically, the locks employed in database systems and in real time systems must support a ``timeout'' capability which allows
a process that waits ``too long'' to abort its attempt to acquire the lock.
The ability of a thread to abort its lock attempt is crucial in data base systems; for instance in Oracle's Parallel Server and IBM's DB2, this ability serves the dual purpose of recovering from transaction deadlock and tolerating preemption of the thread that holds the lock \cite{Scott:non-blocking}.
In real time systems, the abort capability can be used to avoid overshooting a deadline.
Locks that allow a process to abort its attempt to acquire the lock are called abortable locks.
Jayanti presented an efficient deterministic abortable lock \cite{Jayanti:Abortable} with worst-case \Order{\log N} RMR complexity, which is optimal for deterministic algorithms. 

In this paper we present the first randomized abortable mutual exclusion algorithm that achieves a sub-logarithmic RMR complexity.
Due to the inherent asynchrony in the system, the RMRs incurred by a process  during a lock capture and release depend on how the steps of all the processes in the system were scheduled one after the other.
Therefore, the maximum RMRs incurred by any process during any lock attempt are  determined by the ``worst'' schedule that makes some process incur a large number of RMRs.
To analyze the RMR complexity of lock algorithms, an adversarial scheduler called the \emph{adversary} is defined.
The lower bound of $\Omega(\log N)$ in \cite{Attiya:MELowerBound} for mutual exclusion algorithms that use only reads and writes holds for deterministic algorithms where the adversary knows all processes' future steps.
The lower bound does not hold for randomized algorithms where processes flip coins to determine their next steps.
Randomized algorithms limit the power of an adversary since the adversary cannot know the result of future coin flips.
Adversaries of varying powers have been defined.
The most common ones are the \emph{oblivious}, the \emph{weak}, and the \emph{adaptive} adversary \cite{Aspnes:AdversaryModels}.
An \emph{oblivious} adversary makes all scheduling decisions in advance, before any process flips a coin.
This model corresponds to a system, where the coin flips made by processes have no influence on the scheduling.
A more realistic model is the \emph{weak} adversary, who sees the coin flip of a process not before that process has taken a step following that coin flip.
The \emph{adaptive} adversary models the strongest adversary with reasonable powers, and it can see every coin flip as it appears, and can use that knowledge for any future scheduling decisions.
Hendler and Woelfel \cite{Philipp:RLock} and later Giakkoupis and Woelfel \cite{PhilippGeorge:LowerBoundRLockAdaptiveAdv} established a tight bound of $\Theta(\log N / \log \log N)$ expected RMR complexity for randomized mutual exclusion against the adaptive adversary. 
Recently Bender and Gilbert \cite{BenderGilbert:ME} presented a randomized lock that has amortized $\Order{\log^2 \log N}$ expected RMR complexity against the oblivious adversary.
Unfortunately, this algorithm is not strictly deadlock-free (processes may deadlock with small probability, so deadlock has to be expected in a long execution).
Our randomized abortable mutual exclusion algorithm is deadlock-free, works against the weak adversary and achieves the same epected RMR complexity as the algorithm by Hendler and Woelfel, namely $\Order{\log N / \log \log N}$ expected RMR complexity against the weak adversary.

The randomized algorithm we present uses \CAS\ objects and read-write registers.
Golab, Hadzilacos, Hendler, and Woelfel \cite{Golab:CASImplementation} (see also \cite{Golab:Thesis}) presented an \Order{1}-RMRs implementation of a \CAS\ object using only read-write registers.
Moreover, they proved that one can simulate any deterministic shared memory algorithm that uses reads, writes, and conditional operations (such as \CAS operations), with a deterministic algorithm that uses only reads and writes, with only a constant increase in the RMR complexity.
Recently in \cite{Philipp:LinearizabilityNeqAtomicity}, Golab, Higham and Woelfel  demonstrated that using linearizable implemented objects in place of atomic objects in randomized algorithms allows the adversary to change the probability distribution of results.
Therefore, in order to safely use implemented objects in place of atomic ones in randomized algorithms, it is not enough to simply show that the implemented objects are linearizable.
Also in \cite{Philipp:LinearizabilityNeqAtomicity}, it is proved that there exists no general correctness condition for the weak adversary, and that the weak adversary can gain additional power depending on the linearizable  implementation of the object.
Therefore, in this paper we assume that \CAS operations are atomic.




\bparagraph{Abortable Mutual Exclusion}\label{sec:Problem}
We formalize the notion of an abortable lock by specifying two methods, \lock{} and \release{}{}, that processes can use to capture and release the lock, respectively.
The model assumes that a process may receive a signal to abort at any time during its \lock{} call.
If that happens, and only then, the process may fail to capture the lock, in which case method \lock{} returns value $\bot$.
Otherwise the process captures the lock, and method \lock{} returns a non-$\bot$ value, and the \lock{} call is deemed \emph{successful}.
Note that a \lock{} call may succeed even if the process receives a signal to abort during a \lock{} call.


Code executed by a process after a successful \lock{} method call and before a subsequent \release{}{} invocation is defined to be its Critical Section.
If a process executes a successful \lock{} call, then the process's \emph{passage} is defined to be the \lock{} call, and the subsequent Critical Section and \release{}{} call, in that order.
If a process executes an unsuccessful \lock{} call, then it does not execute the Critical Section or a \release{}{} call, and the process's passage is just the \lock{} call.
Code executed by a process outside of any passage is defined to be its Remainder Section.


The \emph{abort-way} is defined to be the steps taken by a process during a passage that begins when the process receives a signal to abort and ends when the process returns to its Remainder Section.
Since it makes little sense to have an abort capability where processes have to wait for other processes, the abort-way is required to be bounded wait-free (i.e., processes execute the abort-way in a bounded number of their own steps).
This property is known as \emph{bounded abort}.
Other properties are defined as follows.
\emph{Mutual Exclusion}: At any time there is at most one process in the Critical Section;
\emph{Deadlock Freedom}: If all processes in the system take enough steps, then at least one of them will return from its \lock{} call;
\emph{Starvation Freedom}: If all processes in the system take enough steps, then every process will return from its \lock{} call.
The abortable mutual exclusion problem is to implement an object that provides methods \lock{} and \release{}{} such that it that satisfies mutual exclusion, deadlock freedom, and bounded abort.

\subsection{Model}\label{sec:Model}
Our model of computation, the asynchronous shared-memory model \cite{Herlihy:Linearizability} with
$N$ processes which communicate by executing \emph{operations} on
\emph{shared objects}.
Every process executes its program by taking \emph{steps}, and does not fail.
A step is defined to be the execution of all local computations followed by an operation on a shared object.
We consider a system that supports atomic read-write registers and \CAS{} objects.

A read-write register $R$ stores a value from some set and supports two atomic operations $R$.\Read{} and $R$.\Write{}.
Operation $R$.\Read{} returns the value of the register and leaves its content unchanged, and operation $R$.\Write{v} writes the value $v$ into the register and returns nothing.
A \CAS object $O$ stores a value from some set and supports two atomic operations $O$.\CAS{} and $O$.\Read{}.
Operation $O$.\Read{} returns the value stored in $O$.
Operation $O$.\CAS{$exp,new$} takes two arguments $exp$ and $new$ and attempts to change the value of $O$ from $exp$ to $new$.
If the value of $O$ equals $exp$ then the operation $O$.\CAS{$exp,new$} \emph{succeeds}, and the value of $O$ is changed from $exp$ to $new$, and \True\ is returned.
Otherwise, the operation fails, and the value of $O$ remains unchanged and \False\ is returned.

In addition, a process can execute local coin flip operations that returns an  integer value distributed uniformly at random from an arbitrary finite set of integers.
The scheduling, generated by the \emph{adversary}, can depend on the random values generated by the processes.
We assume the weak adversary model (see for example \cite{Aspnes:AdversaryModels}) that decides at each point in time the process that takes the next step.
In order to make this decision, it can take all preceding events into account, except the results of the most recent coin flips by processes that are yet to execute a shared memory operation after the coin flip.

As mentioned earlier, we consider the \emph{cache-coherent} (CC) model where each processor has a private cache in which it maintains local copies of shared objects that it accesses.
The private cache is logically situated ``closer'' to the processor than the shared memory, and therefore it can be accessed for free.
The shared memory is an external memory accessible to all processors, and is considered remote to all processors.
We assume that a  hardware protocol ensures cache consistency (i.e., that all copies of the same object in different caches are valid and consistent).
A memory access to a shared object that requires access to remote memory is called a \emph{remote memory reference} (RMR).
The \emph{RMR complexity} of a algorithm is the maximum number of RMRs that a process can incur during any execution of the algorithm.

\subsection{Results}

We present several building blocks for our algorithm 
in Section~\ref{sec:RCASCounter}.
In Sections~\ref{sec:ARLockTree} and \ref{sec:ARLockArray} we give an overview of the randomized mutual exclusion algorithm.
Our results are summarized by the following theorem.
\begin{theorem}\label{thm:main}
There exists a starvation-free randomized abortable $N$ process lock against the weak adversary, where a process incurs $\Order{\log N / \log \log N}$ RMRs in expectation per passage.
The lock requires \Order{N} \CAS\ objects and read-write registers
\end{theorem}
%

\section{Building Blocks}
\label{sec:BuildingBlocks}

\paragraph{A Randomized CAS Counter.}
\label{sec:RCASCounter}

A \emph{CAS counter} object with parameter $k \in \mathds{Z}^+$ complements a \CAS\ object by supporting an additional \inc{} operation (apart from \CAS{} and \Read{} operations) that increments the object's value.
The object takes values in \Set{0,\ldots,k}, and initially the object's value is $0$.
Operation \inc{} takes no arguments, and if the value of the object is in \Set{0,\ldots,k-1}, then the operation increments the value and returns the previous value.
Otherwise, the value of the object is unchanged and the integer $k$ is returned.
We will use such an object for $k=2$ to assign three distinct roles to processes.

Our implementation of the \inc{} operation needs only $O(1)$ RMRs in expectation.
A deterministic implementation of a \CAS counter for $k=2$ and constant worst-case RMR complexity does not exist:
Replacing our randomized \CAS counter with a deterministic one that has worst-case RMR complexity $T$ yields a deterministic abortable mutual exclusion algorithm with worst-case RMR complexity $\Order{T\cdot\log N/\log\log N}$.
From the lower bound for deterministic mutual exclusion by Attiya etal.~\cite{Attiya:MELowerBound}, such an algorithm does not exist, unless $T=\Omega(\log\log N)$.%
\footnote{For the DSM model, this also follows from a result by Golab, Hadzilacos, Hendler, and Woelfel~\cite{Golab:CASImplementation}. They established a super-constant lower bound on the RMR complexity of a deterministic bounded counter that can count up to two, and also supports a reset operation.}

In Appendix~\ref{sec:appendix:RCASCounter}, we describe a randomized CAS counter, called \RCASCounter{k}, where the \inc{} method is allowed to fail.
The idea is, that to increase the value of the object, a process randomly guesses its current value, $v$, and then executes a \CAS{$v$,$v+1$} operation.
An adaptive adversary could intervene between the steps involving the random guess and the subsequent \CAS operation, thereby affecting the failure probability of an \inc{} method call, but a weak adversary cannot do so.
\begin{lemma}
\label{thm:ctr:linearizable}
Object \RCASCounter{k} is a randomized wait-free linearizable CAS Counter, where the probability that an \inc{} method call fails is $\frac{k}{k+1}$ against the weak adversary.
Each of the methods of $\RCASCounter{k}$ has \Order{1} step complexity.
\end{lemma}

\paragraph{A Single-Fast-Multi-Slow Universal Construction.}
\label{sec:SFMSUniversalConstruction}

A \emph{universal construction} object provides a linearizable concurrent implementation of any object with a sequential specification that can be given by deterministic code.
In Appendix~\ref{appendix:sec:AbortableProArray} we devise a universal construction object \UC{\T} for $N$ processes \footnote{We use the universal construction object for smaller sets of processes, specifically for sets of size \Order{\log N / \log \log N}.} which provides two methods, \performFast{$op$} and \performSlow{$op$}, to perform any operation $op$ on an object of type \T.
The idea is that \performFast{} methods cannot be called concurrently, but are executed very fast, i.e., they have \Order{1} step complexity.
On the other hand, \performSlow{} methods need \Order{N} steps.
The algorithm is based on a helping mechanism in which \performSlow{} methods help a process that wants to execute a \performFast{} method.
%
\begin{lemma}\samepage
\label{theorem:UC}
Object \UC{\T} is a wait-free universal construction that implements an object \Obj\ of type \T, for $N$ processes,
 and an operation $op$ on object \Obj\ is performed by executing either method \performFast{op} or \performSlow{op}, and no two processes execute method \performFast{} concurrently.
Methods \performFast{} and \performSlow{} have \Order{1} and \Order{N} step complexity respectively.
\end{lemma}



%
%

\paragraph{The Abortable Promotion Array.}
\label{sec:AbortableProArray}
An object $O$ of type \APArray{k} stores a vector of $k$ integer pairs.
It provides some specialized operations on the vector, such as conditionally adding/removing elements, and earmarking a process (associated with an element of the vector) for some future activity.
Initially the value of $O=(O[0],O[1],\ldots,O[k-1])$ is $(\pair{0}{\bot},\ldots,\pair{0}{\bot})$.
The object supports operations \collect{}, \abort{}, \promote{}, \remove{} and \reset{} (see Figure~\ref{fig:SArray:seq} in the appendix).
Operation \collect{$X$} takes as argument an array $X[0 \ldots k-1]$ of integers, and is used to ``register'' processes into the array.
The operation changes $O[i]$, for all $i$ in \Set{0,\ldots,k-1}, to value \pair{\cReg}{X[i]} except if $O[i]$ is $\pair{\cAbort}{s}$, for some $s \in \mathbb{Z}$.
In the latter case the value of $O[i]$ is unchanged.
Process $i$ is said to be \emph{registered} in the array if a \collect{} operation changes $O[i]$ to value \pair{\cReg}{s}, for some $s \in \mathbb{Z}$.
The object also allows processes to ``abort'' themselves from the array using the operation \abort{}.
Operation \cUpdate{$i,s$} takes as argument the integers $i$ and $s$, where $i \in \Set{0,\ldots,k-1}$ and $s \in \mathbb{Z}$.
The operation changes $O[i]$ to value \pair{\cAbort}{s} and returns \True, only if $O[i]$ is not equal to \pair{\cPro}{s'}, for some $s' \in \mathbb{Z}$.
Otherwise the operation returns \False.
Process $i$ \emph{aborts} from the array if it executes an \cUpdate{$i,s$} operation that returns \True.
A registered process in the array that has not aborted can be ``promoted'' using the \promote{} operation.
Operation \promote{} takes no arguments, and
changes the value of the element in $O$ with the smallest index and that has value \pair{\cReg}{s}, for some $s \in \mathbb{Z}$, to value \pair{\cPro}{s}, and returns \pair{i}{s}, where $i$ is the index of that element.
If there exists no element in $O$ with value \pair{\cReg}{s}, for some $s \in \mathbb{Z}$, then $O$ is unchanged and the value $\pair{\bot}{\bot}$ is returned.
Process $i$ is \emph{promoted} if a \promote{} operation returns \pair{i}{s}, for some $s \in \mathbb{Z}$.
Operation \reset{} resets the entire array to its initial state.

Note that an aborted process in the array, cannot be registered into the array or be promoted, until the array is reset.
If a process tries to abort itself from the array but finds that it has already been promoted, then the abort fails.
This ensures that a promoted process takes responsibility for some activity that other processes expect of it.


In the context of our abortable lock, the $i$-th element of the array stores the current state of process with ID $i$, and a sequence number associated with the state.
Operation \collect{} is used to register a set of participating processes into the array.
Operation \cUpdate{$i,s$} is executed only by process $i$, to abort from the array.
Operation \promote{} is used to promote an unaborted registered process from the array, so that the promoted process can fulfill some future obligation.

In our abortable lock of Section~\ref{sec:ARLockTree}, we need a wait-free linearizable implementation of type \APArray{\Delta}, where $\Delta$ is the maximum number of processes that can access the object concurrently, and we achieve this by using object \UC{\APArray{\Delta}}.
We ensure that no two processes execute operations \collect{}, \promote{}, \reset{} or \remove{} concurrently, and therefore by we get \Order{1} step complexity for these operations by using method \performFast{}.
Operation \abort{} has \Order{\Delta} step complexity since it is performed using method \performSlow{}, which allows processes to call \abort{} concurrently.


%
\section{The Tree Based Abortable Lock}
\label{sec:ARMLockStateless}
\label{sec:ARLockTree}

Our abortable lock algorithm is based on an arbitration tree with branching factor approximately $\Theta(\log N/\log\log N)$.
For convenience we assume (w.l.o.g.) that $N = \Delta^{\Delta-1}$ for some positive integer $\Delta$, where $N$ is the maximum number of processes in the system.
Then it follows that $\Delta = \Theta (\log N $/ $\log \log N)$.

As in the algorithm by Hendler and Woelfel \cite{Philipp:RLock}, we consider a tree with $N$ leafs and where each non-leaf node has $\Delta$ children.
Every non-leaf node is associated with a lock.
Each process is assigned a unique leaf in the tree and climbs up the tree by capturing the locks on nodes on its path until it has captured the lock at the root.
Once a process locks the root, it can enter the Critical Section.

The main difficulty is that of designing the locks associated with the nodes of the tree.
A simple \CAS\ object together with an ``announce array'' as used in \cite{Philipp:RLock} does not work.
Suppose a process $p$ captures locks of several nodes on its path up to the root and aborts before capturing the root lock.
Then it must release all captured node locks and therefore these lock releases cause other processes, which are  busy-waiting on these nodes, to incur RMRs.
So we need a mechanism to guarantee some progress to these processes, while we also need a mechanism that allows busy-waiting processes to abort their attempts to capture node locks.
In \cite{Philipp:RLock} progress is achieved as follows: A process $p$, before releasing a lock on its path, searches(with a random procedure) for other processes that are busy-waiting for the node lock to become free.
If $p$ finds such a process, it promotes it into the critical section.
This is possible, because at the time of the promotion $p$ owns the root lock and can hand it over to a promoted process.
Unfortunately, this promotion mechanism fails for abortable mutual exclusion: When $p$ aborts its own attempt to enter the Critical Section, it may have to release node locks at a time when it doesn't own the root lock.
Another problem is that if $p$ finds a process $q$ that is waiting for $p$ to release a node-lock, then $q$ may have already decided to abort.
We use a carefully designed synchronization mechanism to deal with such cases.	

To ensure that waiting processes make some progress, we desire that $p$ ``collect'' busy-waiting processes (if any) at a node into an instance of an object of type \APArray{\Delta}, \PawnSet,  using the operation \collect{}.
Once busy-waiting processes are collected into \PawnSet, $p$ can identify a busy-waiting process, if present, using the \PawnSet.\promote{} operation,
while busy-waiting processes themselves can abort using the \PawnSet.\abort{} operation.
Note that $p$ may have to read $\Order{\Delta}$ registers just to find a single busy-waiting process at a node, where $\Delta$ is the branching factor of the arbitration tree.
This is problematic since our goal is to bound the number of steps during a passage to \Order{\Delta} steps, and thus a process cannot collect at more than one node.
For this reason we desire that $p$ transfer all unreleased node locks that it owns to the first busy-waiting process it can find, and then it would be done.
And if there are no busy-waiting processes at a node, then $p$ should somehow be able to release the node lock in \Order{1} steps.
Since there are at most $\Delta$ nodes on a path to the root node, $p$ can continue to release captured node locks where there are no busy-waiting processes, and thus not incur more than \Order{\Delta} overall.
We use an instance of \RCASCounter{2}, \ctr, to help decide if there are any busy-waiting processes at a node lock.
Initially, \ctr\ is $0$, and processes attempt to increase \ctr\ using the \ctr.\inc{} operation after having registered at the node.
Process $p$ attempts to release a node lock by first executing a \ctr.\CAS{$1,0$} operation.
If the operation fails then some process $q$ must have further increased \ctr\ from $1$ to $2$, and thus $p$ can transfer all unreleased locks to $q$, if $q$ has not aborted itself.
If $q$ has aborted, then $q$ can perform the collect at the node lock for $p$, since $q$ can afford to incur an additional one-time expense of \Order{\Delta} RMRs.
If $q$ has not aborted then $p$ can transfer its captured locks to $q$ in \Order{1} steps, and thus making sure some process makes progress towards capturing the root lock.
We encapsulate these mechanisms in a randomized abortable lock object, \ARMLockArray{\Delta}.

More generally, we specify an object \ARMLockArray{n} for an arbitrary parameter $n< N$.
Object \ARMLockArray{n} provides methods \lock{} and \release{}{} that can be accessed by at most $n+1$ processes concurrently. 
The object is an abortable lock, but with an RMR complexity of $O(n)$ for the abort-way, and constant RMR complexity for \lock{}.
The \release{}{} method is special.
If it detects contention (i.e., other processes are busy-waiting), then it takes $O(n)$ RMRs, but helps those other processes to make progress.
Otherwise, it takes only $O(1)$ RMRs.
Each non-leaf node $u$ in our abritration tree will be associated with a lock \ARMLockArray{\Delta} and can only be accessed concurrently by the processes owning locks associated with the children of $u$ and one other process.

Method \lock{} takes a single argument, which we will call pseudo-ID, with value in \Set{0,\ldots,n-1}.
We denote a \lock{} method call with argument $i$ as \lock{i}, but refer to \lock{i} as \lock{} whenever the context of the discussion is not concerned with the value of $i$.
Method \lock{} returns a non-$\bot$ value if a process captures the lock, otherwise it returns a $\bot$ value to indicate a failed \lock{} call.
A \lock{} by process $p$ can fail only if $p$ aborts during the method call.
Method \release{}{} takes two arguments, a pseudo-ID $i \in \Set{0,\ldots,n-1}$ and an integer $j$.
Method \release{i}{j} returns \True\ if and only if there exists a concurrent call to \lock{} that eventually returns $j$.
Otherwise method \release{i}{j} returns \False.
The information contained in argument $j$ determines the transfered node locks.
Process pseudo-IDs are passed as arguments to the methods to allow the ability for a process to ``transfer'' the responsibility of releasing the lock to another process.
Specifically, we desire that if a process $p$ executes a successful \lock{i} call and becomes the owner of the lock, then $p$ does not have to release the lock itself, if it can find some process $q$ to call \release{i}{} on its behalf.
In Section \ref{sec:ARLockArray} we implement object \ARMLockArray{n}, and prove its properties in Appendix~\ref{sec:Analysis}, and thus we get the following lemma.

\begin{lemma}
\label{theorem:ARMLockArray}
Object \ARMLockArray{n} can be implemented against the weak adversary for the CC model with the following properties using only \Order{n} \CAS\ objects and read-write registers.
\begin{enumerate}[(a)]
 \item Mutual exclusion, starvation freedom, bounded exit, and bounded abort.
 \item The abort-way has \Order{n} RMR complexity.
 \item If a process does not abort during a \lock{} call, then it incurs \Order{1} RMRs in expectation during the call, otherwise it incurs \Order{n} RMRs in expectation during the call.
 \item If a process' call to \release{}{j} returns \False, then it incurs \Order{1} RMRs during the call, otherwise it incurs \Order{n} RMRs during the call.
\end{enumerate}
\end{lemma}

\paragraph{High Level Description of the Abortable Lock.}
We use a complete $\Delta$-ary tree \tree\ of height $\Delta$ with $N$ leaves, called the arbitration tree.
The root has height $\Delta$ and the leaves of the tree have height 0.
The $N$ processes in the system line up as $N$ unique leaf nodes, such that each process $p$ is associated with a unique leaf $\myleaf{p}$ in the tree.
Let \mypath{p} denote the path from \myleaf{p} up to \root, and \h{u} denote the height of node $u$.

Each node of our arbitration tree \tree\ is a structure of type \Node\ that contains a single instance \L\ of the abortable randomized lock object $\ARMLockArray{\Delta}$.
This allows processes the ability to abort their attempt at any point in time during their ascent to the root node.

\bparagraph{Lock capture protocol - \lock{p}}
During \lock{p} a process $p$ attempts to \textit{capture} every node on its path \mypath{p} that it does not own, as long as $p$ has not received a signal to abort.
Process $p$ attempts to capture a node $u$ by executing a call to $u$.\L.\lock{}.
If $p$'s $u$.\L.\lock{} call returns $\infty$ then $p$ is said to have captured $u$, and if the call returns an integer $j$, then $p$ is said to have been \emph{handed over} all nodes from $u$ to $v$ on \mypath{p}, where $\h{v} = j$.
We ensure that $j \geq \h{u}$.
Process $p$ starts to \textit{own} node $u$ when $p$ captures $u$.\L\ or when $p$ is handed over node $u$ from the previous owner of node $u$.
Process $p$ can enter its Critical Section when it owns the root node of \tree.
Process $p$ may receive a signal to abort during a call to $u$.\L.\lock{} as a result of which $p$'s call to $u$.\L.\lock{} returns either $\bot$ or a non-$\bot$ value.
In either case, $p$ then calls \release{p}{} to release all locks of nodes that $p$ has captured in its passage, and then returns from its \lock{p} call with value $\bot$.

\bparagraph{Lock release protocol - \release{p}{}}
An exiting process $p$ \emph{releases} all nodes that it owns during \release{p}{}.
Process $p$ is said to \textit{release} node $u$ if $p$ releases $u$.\L\ (by executing $u$.\L.\release{}{} call), or if $p$ hands over node $u$ to some other process.
Recall that $p$ hands over node $u$ if $p$ executes a $v$.\L.\release{}{j} call that returns \True\ where $\h{v} \leq \h{u} \leq j$.
Let $s$ be the height of the highest node $p$ owns.
During \release{p}{}, $p$ climbs up \tree\ and calls $u$.\L.\release{p}{s} at every node $u$ that it owns, until a call returns \True.
If a $u$.\L.\release{p}{s} call returns \False (process $p$ incurs \Order{1} steps), then $p$ is said to have released lock $u$.\L\ (and therefore released node $u$), and thus $p$ continues on its path.
If a $u$.\L.\release{p}{s} call returns \True (process $p$ incurs \Order{\Delta} steps), then $p$ has handed over all remaining nodes that it owns to some process that is executing a concurrent $u$.\L.\lock{} call at node $u$, and thus $p$ does not release any more nodes.

Notice that our strategy to release node locks is to climb up the tree until all node locks are released or a hand over of remaining locks is made.
Climbing up the tree is necessary (as opposed to climbing down) in order to hand over node locks to a process, say $q$, such that the handed over nodes lie on $\mypath{q}$.


%
\section{The Array Based Abortable Lock}
\label{sec:ARLockArray}


We specified object \ARMLockArray{n} in Section \ref{sec:ARLockTree} and now we describe and implement it (see Figures~\ref{fig:ARMEAlgorithm1} and \ref{fig:ARMEAlgorithm2}).
Let $\L$ be an instance of object \ARMLockArray{n}.
{\renewcommand{\baselinestretch}{0.75}
\begin{classfigure}[!htbp]
\begin{object}{\ARMLockArray{n}} \label{algo:ARMLockStructure}
  \shared \;
    \qquad $\ctr$: \RCASCounter{2} \Init $0$;\;
    \qquad $\PawnSet$: Object of type $\APArray{n}$ \Init $\varnothing$;\;
    \qquad $\Apply$: \Array $[0 \dots n-1]$ \Of \Int\ pairs \Init all $\pair{\bot}{\bot}$;\;
     \qquad $\Role$: \Array $[0 \dots n-1]$ \Of \Int \Init $\bot$;\;
    \qquad $\X,\LSync$: \Int \Init $\bot$;\;
    \qquad \cKing, \cQueen, \cPawn, \cPPawn, \cWant, \cOk: \const \Int\ $0,1,2,3,4,5$ respectively;\;
    \qquad \getSequenceNo{}: returns integer $k$ on being called for the $k$-th time from a call to \;
    \qquad \lock{i}. (Since calls to \lock{i} are executed sequentially, a sequential shared counter \;
    \qquad suffices to implement method \getSequenceNo{}.)\;
  \local \;
    \qquad $s,val,seq,dummy$: \Int \Init $\bot$;\;
    \qquad $flag,r$: \Bool \Init $\False$;\;
    \qquad $A$: \Array $[0 \dots n-1]$ \Of \Int \Init $\bot$ \;
\tcp{If process $i$ satisfies the loop condition in line~\ref{getLock:ApplyBotWant},~\ref{getLock:awaitAckOrCtrDecrease}, or~\ref{getLock:awaitX}, and $i$ has received a signal to abort, then $i$ calls \abort{i}}
\end{object}
\begin{minipage}{0.56\textwidth}
\LinesNumbered
\begin{algorithm}[H]\Method{lock$_{i}$( )} \label{function:getLock}
\setcounter{AlgoLine}{0}
  $s$ \la\ \getSequenceNo{} \; \label{getLock:getSequenceNo}
  \Await {\Apply$[i]$.\CAS{$\pair{\bot}{\bot},\pair{\cWant}{s}$}} \;  \label{getLock:ApplyBotWant}
 $flag$ \la\ \True \;  \label{getLock:setFlag}
      \Repeat {$(\Role[i] \in \Set{\cKing, \cQueen, \cPPawn})$ } {                  \label{getLock:BeginInnerLoop}
          $\Role[i]$ \la\  $\ctr$.\inc{}       \;              \label{getLock:IncCounter}
          \If {$(\Role[i] = \cPawn)$}  {                          \label{getLock:ifSoldier}
              \Await {$\Apply[i] = \pair{\cOk}{s} \vee \ctr.\Read{} \neq 2$} \;  \label{getLock:awaitAckOrCtrDecrease}
              \If {$(\Apply[i] = \pair{\cOk}{s})$} {                  \label{getLock:ifBackpacked}
                  $\Role[i]$ \la\ \cPPawn \;                   \label{getLock:RolePPawn}
              }
          }
      }                                                     \label{getLock:EndInnerLoop}
  \If {$(\Role[i] = \cQueen)$}  {                            \label{getLock:ifQueen}
        \Await{$\X \neq \bot$}  \;                            \label{getLock:awaitX}
  }
  \Apply[i].\CAS{$\pair{\cWant}{s},\pair{\cOk}{s}$} \;                          \label{getLock:ApplyWantOk}
   \lIf {$\Role[i]=\cQueen$} {                              \label{getLock:ifRolePQueen}
       \return \X                                           \label{getLock:returnX}
   }
   \lElse {
       \return $\infty$ \;                                  \label{getLock:returninfty}
   }                                                        \label{getLock:end}
\end{algorithm}
\end{minipage}
\begin{minipage}{0.44\textwidth}
\LinesNumbered
\begin{algorithm}[H]\Method{abort$_{i}$( )} \label{function:abort}
  \lIf{$\neg flag$} {                     \label{abort:ifNotRegistered} \label{abort:ifFlag}
      \return $\bot$ \;                            \label{abort:returnbot}
  }
  \Apply[i].\CAS{$\pair{\cWant}{s},\pair{\cOk}{s}$} \;              \label{abort:ApplyWantOk}
  \eIf {$\Role[i] = \cPawn$} {                 \label{abort:ifPawn}
      \If {$\neg \PawnSet$.\cUpdate{$i,s$}} {  \label{abort:ifHead}
        $\Role[i]$ \la\ \cPPawn \;              \label{abort:RolePPawn}
        \return $\infty$ \;              \label{abort:returninfty}
      }
  } {
  \If {$\neg$\X.\CAS{$\bot,\infty$}} {                     \label{abort:setX}
        \return \X \;                     \label{abort:returnX}
  }
      \doCollect{i} \;                               \label{abort:doCollect}
      \helpRelease{i} \;                            \label{abort:callhRelease}
  }
  \Apply$[i]$.\CAS{$\pair{\cOk}{s},\pair{\bot}{\bot}$} \;    \label{abort:ApplyOkBot}
  \return $\bot$ \;                                \label{abort:returnr}
\end{algorithm}
\end{minipage}
\begin{minipage}{\textwidth}
\LinesNumbered
\begin{algorithm}[H]\Method{doCollect$_{i}$()} \label{function:doCollect}
\setcounter{AlgoLine}{50}
      \For {$k \la 0$ \KwTo\ $n-1$} {           \label{collect:collectLoop}
            $\pair{val}{seq}$ \la\ \Apply$[k]$ \; \label{collect:ApplyRead}
            \lIf {$val = \cWant$} {          \label{collect:IfRegistered}
                $A[k]$ \la\ $seq$
            }
            \lElse {
                $A[k]$ \la\ $\bot$ \;
            }
      }
    \PawnSet.\collect{$A$} \;      \label{collect:updateAll}
\end{algorithm}
\end{minipage}
\caption{Implementation of Object \ARMLockArray{n}}\label{fig:ARMEAlgorithm1}
\end{classfigure}

\begin{classfigure}[!htbp]

\begin{minipage}{0.5\textwidth}
\LinesNumbered
\setcounter{AlgoLine}{33}
\begin{algorithm}[H]\Method{release$_{i}$(int $j$)}  \label{function:release}
  $r$ \la\ \False \;                        \label{release:setrNotCollected} \label{release:safetyCheck}
      \If {$\Role[i] = \cKing$} {                       \label{release:ifKing}
          \If {$\neg$\ctr.\CAS{$1,0$}} {               \label{release:ctr10}
            $r$ \la\ \X.\CAS{$\bot,j$} \;              \label{release:setX}
            \lIf {$r$} {
                \doCollect{i} \;                         \label{release:doCollect}
            }
              \helpRelease{i} \;                            \label{release:callhRelease:King}
          }                                          \label{release:endIfFlag}
      }
      \If {$\Role[i] = \cQueen$}  {            \label{release:ifQueen}
        \helpRelease{i} \;                                  \label{release:callhRelease:Queen}
      }
      \If {$\Role[i] = \cPPawn$} {      \label{release:ifPPawn}
             \doPromote{i} \;           \label{release:callPromote}
      }
  \pair{dummy}{s} \la\ \Apply$[i]$ \;   \label{release:ApplyRead}
  \Apply$[i]$.\CAS{$\pair{\cOk}{s},\pair{\bot}{\bot}$} \;    \label{release:ApplyOkBot}
  \return $r$ \;                        \label{release:return}
\end{algorithm}
\end{minipage}
\begin{minipage}{0.5\textwidth}
\LinesNumbered
\begin{algorithm}[H]\Method{helpRelease$_{i}$()} \label{function:hRelease}
\setcounter{AlgoLine}{55}
   \DontPrintSemicolon
   \If {$\neg$\LSync.\CAS{$\bot,i$}} {                    \label{hRelease:setT}
       $j$ \la\ \X.\Read{} \;                            \label{hRelease:readX}
       \X.\CAS{$j,\bot$} \;          \label{hRelease:resetX}
       $j$ \la\ \LSync.\Read{} \;                    \label{hRelease:readT}
       \LSync.\CAS{$j,\bot$} \;                            \label{hRelease:resetT}
       \PawnSet.\remove{$j$} \;                \label{hRelease:collectFixOther}
       \doPromote{i}                       \label{hRelease:callPromote}
   }                                     \label{hRelease:end}
\end{algorithm}
\LinesNumbered
\begin{algorithm}[H]\Method{doPromote$_{i}$()}  \label{function:promote}
  \DontPrintSemicolon
    \PawnSet.\remove{$i$} \;                \label{promote:collectFixSelf}
    \pair{j}{seq} \la\ \PawnSet.\promote{} \;  \label{promote:FR12}
    \eIf {$j = \bot$} {                         \label{promote:ifNotPromoted}
        \PawnSet.\reset{} \;              \label{promote:resetBackpack}
        \ctr.\CAS{$2,0$} \;                   \label{promote:ctr20}
    } {
        \Apply[$j$].\CAS{$\pair{\cWant}{seq},\pair{\cOk}{seq}$}\;  \label{promote:ApplyWantOk}
    }                                           \label{promote:end}
\end{algorithm}
\end{minipage}

\caption{Implementation of Object \ARMLockArray{n} (continued)}\label{fig:ARMEAlgorithm2}
\end{classfigure}
\renewcommand{\baselinestretch}{1.0}}

\bparagraph{Registering and Roles at lock \L}
At the beginning of a \lock{} call processes \emph{register} themselves in the \Apply\ array by swapping the value $\cWant$ atomically into their designated slots (\Apply[i] for process with pseudo-ID $i$) using a \CAS\ operation.
The array {\Apply} of $n$ \CAS\ objects is used by processes to register and ``deregister'' themselves from lock \L, and to notify each other of certain events at lock \L.
On registering in the \Apply\ array, processes attempt to increase \ctr, an instance of \RCASCounter{2}, using operation \ctr.\inc{}.
Recall that $\RCASCounter{2}$ is a bounded counter, initially 0, and returns values in \Set{0,1,2} (see Section~\ref{sec:RCASCounter}).
Each of these values corresponds to a role at lock \L.
There are four \emph{roles} that a process can assume during its passage of lock \L, namely \emph{king}, \emph{queen}, \emph{pawn} and \emph{promoted pawn}, and a role defines the protocol a process follows during a passage.
During an execution, \ctr\ cycles from its initial value $0$ to non-$0$ values and then back to $0$, multiple times, and we refer to each such cycle as a \ctr-cycle.
The process that increases \ctr\ from $0$ to $1$ becomes the king.
The process that increases \ctr\ from $1$ to $2$ becomes the queen.
All processes that attempt to increase \ctr\ any further, are returned value $2$ (by specification of object \RCASCounter{2}), and they assume the role of a pawn process.
A pawn process busy-waits until it gets ``promoted'' at lock \L\ (a process is said to be \emph{promoted} at lock \L\ if it is promoted in \PawnSet), or until it sees the \ctr\ value decrease, so that it can attempt to increase \ctr\ again.
We ensure that a pawn process repeats an attempt to increase \ctr\ at most once, before getting promoted.
We ensure that at any point in time during the execution, the number of processes that have assumed the role of a king, queen and promoted pawn at lock \L, respectively, is at most one, and thus we refer to them as \king{\L}, \queen{\L} and \ppawn{\L}, respectively.
We describe the protocol associated with each of the roles in more detail shortly.
An array {\Role} of $n$ read-write registers is used by processes to record their role at lock \L.

\bparagraph{Busy-waiting in lock \L}
The king process, \king{\L}, becomes the first owner of lock \L\ during the current \ctr-cycle, and can proceed to enter its Critical Section, and thus it does not busy-wait during \lock{}.
The queen process, \queen{\L}, must wait for \king{\L}  for a notification of its turn to own lock \L. 
Then \queen{\L} spins on \CAS\ object \X, waiting for \king{\L} to \CAS\ some integer value into \X.
Process \king{\L} attempts to \CAS\ an integer $j$ into \X\ only during its call to \release{}{j}, after it has executed its Critical Section.
The pawn processes wait on their individual slots of the \Apply\ array for a notification of their promotion.

\bparagraph{A \emph{collect} action at lock \L}
A collect action is conducted by either \king{\L} during a call to \release{}{}, or by \queen{\L} during a call to \abort{}.
A collect action is defined as the sequence of steps executed by a process during a call to \doCollect{}.
During a call to \doCollect{}, the collecting process (say $q$) iterates over the array \Apply\ reading every slot, and then creates a local array $A$ from the values read and stores the contents of $A$ in the \PawnSet\ object in using the operation \PawnSet.\collect{A}.
A key point to note is that operation \PawnSet.\collect{A} does not overwrite an aborted process's value in \PawnSet\ (a process aborts itself in \PawnSet\ by executing a successful \PawnSet.\abort{} operation).

\bparagraph{A \emph{promote} action at lock \L}
Operation \PawnSet.\promote{} during a call to method \doPromote{} is defined as a promote action.
The operation returns the pseudo-ID of a process that was collected during a collect action, and has not yet aborted from \PawnSet.
A promote action is conducted at lock \L\ either by \king{\L}, \queen{\L} or \ppawn{\L}.

\bparagraph{Lock \emph{handover} from \king{\L} to \queen{\L}}
As mentioned, process \queen{\L} waits for \king{\L} to finish its Critical Section and then call \release{}{j}.
During \king{\L}'s \release{}{j} call, \king{\L} attempts to swap integer $j$ into \CAS object \X, that only \king{\L} and \queen{\L} access.
If \queen{\L} has not ``aborted'', then \king{\L} successfully swaps $j$ into \X, and this serves as a notification to \queen{\L} that \king{\L} has completed its Critical Section, and that \queen{\L} may now proceed to enter its Critical Section.

\bparagraph{\emph{Aborting} an attempt at lock \L\ by \queen{\L}}
On receiving  a signal to abort, \queen{\L} abandons its \lock{} call and executes a call to \abort{} instead.
\queen{\L} first changes the value of its slot in the \Apply\ array from \cWant\ to \cOk, to prevent itself from getting collected in future collects.
Since \king{\L} and \queen{\L} are the first two processes at \L, \king{\L} will eventually try to handover \L\ to \queen{\L}.
To prevent \king{\L} from handing over lock \L\ to \queen{\L}, \queen{\L} attempts to swap a special value $\infty$ into \X\ in one atomic step.
If \queen{\L} fails then this implies that \king{\L} has already handed over \L\ to \queen{\L}, and thus \queen{\L} returns from its call to \abort{} with the value written to \X\ by \king{\L}, and becomes the owner of \L.
If \queen{\L} succeeds then \queen{\L} is said to have successfully aborted, and thus \king{\L} will eventually fail to hand over lock \L.
Since \queen{\L} has aborted, \queen{\L} now takes on the responsibility of collecting all registered processes in lock \L, and storing them into the \PawnSet\ object.
After performing a collect, \queen{\L} then synchronizes with \king{\L} again, to perform a promote, where one of the collected processes is promoted.
After that, \queen{\L} deregisters from the \Apply\ array by resetting its slot to the initial value \pair{\bot}{\bot}.

\bparagraph{\emph{Aborting} an attempt at lock \L\ by a pawn process}
On receiving a signal to abort a pawn process (say $p$) busy-waiting in lock \L, abandons its \lock{} call and executes a call to \abort{} instead.
Process $p$ first changes the value of its slot in the \Apply\ array from \cWant\ to \cOk, to prevent itself from getting collected in future collects.
It then attempts to abort itself in \PawnSet\ by executing the operation \PawnSet.\cUpdate{$p$}).
If $p$'s attempt is unsuccessful then it implies that $p$ has already been promoted in \PawnSet, and thus $p$ can assume the role of a promoted pawn, and become the owner of \L.
In this case, $p$ returns from its \abort{} call with value $\infty$ and becomes the owner of \L.
If $p$'s attempt is successful then $p$ cannot be collected or promoted in future collects and promotion events.
In this case, $p$ deregisters from the \Apply\ array by resetting its slot to the initial value \pair{\bot}{\bot}, and returns $\bot$ from its call to \abort{}.

\bparagraph{Releasing lock \L}
Releasing lock \L\ can be thought of as a group effort between the \king{\L}, \queen{\L} (if present at all), and the promoted pawns (if present at all).
To completely release lock \L, the owner of \L\ needs to reset \ctr\ back to $0$ for the next \ctr-cycle to begin.
However, the owner also has an obligation to hand over lock \L\ to the next process waiting in line for lock \L.
We now discuss the individual strategies of releasing lock \L, by \king{\L}, \queen{\L} and the promoted processes.
To release lock \L, the owner of \L\ executes a call to \release{}{j}, for some integer $j$.

\bparagraph{Synchronizing the release of lock \L\ by \king{\L} and \queen{\L}}
Process \king{\L} first attempts to decrease \ctr\ from $1$ to $0$ using a \CAS\ operation.
If it is successful, then \king{\L} was able to end the \ctr-cycle before any process could increase \ctr\ from $1$ to $2$.
Thus, there was no \queen{\L} process or pawn processes waiting for their turn to own lock \L, during that \ctr-cycle.
Then \king{\L} is said to have released lock \L.

If \king{\L}'s attempt to decrease \ctr\ from $1$ to $0$ fails, then \king{\L} knows that there exists a \queen{\L} process that increased \ctr\ from $1$ to $2$.
Since \queen{\L} is allowed to abort, releasing lock \L\ is not as straight forward as raising a flag to be read by \queen{\L}.
Therefore, \king{\L} attempts to synchronize with \queen{\L} by swapping the integer $j$ into the object \X\ using a \X.\CAS{$\bot,j$} operation.
Recall that \queen{\L} also attempts to swap a special value $\infty$ into object \X\ using a \X.\CAS{$\bot,j$} operation, in order to abort its attempt.
Clearly only one of them can succeed.
If \king{\L} succeeds, then \king{\L} is said to have successfully handed over lock \L\ to \queen{\L}.
If \king{\L} fails, then \king{\L} knows that \queen{\L} has aborted and thus \king{\L} then tries to hand over its lock to one of the waiting pawn processes.
The procedure to hand over lock \L\ to one of the waiting pawn processes is to execute a collect action followed by a promote action.

The collect action needs to be executed only once during a \ctr-cycle, and thus we let the process (among \king{\L} or \queen{\L}) that successfully swaps a value into \X, execute the collect action.

If \king{\L} successfully handed over \L\ to \queen{\L}, it collects the waiting pawn processes, so that eventually when \queen{\L} is ready to release lock \L, \queen{\L}  can simply execute a promote action.
Since there is no guarantee that \king{\L} will finish collecting before \queen{\L} desires to execute a promote action, the processes synchronize among themselves again, to execute the first promote action of the current \ctr-cycle.
They both attempt to swap their pseudo-IDs into an empty \CAS object \LSync, and therefore only one can succeed.
The process that is unsuccessful, is the second among them, and therefore by that point the collection  of the waiting pawn process must be complete.
Then the process that is unsuccessful, resets \X\ and \LSync\ to their initial value $\bot$, and then executes the promote action, where a waiting pawn process is promoted and handed over lock \L.
If no process were collected during the \ctr-cycle, or all collected pawn processes have successfully aborted before the promote action, then the promote action fails, and thus the owner process resets the \PawnSet\ object, and then resets \ctr\ from $2$ to $0$ in one atomic step, thus releasing lock \L, and resetting the \ctr-cycle.

\bparagraph{The release of lock \L\ by \ppawn{\L}}
If a process was promoted by \king{\L} or \queen{\L} as described above, then the promoted process is said to be handed over the ownership of \L, and becomes the first promoted pawn of the \ctr-cycle.
Since a collect for this \ctr-cycle has already been executed, process \ppawn{\L} does not execute any more collects, but simply attempts to hand over lock \L\ to the next collected process by executing a promote action.
This sort of promotion and handing over of lock \L\ continues until there are no more collected processes to promote, at which point the last promoted pawn resets the \PawnSet\ object, and then resets \ctr\ from $2$ to $0$ in one atomic step, thus releasing lock \L, and resetting the \ctr-cycle.

All owner processes also \emph{deregister} themselves from lock \L, by resetting their slot in the \Apply\ array to the initial value \pair{\bot}{\bot}.
This step is the last step of their \release{}{j} calls, and processes return a boolean to indicate whether they successfully wrote integer $j$ into \X\ during their \release{}{j} call.
Note that only \king{\L} could possibly return \True\ since it is the only process that attempts to do so, during its \release{}{j} call.

%
%

%
\section{Conclusion}
We presented the first randomized abortable lock that achieves sub-logartihmic expected RMR complexity.
While the speed-up is only a modest $O(\log\log n)$ factor over the most efficient deterministic abortable mutual exclusion algorithm, our result shows that randomization can help in principle, to improve the efficiency of abortable locks.
Unfortunately, our algorithm is quite complicated; it would be nice to find a simpler one.
It would also be interesting to find an algorithm with sub-logarithmic RMR complexity that works against the stronger adversary.
In the weak adversary model, no non-trivial lower bounds for mutual exclusion are known, but it seems hard to improve upon $O(\log n/\log\log n)$ RMR complexity, even without the abortability property.

As shown by Bender and Gilbert, \cite{BenderGilbert:ME}, the picture looks different in the oblivious adversary model.
However, their algorithm is only lock-free with high probability.
It would be interesting to find a mutual exclusion algorithm with $o(\log n/\log\log n)$ RMR complexity against the oblivious adversary that is lock-free with probability one.
It would also be interesting to know whether such an algorithm can be made abortable.

\subsubsection*{Acknowledgement.}
We are indebted to Lisa Higham and Bill Sands for their careful reading of an earlier version of the paper and their valuable comments.
We also thank the anonymous referees of DISC 2012 for their helpful comments.

\bibliographystyle{plain}
\bibliography{thesis-bibitems}

\pagebreak
\appendix

\section*{Appendix}
\section{Implementation of Object \RCASCounter{k}}
\label{sec:appendix:RCASCounter}

The sequential specification of the CAS Counter object is presented in Figure~\ref{fig:counterSpecfication} in the form of type \CASCounter{k}.
The implementation of our randomized CAS counter object,  \RCASCounter{k} of type \CASCounter{k} is presented in Figure~\ref{fig:counterImpleemntation}.
A shared \CAS\ object \Count\ is used to store the value of the counter object, and is initialized to $0$.
The object provides methods \inc{}, \CAS{} and \Read{}, where the \inc{} method is allowed to \emph{fail}, in which case the operation does not change the object state, and returns $\bot$ to indicate the failure.

\begin{classfigure}[!htbp]
\setcounter{AlgoLine}{0}
  \begin{type}{\CASCounter{k}}
 $x$:  \Int\ \Init $0$ \;
  \end{type}
\begin{minipage}{0.29\textwidth}
\begin{operation}{inc()}
  \setcounter{AlgoLine}{0}
  \lIf {$x = k$} \label{atomic:inc:ifXequalsK} {
      \return $x$\;
  }
  $x$ \la\ $x + 1$ \label{atomic:inc:increment} \;
  \return $x-1$ \label{atomic:inc:return} \;
\end{operation}
\end{minipage}
\begin{minipage}{0.53\textwidth}
\begin{operation}{CAS($old,new$)}
  \lIf {$x \neq old \vee new \notin \Set{0,\ldots,k}$}  { \label{atomic:CAS:ifCondition}
    \return \False \label{atomic:CAS:returnFalse}\;
  }
  $x$ \la\ $new$ \;
  \return \True \;
\end{operation}
\end{minipage}
\begin{minipage}{0.16\textwidth}
\begin{operation}{Read()}
  \return $x$ \;
\end{operation}
\end{minipage}
\caption{Sequential Specification of Type \CASCounter{k}} \label{fig:counterSpecfication}
\end{classfigure}

\begin{classfigure}[!htbp]
  \begin{object}{\RCASCounter{k}} \label{Algo:CASCounter}
  \shared
    \qquad $\Count$: \Int\ \Init $0$ \;
  \local
    \qquad $\beta$: \Int\ \Init $0$ \;
  \end{object}
\begin{minipage}{0.525\textwidth}
\LinesNumbered
\begin{algorithm}[H]\Method{inc( )}
    $\beta$ \la\ \random{$0,1,\ldots,k$}    \label{RCAS:inc:randomChoice}\;
    \eIf {$(\beta = k)$} {\label{RCAS:inc:ifBetaEqualsK}
      \lIf {$(\Count.\Read{} = k)$} \label{RCAS:inc:Read}{
            \return $k$ \label{RCAS:inc:returnRead} \;
      }
    } {
      \lIf {\Count.\CAS{$\beta,\beta+1$}} \label{RCAS:inc:CAS} {
          \return $\beta$ \; \label{RCAS:inc:returnCAS}
      }
    }
     \return $\bot$ \; \label{RCAS:inc:returnFail}
\end{algorithm}
\end{minipage}
\begin{minipage}{0.47\textwidth}
\LinesNumbered
\begin{algorithm}[H]\Method{CAS($old,new$)}
    \lIf {$new \notin \Set{0,\ldots,k}$} {
      \return \False \label{RCAS:CAS:safety} \;
    }
    \return $\Count.\CAS{$old,new$}$  \label{RCAS:CAS:CAS} \;
\end{algorithm}
\LinesNumbered
\begin{algorithm}[H]\Method{Read( )}
    \return $\Count.\Read{}$  \label{RCAS:Read:Read} \;
\end{algorithm}
\end{minipage}
\caption{Implementation of Object \RCASCounter{k} }\label{fig:counterImpleemntation}
\end{classfigure}

During the \inc{} method, a process $p$ first makes a guess at the counter's current value by rolling a $(k+1)$-sided dice (in line~\ref{RCAS:inc:randomChoice}) that returns a value in \Set{0,\ldots,k} uniformly at random, and stores the value in local variable $\beta$.
If $\beta =k$, then $p$ performs a \Read{} on \Count (in line~\ref{RCAS:inc:Read}) to verify the correctness of its guess.
If $p$'s guess is correct, then it returns $k$, otherwise it returns $\bot$ (in line~\ref{RCAS:inc:returnFail}) to indicate a failed \inc{} method call.
If $\beta \in \Set{0,\ldots,k-1}$, then $p$ performs a \Count.\CAS{$\beta,\beta+1$} operation (in line~\ref{RCAS:inc:Read}) in order to verify the correctness of its guess and to increment \Count\ in one atomic step.
If $p$'s guess is correct, then the \CAS\ operation succeeds and the \inc{} method returns the previous value.
Otherwise the \inc{} method returns $\bot$ (in line~\ref{RCAS:inc:returnFail}) to indicate a failed \inc{} method call.

Method \Read{} simply reads the current value of \Count\ using a \Count.\Read{} operation (line~\ref{RCAS:Read:Read}) and returns the result of the operation.
Method \CAS{} takes two integer parameters $old,new$, and in \Line{RCAS:CAS:safety} performs a safety check, where it checks whether the value of $new$ is in \Set{0,\ldots,k}.
If the safety check fails, then the method simply returns \False.
Otherwise, it attempts to change the value of \Count\ from $old$ to $new$ using the \Count.\CAS{$old,new$} operation  (in \Line{RCAS:CAS:CAS}) and returns the result of the operation.

\subsection{Analysis and Properties of Object \RCASCounter{k}\ }
Consider an instance of the \RCASCounter{k} object.
Let $H$ be an arbitrary history that consists of all method calls on the instance, except failed \inc{} calls and pending calls that are yet to execute \Line{RCAS:inc:Read} (\Read operation), \Line{RCAS:inc:CAS} (\CAS operation), \Line{RCAS:CAS:CAS} (\CAS operation) or \Line{RCAS:Read:Read} (\Read operation).
If a failed \inc{} is in the history, it can be linearized at an arbitrary point between its invocation and response, as it does not affect the validity of any other operations.
Therefore, it suffices to prove that the history without failed \inc{} operations is linearizable, and then linearizability of the original history follows.
The same argument applies to omitting the selected pending method calls.
Since the selected pending method calls do not change any shared object, they cannot affect the validity of any other operations.

We define a point $pt(u)$ for every method $u$ in $H$.
Let $I(u)$ be the interval between $u$'s invocation and response.
Let $S$ be the sequential history obtained by ordering the method calls in $H$ according to the points $pt(u)$.
To show that \RCASCounter{k} is a randomized linearizable implementation of the  type \CASCounter{k}, we need to show that the sequential history $S$ is valid, i.e., $S$ lies in the specification of type \CASCounter{k} object, and that $pt(u)$ lies in $I(u)$.
Let \sC\ be an object of type \CASCounter{k}, and let $S_v$ be the sequential history obtained when the operations of $S$ are executed sequentially on object \sC\ in the order as given in $S$.
Clearly, $S_v$ is a valid sequential history in the specification of type \CASCounter{k} by construction.
Then to show that $S$ is valid, we show that $S = S_v$.

\begin{lemma} \label{claim:ctr:linearizable}
Object \RCASCounter{k} is a randomized linearizable implementation of type \CASCounter{k}.
\end{lemma}

\begin{proof}
Let $A$ be an instance of the \RCASCounter{k} object.
Consider an arbitrary history $H$ that consists of all completed method calls on $A$, except failed \inc{} calls, and all pending method calls on $A$ that have executed a successful \CAS\ operation.
We now define point $pt(u)$ for every method $u$ in $H$.

If $u$ is a \Read{} method call then define $pt(u)$ to be the point in time when the \Read\ operation in line~\ref{RCAS:Read:Read} is executed.

If $u$ is an \inc{} method call that returns from \Line{RCAS:inc:returnRead} then $pt(u)$ is the point in time of the \Read operation in \Line{RCAS:inc:Read}, and if $u$'s \CAS\ operation in \Line{RCAS:inc:CAS} succeeds then $pt(u)$ is the point in time of the \CAS operation in \Line{RCAS:inc:CAS}.
By construction, a \Read\ or \CAS\ operation has been executed during every \inc{} call in $H$, and no failed \inc{} calls are in $H$.
Then it follows that we have defined $pt(u)$ for every \inc{} call $u$ in $H$.

If $u$ is a \CAS{} method call that returns from \Line{RCAS:CAS:safety} then $pt(u)$ is any arbitrary point during $I(u)$, and if $u$ returns from \Line{RCAS:CAS:CAS} then $pt(u)$ is the point in time of the \CAS operation in \Line{RCAS:CAS:CAS}.

Clearly $pt(u) \in I(u)$ for every method $u$ in $H$.

Let $u_i$ be the $i$-th operation in $S$ and $v_i$ be the $i$-th operation in $S_v$.
Let \valcRegA{u_i} denote the value of object \Count\ immediately after $pt(u_i)$, and let \valCtrA{v_i} denote the value of object \sC\ after operation $v_i$ in $S_v$.
We assume that $u_0$ is a method call that does not change the state of any shared object of instance $A$ (such as a \Read{} method) and returns the initial value of the object.
This assumption can be made without loss of generality, because the removal of a method call that does not change the state of the object from a linearizable history always leaves a history that is also linearizable.
The purpose of the assumption is to simplify the base case of our induction hypothesis.

We now prove by induction on integer $i$, that $\valcRegA{u_i} = \valCtrA{v_i}$, and that the return value of $u_i$ matches the value returned by $v_i$, thereby proving $S = S_v$.

\textbf{Basis $(i=0)$} Since initially the value of object \Count\ and the value of the atomic \CASCounter{k} object is $0$, it follows from the definition of the method call $u_0$, that $\valcRegA{u_0} = \valCtrA{v_0} = 0$, and the return value of $u_0$ matches that of $v_0$.

\textbf{Induction Step $(i>0)$}
From the induction hypothesis, $\valcRegA{u_{i-1}} = \valCtrA{v_{i-1}}$.

\textbf{Case a - } $u_i$ is an \inc{} method call that executes a successful \CAS{} operation in \Line{RCAS:inc:CAS}.
Then $pt(u_i)$ is when object \Count\ is incremented from $\beta$ to $\beta+1$ by a successful \Count.\CAS{$\beta,\beta+1$} operation in \Line{RCAS:inc:CAS}, and thus $\valcRegA{u_{i-1}} = \beta$ holds.
Also, $u_i$ returns $\beta = \valcRegA{u_{i-1}}$.
Since $u_i$ fails the if-condition of line~\ref{RCAS:inc:ifBetaEqualsK}, $\beta \neq k$ and therefore $\valcRegA{u_{i-1}} = \beta \neq k$ holds.
Now consider operation $v_i$ in $S_v$.
Since $\valCtrA{v_{i-1}} = \valcRegA{u_{i-1}} \neq k$, the if-condition of line~\ref{atomic:inc:ifXequalsK} fails, and the value of the atomic $\CASCounter{k}$ is incremented in line~\ref{atomic:inc:increment} and \valCtrA{v_{i-1}} returned in line~\ref{atomic:inc:return}.
Hence $\valcRegA{u_i} = \valCtrA{v_i}$ and the return values match.

\textbf{Case b - }
$u_i$ is an \inc{} method call that returns from \Line{RCAS:inc:returnRead}.
Then $pt(u_i)$ is when the  \Read{} operation on the object \Count\ is executed in \Line{RCAS:inc:Read}.
Clearly, the value returned by the \Read{} operation on the object \Count\ at $pt(u_i)$ is \valcRegA{u_{i-1}}.
Since the if-condition of line~\ref{RCAS:inc:Read} is satisfied, $\valcRegA{u_{i-1}} =k$ and $u_i$ returns integer $k$ without changing object \Count.
Now consider operation $v_i$ in $S_v$.
Since $\valCtrA{v_{i-1}} = \valcRegA{u_{i-1}}$ and $\valcRegA{u_{i-1}} =k$, the if-condition of line~\ref{atomic:inc:ifXequalsK} is satisfied and integer $k$ is returned without changing the atomic \CASCounter{k} object.
Hence $\valcRegA{u_i} = \valCtrA{v_i}$ and the return values match.

\textbf{Case c - }
$u_i$ is a \CAS{} method call that returns from \Line{RCAS:CAS:safety}.
Then the if-condition of line~\ref{RCAS:CAS:safety} is satisfied and thus $new \notin \Set{0,1,\ldots,k}$ and $u_i$ returns \False\ without changing \Count.
Now consider operation $v_i$ in $S_v$.
Since $new \notin \Set{0,1,\ldots,k}$, the if-condition of line~\ref{atomic:CAS:ifCondition} will be satisfied and the Boolean value \False is returned without changing the value of object \sC.
Hence $\valcRegA{u_i} = \valCtrA{v_i}$ and the return values match.

\textbf{Case d - }
$u_i$ is a \CAS{} method call that returns from \Line{RCAS:CAS:CAS}.
Then $pt(u_i)$ is when the  \CAS operation on the object \Count\ is executed in \Line{RCAS:CAS:CAS}, and $u_i$ returns the result of this \CAS operation.
The \CAS operation attempts to change the value of \Count\ from $old$ to $new$, therefore if $\valcRegA{u_{i-1}} = old$ then $\valcRegA{u_i} = new$ and $u_i$ returns \True, or else \Count\ remains unchanged and $u_i$ returns \False.
Now consider operation $v_i$ in $S_v$.
From the code structure, if $\valCtrA{v_{i-1}} = old$ then $\valCtrA{v_i} = new$ and the Boolean value \True is returned.
And if $\valCtrA{v_{i-1}} \neq old$ then the value of object \sC\ remains unchanged and the Boolean value \False is returned.
Hence $\valcRegA{u_i} = \valCtrA{v_i}$ and the return values match.
\end{proof}


\begin{lemma} \label{claim:CASCounterFailureProbability}
The probability that an \inc{} method call returns $\bot$ is $ k/(k+1)$ against the weak adversary.
\end{lemma}

\begin{proof}
Let the process calling the \inc{} method call (say $u$) be $p$ and let the value of the object \Count\ immediately before $p$ executes line~\ref{RCAS:inc:randomChoice} be $z$.
Since the adversary is weak, no other process executes a shared memory operation after $p$ chooses $\beta$ in line~\ref{RCAS:inc:randomChoice} and before $p$ finishes executing its next shared memory operation.
From the code structure, $p$ returns $\bot$ during $u$ (in \Line{RCAS:inc:returnFail}) if and only if $z \neq \beta$.
Since
\begin{align*}
\Prob{z \neq \beta}  = 1 - \Prob{z = \beta} =  1 - \frac{1}{k+1} = \frac{k}{k+1},
\end{align*}
the claim follows.
%
%
\end{proof}

The following claim follows immediately from an inspection of the code.

\begin{lemma} \label{clm:ctr:complexity}
Each of the methods of $\RCASCounter{k}$ has step complexity \Order{1}, and is wait-free.
\end{lemma}


Lemma~\ref{thm:ctr:linearizable} follows from Lemmas~\ref{claim:ctr:linearizable}, \ref{claim:CASCounterFailureProbability} and \ref{clm:ctr:complexity}.
%

\section{Specification of Type \APArray{k}}
\label{appendix:sec:AbortableProArray}
Type \APArray{k} is presented in Figure~\ref{fig:SArray:seq}.
\begin{classfigure}[!htbp]
\setcounter{AlgoLine}{0}
  \begin{type}{\APArray{k}}
     $A$:  \Array$[0 \dots k-1]$ \Of \Int\ pairs \Init all \pair{\bot}{\bot};
     \qquad \cReg, \cPro, \cAbort: \Int\ \Init $1,2,3$\;
  \end{type}
\begin{minipage}{0.62\textwidth}
\LinesNumbered
\begin{algorithm}[H]\Operation{\collect{int[] $X$}}
  \For {$i \la\ 0$ \KwTo $k-1$} {
      $\pair{v}{s} \la\ A[i]$ \;
      \lIf {$v \neq \cAbort \wedge X[i] \neq \bot$} {
          $A[i]$ \la\ \pair{\cReg}{X[i]}\;
      }
  }
\end{algorithm}
\LinesNumbered
\begin{algorithm}[H]\Operation{\cUpdate{int $i$, int $seq$}}
  $\pair{v}{s} \la\ A[i]$ \;
  \lIf {$v = \cPro$} {             \label{update:ifCondition}
      \return \False \;                 \label{update:returnFalse}
  }
  $A[i]$ \la\ \pair{\cAbort}{seq} \;                    \label{update:setAi}
  \return \True \;                      \label{update:returnTrue}
\end{algorithm}
%
\LinesNumbered
\begin{algorithm}[H]\Operation{\reset{}}
  \lFor {$i \la\ 0$ \KwTo $k-1$} {
      $A[i]$ \la\ \pair{0}{\bot} \;
  }
\end{algorithm}
\end{minipage}
\begin{minipage}{0.36\textwidth}
\LinesNumbered
\begin{algorithm}[H]\Operation{\promote{}}
  \For {$i \la\ 0$ \KwTo $k-1$} {
      $\pair{v}{s} \la\ A[i]$ \;
      \If {$v = \cReg$} {
          $A[i]$ \la\ \pair{\cPro}{s} \;
          \return \pair{i}{s} \;
      }
  }
  \return \pair{\bot}{\bot} \;
\end{algorithm}
%
\LinesNumbered
\begin{algorithm}[H]\Operation{remove(int $i$)}
  $\pair{v}{s} \la\ A[i]$ \;
  $A[i]$ \la\ $(\cAbort,s)$ \;                    \label{remove:setAi}
\end{algorithm}
\end{minipage}
\caption{Sequential Specification of Type \APArray{k}} \label{fig:SArray:seq}
\end{classfigure}

\section{The Single-Fast-Multi-Slow Universal Construction}

In this section, rather than implementing object \UC{\T}, we implement a \emph{lock-free} universal construction object \UCWeak{\T}, with slightly weaker properties than \UC{\T}.
An object implementation is lock-free, if in any infinite history $H$ where processes continue to take steps, and $H$ contains only operations on that object, some operation finishes.
Object \UCWeak{\T} has the same properties as object \UC{\T} except method \performFast{} is lock-free with unbounded step-complexity.

There is a standard technique called \emph{operation combining} \cite{Herlihy:UCMethodology} that can be applied to transform our lock-free object \UCWeak{\T} to the wait-free object \UC{\T} with \Order{N} step complexity for method \performFast{}.


By applying the technique of operation combining we can transform our lock-free universal construction \UCWeak{\T} into our wait-free object \UC{\T}.
We however do not provide a proof of its properties.
Doing so would be repeating the same ``standard'' proof ideas from \cite{Herlihy:UCMethodology}, and would result in increasing the size of the paper without contributing to the main ideas of this paper.
We do provide proofs for our lock-free universal construction \UCWeak{\T}, and the proofs illustrate the main idea from this section, i.e., how to achieve a linearizable concurrent implementation with support for a \performFast{} method of \Order{1} step complexity.
We now present the implementation of object \UCWeakI{\T} (see Figure~\ref{fig:UCWeakI:Shared}).

\begin{classfigure}[!htbp]
\setcounter{AlgoLine}{0}
  \begin{object}{\UCWeakI{\T}} \label{Algo:UCWeakI}
  \shared
    \qquad $\mReg$: \Int\ \Init $(s_0,\bot,0,0)$;
    \qquad $\fastOp$: \Int\ \Init $(\bot,0)$;\;
  \local
    \qquad $state,res,fc,sc,s1,s1,r1,r2,seq$: \Int\ \Init $0$\;
  \end{object}
\begin{minipage}{0.42\textwidth}
  \begin{method}{\performFast{op}}
    $(state,res,fc,sc)$ \la\ \mReg.\Read{} \;                   \label{performFast:readmReg1}
    \sReg\ \la\ $(op,fc+1)$ \;          \label{performFast:sRegWrite}
    \lIf {$\neg \help{}$} {                   \label{performFast:callToHelp1}
        \help{} \;                                \label{performFast:callToHelp2}
    }
    $(state,res,fc,sc)$ \la\ \mReg.\Read{} \;                   \label{performFast:readmReg2}
    \return $res$ \;                           \label{performFast:return}
  \end{method}
\end{minipage}
\setcounter{AlgoLine}{8}
\begin{minipage}{0.565\textwidth}
  \begin{method}{\help{}}
      $(s1,r1,fc,sc)$ \la\ \mReg.\Read{} \;                 \label{help:readmReg}
      $(op,seq)$ \la\ \sReg.\Read{} \;        \label{help:readsReg}
      \lIf {$fc \geq seq$} {                \label{help:ifCondition}
          \return $\True$ \;
      }
      $(s2,r2)$ \la\ \f{$s1,op$} \;          \label{help:callTof}
      \return \mReg.\CAS{$(s1,r1,fc,sc),(s2,r2,seq,sc)$}\;            \label{help:CAS}
  \end{method}
\end{minipage}
\setcounter{AlgoLine}{5}
\begin{minipage}{0.42\textwidth}
  \begin{method}{\f{$state_1,op$}}
    $state_2$ \la\ state generated when $op$ is applied to object \SD\ with state $state_1$ \;
    $res$ \la\ result when $op$ is applied to object \SD\ with state $state_1$ \;
    \return $(state_2,res)$ \;  \label{f:returnResult}
  \end{method}
\end{minipage}
\setcounter{AlgoLine}{13}
\begin{minipage}{0.565\textwidth}
  \begin{method}{performSlow(op)}
    \Repeat {\mReg.\CAS{$(s1,r1,fc,sc),(s2,r1,fc,sc+1)$}} { \label{performSlow:beginLoop}
        $(s1,r1,fc,sc)$ \la\ \mReg.\Read{} \;               \label{performSlow:readmReg}
        $(s2,r2)$ \la\ \f{$s1,op$} \;    \label{performSlow:callTof}
        \lIf {$s2 = s1$} {                \label{performSlow:ifCondition}
            \return $r2$ \;                      \label{performSlow:returnWithoutChange}
        }
            \help{} \;                            \label{performSlow:callToHelp}
    }                                             \label{performSlow:CAS}
    \return $r2$ \;                              \label{performSlow:returnWithChange}
  \end{method}
\end{minipage}
\caption{Implementation of Object $\UCWeakI{\T}$.}\label{fig:UCWeakI:Shared}
\end{classfigure}


\bparagraph{Shared Data}
A shared register \mReg\ stores a $4$-tuple $(m_0, m_1, m_2, m_3)$.
We use the notation \mReg[$i$] to refer to the $(i+1)$-th tuple element, $m_i$, stored in register  \mReg.
Element \mReg[0] stores the state of object \SD.
Element \mReg[1] stores the result of the most recent fast operation performed.
Elements \mReg[2] and \mReg[3] store counts of the number of fast and slow operations performed
respectively.
Initially \mReg[$0$] stores the initial state of \SD, \mReg[$1$] has value $\bot$ and (\mReg[$2$],\mReg[$3$]) is ($0,0$).

A shared register \sReg\ is used to \emph{announce} a fast operation to be performed in a pair $(s_0,s_1)$.
Element \sReg[0] stores the complete description of a fast operation to be performed.
Element \sReg[1] stores a sequence number indicating the number of fast operations that have been announced in the past.
This sequence number is used by processes to determine whether an announced fast operation is pending execution.
Initially \sReg\ is ($\bot,0$).
The methods \performFast{} and \performSlow{} make use of two private methods \help{} and \f{} (see Figure~\ref{fig:UCWeakI:Shared}).

\bparagraph{Description of the \f{} method}
Method \f{} is implemented using the specification provided by type \T.
The method takes two arguments $state_1$ and $op$, where $state_1$ is a state of object \SD\ and $op$ is the complete description of an operation to be applied on object \SD.
The method computes the new state $state_2$ and the result $result$, when operation $op$ is applied on object \SD\ with state $state_1$.
The method then returns the pair ($state_2, result$).
Since no shared memory operations are executed during the method, the method has $0$ step complexity.

\bparagraph{Description of the \performFast{} method}
Let $p$ be a process that executes \performFast{$op$}.
In \Line{performFast:readmReg1}, process $p$ first copies the $4$-tuple read from register \mReg\ to its local variables $state, res, fc$ and $sc$.
Then $p$ announces the operation $op$ by writing the pair $(op,fc+1)$ to register \sReg\ in \Line{performFast:sRegWrite}.
After announcing the operation, process $p$ \emph{helps} perform the announced operation by calling the private method \help{} in \Line{performFast:callToHelp1}.
If the call to \help{} returns \false, then $p$ concludes that the announced operation may not have been performed yet.
In this case $p$ makes another call to \help{} in \Line{performFast:callToHelp2} to be sure that the announced operation is performed
(we prove later that at most two calls to \help{} are required to perform an announced operation).
Process $p$ then reads and returns the result of the performed operation stored in \mReg[1] in \Line{performFast:readmReg2} and~\ref{performFast:return}, respectively.
Since method \performFast{} is not executed concurrently (by assumption), the result of $p$'s operation stored in register \mReg\ is not overwritten before the end of $p$'s \performFast{op} call.

\bparagraph{Description of the \help{} method}
Let $q$ be a process that calls and executes \help{}.
In \Line{help:readmReg}, process $q$ first copies the $4$-tuple read from register \mReg\ into its local variables $s1, r1, fc$ and $sc$.
The value read from \mReg[0] constitutes the state of object \SD, to which $q$ will attempt to apply the announced operation if required.
The value read from \mReg[1] is the result of the last fast operation performed on object \SD.
The value read from \mReg[2] and \mReg[3] is the count of the number of fast and slow operations performed respectively.
Process $q$ then reads \sReg\ in \Line{help:readsReg} to find out the announced operation $op$ and the announced sequence number $seq$.
Process $q$ then determines whether the announced operation has already been performed, by checking whether $seq$ is less than or equal to $fc$ in \Line{help:ifCondition}.
If so, $q$ concludes that operation $op$ has been performed and returns \True, otherwise it attempts to perform $op$ in lines~\ref{help:callTof} and~\ref{help:CAS}.
In \Line{help:callTof} process $q$ calls the private method \f{} to compute the new state $s2$ and the result $r2$ when operation $op$ is applied to object \SD\ with state $s1$.
In \Line{help:CAS}, process $p$ attempts to perform $op$ by swapping the $4$-tuple $(s1,r1,fc,sc)$ with $(s2,r2,fc+1,sc)$ using a \CAS operation on \mReg.
If the \CAS is unsuccessful then no changes are made to \mReg.
This can happen only if some other process performs an announced fast operation in \Line{help:CAS} or a slow operation in \Line{performSlow:CAS}.
The result of the \CAS operation of \Line{help:CAS} is returned in either case.

\bparagraph{Description of the \performSlow{} method}
Let $p$ be a process that calls and executes \performSlow{$op$}.
During the method, $p$ repeats the while-loop of lines~\ref{performSlow:readmReg}-\ref{performSlow:CAS} until $p$ is able to successfully apply its operation $op$.
In \Line{performSlow:readmReg}, process $p$ first copies the $4$-tuple read from register \mReg\ to its local variables $s1, r1, fc$ and $sc$.
In \Line{performSlow:callTof} process $q$ calls the private method \f{} to compute the new state $s2$ and the result $r2$ when operation $op$ is applied to object \SD\ with state $s1$.
In the case that operation $op$ does not cause a state change in object \SD, i.e., $s1=s2$, then $p$ returns result $r2$ in \Line{performSlow:returnWithoutChange}.
Otherwise $p$ attempts to apply operation $op$ in \Line{performSlow:CAS} by swapping the $4$-tuple $(s1,r1,fc,sc)$ with $(s2,r2,fc,sc+1)$ using a \CAS operation on register \mReg.
Before attempting to apply its own operation in \Line{performSlow:CAS} $p$ makes a call to \help{} in \Line{performSlow:callToHelp} to help perform an announced fast operation (if any).
On completing the while-loop, $p$ would have successfully applied its operation $op$, and thus $p$ returns the result of the applied operation in \Line{performSlow:returnWithChange}.

The following lemma (proven in Section~\ref{sec:appendix:analysis:UCWeakI}) summarizes the properties of object \UCWeakI{\T}.
\begin{lemma}
\label{theorem:UCWeakI}
Object \UCWeakI{\T} is a lock-free universal construction object that implements an object \Obj\ of type \T, for $n$ processes, where $n$ is the maximum number of processes that can access object \UCWeakI{\T} concurrently and operations on object \Obj\ are performed using either method \performFast{} or \performSlow{}, and no two processes execute method \performFast{} concurrently.
Method \performFast{} has \Order{1} step complexity.
\end{lemma}

\subsection{Operation Combining Technique}
In principle the technique works as follows:
Processes maintain an $N$-element array, say \announce, where process $i$ ``owns'' slot $i$, and processes store in their respective slots the operation that they want to apply.
When a process $p$ wants to apply an operation it first ``announces'' its operation by writing the operation to the $p$-th element of the array.
Then $p$ attempts to \emph{help} the ``next'' operation in the \announce\ array by attempting to apply that operation if it has not been applied, yet.
An index to the ``next'' operation to be applied is maintained in the same register that stores the state of the concurrent object.
Every time an announced operation is applied, the index is also incremented modulo $N$ in one atomic step.
The response of applied operations is stored in another $N$-element array, say \response, which can sometimes be combined with the \announce\ array.
Sequence numbers are used to ensure that an announced operation is not applied more than once.
Since the index of the ``next'' operation cycles the \announce\ array, a process needs to help announced operations $\Order{N}$ times before its own announced operation is applied, at which point it can stop.

Herlihy \cite{Herlihy:UCMethodology} introduced this technique as a general methodology to transform lock-free universal constructions to wait-free ones.
Herlihy presents another example \cite{Herlihy:BookUCExample} that employs the technique of operation combining to transform a lock-free universal construction to a wait-free one, where the step complexity of the method that performs the operation is bounded to \Order{N}.

On applying the standard technique of operation combining \cite{Herlihy:UCMethodology} to object \UCWeakI{\T} we obtain object \UC{\T} and Lemma~\ref{theorem:UC}


\subsection{Analysis and Proofs of Correctness of Object \UCWeakI{\T} }
\label{sec:appendix:analysis:UCWeakI}
Let a \help{} method call that returns \true\ in \Line{help:CAS} (on executing a successful \CAS operation) be called a \emph{successful} \help{}.

\begin{claim}
\label{claim:UCWeakI:register}
 \begin{enumerate}[(a)]
 \item The value of $\sReg[1]$ changes only in \Line{performFast:sRegWrite}.\label{sReg[1]}
 \item The value of $\mReg[3]$ increases by one with every successful \CAS operation in \Line{performSlow:CAS} and no other operation changes $\mReg[3]$.\label{mReg[3]}
 \item The value of $\mReg[2]$ increases with every successful \CAS operation in \Line{help:CAS} (during a successful \help{}), and no other operation changes $\mReg[2]$.\label{mReg[2]}
\end{enumerate}
\end{claim}
\begin{proof}
Part~\refC{sReg[1]} follows immediately from an inspection of the code.
Register \mReg\ is changed only when a process executes a successful \CAS operation in lines~\ref{help:CAS} or~\ref{performSlow:CAS}.
Furthermore, in \Line{help:CAS} \mReg[3] is not changed and in \Line{performSlow:CAS} \mReg[2] is not changed.
Since, in \Line{performSlow:CAS} \mReg[3] is incremented Part~\refC{mReg[3]} follows immediately.
Now, for a process to execute \Line{help:CAS}, the if-condition of \Line{help:ifCondition} must fail, hence \mReg[2] is increased from its previous value and Part~\refC{mReg[2]} follows.
\end{proof}

Consider an arbitrary history $H$ where processes access an \UCWeakI{\T}\ object but no two \performFast{} method calls are executed concurrently.
Since the fast operations are executed sequentially the happens before order on all \performFast{} method calls in $H$ is a total order.

\begin{claim}
\label{claim:UCWeakI:thereExistsASuccessfulHelp}
Let $u_t$ be the $t$-th \performFast{} method call in history $H$ being executed by process $p_t$.
For $t \geq 1$ let $\Alpha{t}$ be the point in time when $p_t$ executes \Line{performFast:sRegWrite} during $u_t$ and $\Gemma{t}$ be the point when $p_t$ is poised to execute \Line{performFast:readmReg2}.
Let $u_t$'s \emph{helpers} be the processes that call \help{} such that the value read by the processes in \Line{help:readsReg} is the value written to register \sReg\ at \Alpha{t}.
Let $\Beta{t}$ be the first point in time when a helper's call to \help{} succeeds after $\Alpha{t}$.
Let $\Alpha{0},\Beta{0},\Gemma{0}$ be the start of execution $H$.
Then the following claims hold for all $t \geq 0$:
\begin{enumerate}[(S$_1$)]
  \item \Beta{t} exists and \Beta{t} is in  $(\Alpha{t},\Gemma{t})$
  \item Throughout $(\Alpha{t},\Beta{t})$ : $\sReg[1] = \mReg[2]+1 = t$
  \item Throughout $(\Beta{t},\Alpha{t+1})$ : $\sReg[1] = \mReg[2] = t$

\end{enumerate}

\end{claim}

\begin{proof}
We prove claims $(S_1),(S_2)$ and $(S_3)$ by induction over $t$.

\textbf{Basis:}
For $t=0$, $(S_1)$ and  $(S_2)$ are trivially true.
By assumption the initial value of \sReg[1] and \mReg[2] is $0$.
Consider the interval $(\Beta{0},\Alpha{1})$.
From Claim~\ref{claim:UCWeakI:register}(\ref{sReg[1]}) it follows that \sReg\ is written for the first time at $\Alpha{1}$.
The first point when one of the invariants $(S_3)$ is destroyed is if a process (say $p$) executes a successful \CAS operation in \Line{help:CAS} during $(\Beta{0},\Alpha{1})$.
Then $p$ read the value $0$ from register $\sReg[1]$ in \Line{help:readsReg}, since the initial value of \sReg[1] is $0$ and \sReg[1] is written to for the first time at $\Alpha{1}$.
Since \mReg[2] is never decremented (from Claim~\ref{claim:UCWeakI:register}(\ref{mReg[2]})) and \mReg[2] initially has value $0$, $p$ satisfies the if-condition of \Line{help:ifCondition} and $p$'s \help{} call returns \true\ in \Line{help:ifCondition}.
Therefore, $p$ does not execute \Line{help:CAS}, which is a contradiction.

\textbf{Induction Step:}
For $t \geq 1$:\

\textbf{Proof of $(S_1)$: }
Consider the interval $(\Alpha{t},\Gemma{t})$.
To show that $(S_1)$ holds for $t$, we need to show that \mReg[2] is changed during $(\Alpha{t},\Gemma{t})$.
Consider $p_t$'s first call to \help{} in \Line{performFast:callToHelp1} during $u_t$.
From induction hypothesis $(S_3)$ for $t-1$, it follows that $\sReg[1] = \mReg[2] = t-1$ during $(\Beta{t-1},\Alpha{t})$.
Then $p_t$ reads value $t-1$ from \mReg[2] in \Line{performFast:readmReg1} and writes value $t$ to \sReg[1] in \Line{performFast:sRegWrite}.
Since \sReg[1] is changed only at \Alpha{t+1} after \Alpha{t}, it follows that $p_t$ reads $t$ from register $\sReg[1]$ in line~\ref{help:readsReg}.

\textbf{Case a - } $p_t$ returns from \Line{help:ifCondition}:
Then $p_t$ read a value from \mReg[2] in \Line{help:readmReg} that is at least $t$.
Since $\mReg[2]= t-1$ holds immediately before \Alpha{t} some process changed \mReg[2] in \Line{help:CAS} during $(\Alpha{t},\Gemma{t})$.
Hence, $(S_1)$ for $t$ holds.

\textbf{Case b - } $p_t$ returns \true\ from \Line{help:CAS}:
Then $p_t$ has changed \mReg[2] and hence $(S_1)$ holds for $t$.

\textbf{Case c - } $p_t$ returns \false\ from \Line{help:CAS}:
Then some process $q$ changed register \mReg\ after $p_t$ read \mReg\ in \Line{help:readmReg}.
Now, register \mReg\ is written to only in \Line{help:CAS} or \Line{performSlow:CAS} (from an inspection of the code).

\textbf{Subcase c1 - } $q$ changed \mReg\ by executing \Line{help:CAS}:
Then $q$ has changed \mReg[2] and hence $(S_1)$ holds for $t$.

\textbf{Subcase c2 - } $q$ changed \mReg\ by executing \Line{performSlow:CAS}:
Then $p_t$ executes a second call to \help{} in \Line{performFast:callToHelp2}.
Let $m$ be the value of \mReg[3] read by $p_t$ in \Line{help:readmReg}.
If $p_t$'s second \help{} call satisfies case (a) or (b) then we get that $(S_1)$ holds for $t$.

If $p_t$'s second \help{} call returns \false\ from \Line{help:CAS}, then some process changed \mReg\ after $p_t$ read \mReg\ in \Line{help:readmReg}.
If some process changed \mReg\ by executing \Line{help:CAS} then we get that $(S_1)$ holds for $t$.
Then some process changed \mReg\ by executing \Line{performSlow:CAS} after $p_t$ read \mReg\ in \Line{help:readmReg} and let $r$ be the first process to do so.
Therefore, $r$ changes the value of \mReg[3] from $m$ to $m+1$ in \Line{performSlow:CAS}.
Then $r$ executed \Line{performSlow:readmReg} after $q$ executed a successful \CAS operation in \Line{performSlow:CAS}.
Then $r$ completed a call to \help{} in \Line{performSlow:callToHelp} after \Alpha{t}.
Since $r$ reads \mReg\ after \Alpha{t}, $r$ satisfied the if-condition of \Line{help:ifCondition} and executed \Line{help:CAS}.
If $r$ successfully executes the \CAS operation in \Line{help:CAS} then we get that $(S_1)$ holds for $t$.
Then some process $s$ must have changed \mReg\ after $r$ read \mReg\ in \Line{help:readmReg}.
Since the value of \mReg[3] is only incremented (by \Claim{claim:UCWeakI:register}(b)) and $r$ changes the value of \mReg[3] from $m$ to $m+1$, it follows that $s$ changed \mReg\ in \Line{help:CAS} and hence $(S_1)$ holds for $t$.

\textbf{Proof of $(S_2)$ and $(S_3)$: }
From $(S_1)$ for $t$ it follows that \Beta{t} exists and $\Alpha{t} < \Beta{t} < \Gemma{t} < \Alpha{t+1}$.
From the induction hypothesis invariants $(S_2)$ and $(S_3)$ are true until $\Alpha{t}$.
Now, one of the invariants $(S_2)$ or $(S_3)$ can be destroyed only if some process executes a successful \CAS operation in \Line{help:CAS} and changes \mReg[2].
By definition of \Beta{t}, \mReg[2] is unchanged during $(\Alpha{t},\Beta{t})$.
Then invariants $(S_2)$ and $(S_3)$ continue to hold until $\Beta{t}$.
Therefore, claim $(S_2)$ holds for $t$.
It still remains to be shown that claim $(S_3)$ holds for $t$.

Let $p$ be the process that executes a successful \CAS operation in \Line{help:CAS} and changes \mReg[2] at \Beta{t}.
Since $\mReg[2] = t-1$ immediately before \Beta{t} and $p$ executes a successful \CAS operation in \Line{help:CAS} at $\Beta{t}$, $p.fc = t-1$.
Then $p$ executed lines~\ref{help:readmReg} and~\ref{help:readsReg} during $(\Alpha{t},\Beta{t})$ and $p.seq = t$.
Therefore, invariant $(S_3)$ is true immediately after $\Beta{t}$.

Now, assume another process (say $q$) destroys one of the invariants $(S_3)$ or $S_4$ by executing a successful \CAS operation in \Line{help:CAS} during $(\Beta{t},\Alpha{t+1})$.
Then $q$ must have read register \mReg[2] and \sReg[1] after \Beta{t}, and therefore $q$ must read the value $t$ from both of them.
Then $q$ must have satisfied the if-condition of \Line{help:ifCondition} and returned \true.
Hence, $q$ does not execute \Line{help:CAS}, which is a contradiction.
Therefore, invariant $(S_3)$ is true up to $\Alpha{t+1}$, and thus claim $(S_3)$ holds for $t$.
\end{proof}


Let $H'$ be a history that consists of all completed method calls in $H$ and all pending method calls that executed \Line{performFast:sRegWrite} (\Write operation on register \sReg), or which executed a successful \CAS operation in \Line{help:CAS} or \Line{performSlow:CAS}.
We omit all other pending method calls, since during those method calls no operations are executed that changes the state of any shared object, and hence those pending method calls cannot affect the validity of any other operation.
Therefore, to prove that history $H$ is linearizable it suffices to prove that history $H'$ is linearizable.

For each method call $u$ in $H'$, we define a point $pt(u)$ and an interval $I(u)$.
Let $I(u)$ denote the interval between $u$'s invocation and response.
If $u$ is a \performSlow{} method call that returns from \Line{performSlow:returnWithoutChange} then $pt(u)$ is the point in time of the \Read operation in \Line{performSlow:readmReg}, otherwise, $pt(u)$ is the point in time of the \CAS operation in \Line{performSlow:CAS}.
If $u$ is a \performFast{} method call, let $v$ be a successful \help{} method call such that $v$'s \Line{help:CAS} is executed after $u$'s \Line{performFast:sRegWrite} and before $u$ returns.
We define $pt(u)$ to be the point of the successful \CAS operation in $v$'s \Line{help:CAS}.

\begin{claim}
\label{claim:UCWeakI:thereExistsPointU}
For every method call $u$ in $H$, $pt(u)$ exists and lies in $I(u)$.
\end{claim}
\begin{proof}
There are two types of method calls in $H$, \performFast{} and \performSlow{}.

\textbf{Case a - } $u$ is a \performFast{} method call.\

From \Claim{claim:UCWeakI:thereExistsASuccessfulHelp} it follows that exactly one of $u$'s helpers (see Claim~\ref{claim:UCWeakI:thereExistsASuccessfulHelp} for definition) succeeds and the helper performs a successful \CAS operation in \Line{help:CAS} at some point in $I(u)$.
Therefore, point $pt(u)$ exists and lies in $I(u)$.

\textbf{Case b - } $u$ is a \performSlow{} method call.
By definition $pt(u)$ is assigned to a line of $u$'s code, therefore $pt(u)$ exists and lies in $I(u)$.
\end{proof}

Let $S$ be the sequential history obtained by ordering all method calls $u$ in $H'$ according to the points $pt(u)$.
To show that \UCWeak{\T} is a linearizable implementation of an object \SD\ of type \T, we need to show that the sequential history $S$ is valid, i.e., $S$ lies in the specification of type \T, and that $pt(u)$ lies in $I(u)$ (already shown in \Claim{claim:UCWeakI:thereExistsPointU}).
Let $S_v$ be the sequential history obtained when the operations of $S$ are executed sequentially on object \SD, as per their order in $S$.
Clearly, $S_v$ is a valid sequential history in the specification of type \T\ by construction.
Then to show that $S$ is valid, we show that $S = S_v$.


Let $v_t$ be the $t$-th operation in $S_v$ and let $u_t$ be the $t$-th method call in $S$.
Let \UCb{t} and \UCa{t} denote the value of \mReg[0] immediately before and after $pt(u_t)$, respectively.
Let \SDb{t} and \SDa{t} denote the state of object \SD\ immediately before and after operation $v_t$, respectively.
Let $\alpha_t$ and $\beta_t$ denote the value returned by $u_t$ and $v_t$, respectively.
Define $\UCa{0} = \UCb{1}$ and $\SDa{0} = \SDb{1}$.
Define $\alpha_{0} = \beta_{0} = \bot$.

\begin{claim} \label{claim:UCWeakI:afterPtuConsistency}
Suppose a process calls method \f{$x_1,op$} and the method returns the value pair $(x_2,y)$.
If $x_1 = \SDb{t}$ then $x_2 = \SDa{t}$ and $y=\beta_t$.
\end{claim}

\begin{proof}
By definition, a call to method \f{$x_1,op$} returns the value pair $(x_2,y)$ such that $x_2$ is the state of \SD\ when operation $op$ is applied to \SD\ while at state $x_1$ and $y$ is the result of the operation.
Then if $x_1 = \SDb{t}$ then $x_2 = \SDa{t}$ and $y=\beta_t$.
\end{proof}

\begin{claim}
\label{claim:UCWeakI:invariants}
For all $t \geq 1$.
\begin{enumerate}[$(S_1)$]
 \item \SDa{t} = \SDb{t+1} and \UCa{t} = \UCb{t+1}
 \item \UCb{t} = \SDb{t}
 \item \UCa{t-1} = \SDa{t-1} and $\alpha_{t-1} = \beta_{t-1}$
\end{enumerate}
\end{claim}
\begin{proof}

\textbf{Proof of $(S_1)$:}
Since operations in $S_v$ are executed sequentially, it follows that \SDa{t} = \SDb{{t+1}}.
We now show that \UCa{t} = \UCb{{t+1}}.
Assume $\UCa{t} \neq \UCb{{t+1}}$.
Then some process $p$ changed the value of \mReg[0] by executing a successful \CAS operation in \Line{help:CAS} or \Line{performSlow:CAS} at some point during the interval $(pt(u_t),pt(u_{t+1}))$.
By definition, $p$'s successful \CAS operation in \Line{help:CAS} or \Line{performSlow:CAS} is $pt(u_\ell)$ for some method call $u_\ell$ where $u_\ell$ is the $\ell$-th method call in $H'$.
Thus, $\ell$ is an integer and $t < \ell < t + 1$ holds, which is a contradiction.

\textbf{Proof of $(S_2)$ and $(S_3)$:}
We prove $(S_2)$ and $(S_3)$ by induction over $t$.

\textbf{Basis $(t=1)$ -}
By assumption, initially, \mReg[0] is the initial state of \SD, hence, $\UCb{1} = \SDb{1}$.
Hence, $(S_2)$ is true.
$(S_3)$ is true trivially.

\textbf{Induction Step -}
We assume $(S_2)$ and $(S_3)$ for $t$ are true and prove that $(S_2)$ and $(S_3)$ for $t+1$ are true.
From $(S_1)$ we have, \SDa{t} = \SDb{{t+1}} and \UCa{t} = \UCb{{t+1}}.
From $(S_3)$ for $t$ we have \UCa{t} = \SDa{t}.
Therefore, it follows that \UCb{{t+1}} = \SDb{{t+1}} and thus $(S_2)$ for $t+1$ is true.

To show $(S_3)$ for $t+1$ is true, we need to show \UCa{t} = \SDa{t} and $\alpha_{t} = \beta_{t}$.
By Claim $(S_2)$ for $t$, \UCb{t} = \SDb{t} holds.
Let $p_t$ be the process executing $u_t$.

\textbf{Case a - } $u_t$ is a \performSlow{$op$} method call:
Let $x_1$ be the most recent value read by $p_t$ from \mReg[0] in \Line{performSlow:readmReg} and let $(x_2,y)$ be the value returned when $p_t$ executes \Line{performSlow:callTof}.
From the code structure, $\alpha_t = y$.

\textbf{Subcase (a1) - } $p_t$ returns from \Line{performSlow:returnWithoutChange}:
Then $pt(u_t)$ is the point when $p_t$ executes a successful \Read operation on register \mReg\ in \Line{performSlow:readmReg}.
Since $p$ satisfies the if-condition of \Line{performSlow:ifCondition}, $x_2 = x_1$.
Thus, $\UCb{t} = \UCa{t} = x_1$.

\textbf{Subcase (a2) - } $p_t$ returns from \Line{performSlow:returnWithChange}:
Then $pt(u_t)$ is the point when $p_t$ executes a successful \CAS operation on register \mReg\ in \Line{performSlow:CAS}.
From the definition of a \CAS operation, it follows that \UCb{t} = $x_1$ and \UCa{t} = $x_2$.

For both subcases \textbf{(a1)} and \textbf{(a2)}, $x_1 = \UCb{t} = \SDb{t}$ holds.
Then from \Claim{claim:UCWeakI:afterPtuConsistency} it follows that $x_2 = \SDa{t}$ and $y = \beta_{t}$.
Since $x_2 = \UCa{t}$ and $y = \alpha_t$, $\SDa{t} = \UCa{t}$ and $\alpha_t = \beta_t$.

\textbf{Case b - } $u_t$ is a \performFast{$op$} method call:
Then $pt(u_t)$ is the point when a successful \CAS operation on register \mReg\ is executed in \Line{help:CAS} of method call $w$ where $w$ is the first successful \help{} method call that begins after $u_t$'s \Line{performFast:sRegWrite} is executed.
Let $q$ be the process executing $w$.
Let $x_1$ be the value read by $q$ from \mReg[0] in \Line{help:readmReg} and let $(x_2,y)$ be the value returned when $q$ executes \Line{help:callTof}.
From the definition of a \CAS operation, it follows that \UCb{t} = $x_1$ and \UCa{t} = $x_2$.
Since $x_1 = \UCb{t} = \SDb{t}$, from \Claim{claim:UCWeakI:afterPtuConsistency} it follows that $x_2 = \SDa{t}$ and $y = \beta_{u_t}$.
Since $x_2 = \UCa{t}$, it follows that $\SDa{t} = \UCa{t}$.

From \Claim{claim:UCWeakI:thereExistsASuccessfulHelp} if follows that \mReg[1] is changed exactly once during $u_t$, specifically at $pt(u_t)$, where $q$ writes the value $y$ to it.
Thus, $p$ reads the value $y$ from \mReg[1] in \Line{performFast:readmReg2} since $p$ executes \Line{performFast:readmReg2} after $pt(u_t)$ (\Claim{claim:UCWeakI:thereExistsASuccessfulHelp}).
Therefore, it follows that $\alpha_{u_t} = y = \beta_{u_t}$.
\end{proof}

\begin{lemma}
\label{claim:UCWeakI:linearizability}
History $H'$ has a linearization in the specification of \T.
\end{lemma}

\begin{proof}
By \Claim{claim:UCWeakI:thereExistsPointU}, for each method call $u$ in $H'$, $pt(u)$ exists and lies in $I(u)$.
Thus, to show that $H'$ is linearizable we only need to show that
$S$ lies in the specification of type \T.
Thus, we need to show that for all $t \geq 1$, the value returned by $v_t$ matches that value returned by $u_t$.
From \Claim{claim:UCWeakI:invariants} $(S_3)$ it follows that for all $t \geq 1$, $\Alpha{t} = \beta_{t}$.
\end{proof}


\begin{lemma}
\label{claim:UCWeakI:lockfree}
Object \UCWeak{\T} is lock-free.
\end{lemma}
\begin{proof}
Suppose not. I.e., there exists an infinite history $H$ during which processes take steps but no method call finishes.
It is clear from an inspection of method \performFast{} and private method \helpFast{}, that both methods are wait-free.
Then if $H$ contains steps executed by a process that executes a call to \performFast{} then the \performFast{} method call finishes since processes continue to take steps in history $H$ -- a contradiction.
Now consider the only other case, where history $H$ contains steps executed by processes only on \performFast{} method calls.
Consider a process $p$ that takes steps in history $H$ and fails to complete its \performSlow{} method call.
Then during $p$'s execution $p$ reads register \mReg\ in \Line{performSlow:readmReg} and fails its \CAS operation in \Line{performSlow:CAS} during an iteration of the loop of lines~\ref{performSlow:beginLoop}-\ref{performSlow:CAS}.
Now $p$'s \CAS\ operation can fail only if some process executes a successful \CAS operation in \Line{help:CAS} or \Line{performSlow:CAS} between $p$'s \Read{} and \CAS\ operation.

\textbf{Case a -} Some process $q$ executes a successful \CAS operation in \Line{performSlow:CAS}.
Then $q$ breaks out of the loop of lines~\ref{performSlow:beginLoop}-\ref{performSlow:CAS}.
Since processes continue to take steps in our infinite history $H$, $q$ eventually returns from its \performSlow{} method call -- a contradiction.

\textbf{Case b -} Some process $q$ executes a successful \CAS operation in \Line{help:CAS}.
Then $q$ has performed a successful \helpFast{} method call and incremented \mReg[2].
Let the value of \mReg[2] after the increment be $z$.
Now consider the next iteration of the loop by process $p$, where $p$'s \CAS operation  in \Line{performSlow:CAS} fails again.
Since \textbf{Case a} leads to a contradiction, some process $r$ executed a successful \CAS operation in \Line{help:CAS}.
Then $r$ read incremented \mReg[2] to some value greater than $z$ in \Line{help:CAS}.
From the code structure of the \helpFast{} method, $r$ failed the if-condition of \Line{help:ifCondition}, and therefore $r$ read $seq = \sReg[1] > z$ in \Line{help:readsReg}.
Since \sReg[1] is incremented only in \Line{performFast:sRegWrite} during a \performFast{} method call, it follows that a \performFast{} method was called after $q$ incremented \mReg[2] to $z$ in \Line{help:CAS}.
This is a contradiction to the assumption that processes take steps executing only method \performSlow{} during our history $H$.
\end{proof}

Lemma~\ref{theorem:UCWeakI} follows from Lemma~\ref{claim:UCWeakI:linearizability} and \ref{claim:UCWeakI:lockfree}.

\section{The Array Based Randomized Abortable Lock}
\label{sec:appendix:ARLockArray}

\subsection{Implementation / Low Level Description} \label{sec:ARLockArray:Implementation}
We now describe the implementation of our algorithm in detail. (See Figure~\ref{fig:ARMEAlgorithm1} and \ref{fig:ARMEAlgorithm2}).
We now describe the method calls in detail and illustrate the use of each of the internal objects as and when we require them.

\bparagraph{The \lock{} method}
Suppose $p$ executes a call to \lock{i}.
Process $p$ first receives a sequence number using a call to \getSequenceNo{} in \Line{getLock:getSequenceNo} and stores it in its local variable $s$.
Method \getSequenceNo{} returns integer $k$ on being called for the $k$-th time from a call to \lock{i}.
Since calls to \lock{i} are executed sequentially, a sequential shared counter suffices to implement method \getSequenceNo{}.
Method \getSequenceNo{} is used to return unique sequence number which helps solve the classic ABA problem.
The ABA problem is as follows:
If a process reads an object twice and reads the value of the object to be 'A' both times, then it is unable to differentiate this scenario from a scenario where the object was changed to value 'B' in between the two reads of the object.
Process $p$ then spins on \Apply[i] in \Line{getLock:ApplyBotWant} until $p$ \emph{registers} itself by swapping the value \pair{\cWant}{s} into \Apply[i] using a \CAS\ operation.
Processes write the value $\cWant$ in the \Apply\ array to announce their presence at lock \L.

Process $p$ then executes the \emph{role-loop}, lines~\ref{getLock:BeginInnerLoop}-\ref{getLock:EndInnerLoop}, until $p$ either increases the value of \ctr\ to $1$ or $2$, or until $p$ is notified of its promotion.
Process $p$ begins an iteration of the role-loop by calling the \ctr.\inc{} operation in \Line{getLock:IncCounter} and stores the returned value into $\Role[i]$.
The returned value determines $p$'s current role at lock \L.
The shared array \Role\ is used by process $p$ to store its role in slot \Role[i], which can later be read to determine the actions to perform at lock \L.
This is important because we want to allow the behavior of transferring locks.
Specifically, to enable a process $q$ to call \release{i}{} on behalf of $p$, $q$ needs to determine $p$'s role at lock \L, which is possible by reading \Role[i].


If the \ctr.\inc{} operation in \Line{getLock:IncCounter} fails, i.e., it returns $\bot$, then $p$ repeats the role-loop.
Such repeats can happen only a constant number of times in expectation (by Claim~\ref{claim:CASCounterFailureProbability}).
If the value returned in \Line{getLock:IncCounter} is $0$ or $1$, then $p$ has incremented the value of \ctr\ (from the semantics of a \RCASCounter{2} object), and it becomes \king{\L} or \queen{\L},
respectively, and breaks out of the role-loop in \Line{getLock:EndInnerLoop}.

If $p$ becomes \king{\L} in \Line{getLock:IncCounter}, then $p$ fails the if-condition of \Line{getLock:ifQueen} and proceeds to execute lines~\ref{getLock:ApplyWantOk}-\ref{getLock:end}.
In \Line{getLock:ApplyWantOk}, $p$ changes \Apply[i] to the value \pair{\cOk}{s}, to prevent itself from getting promoted in future promote actions.
In \Line{getLock:returnX}, $p$ returns from its \lock{} call by returning the special value $\infty$ (a non-$\bot$ value indicating a successful \lock{} call), since $p$ is \king{\L}.

If $p$ becomes \queen{\L} in \Line{getLock:IncCounter}, then $p$ knows that there exists a king process at lock \L, and thus \queen{\L} proceeds to spin on \X\ in \Line{getLock:awaitX} awaiting a notification from \king{\L}.
Recall that \king{\L} notifies \queen{\L} of \queen{\L}'s turn to own lock \L\ by writing the integer $j$ into \X\ during a \release{}{j} call.
Once $p$ receives \king{\L}'s notification (by reading a non-$\bot$ value in \X\ in \Line{getLock:awaitX}), $p$ breaks out of the spin loop of \Line{getLock:awaitX}, and  proceeds to execute lines~\ref{getLock:ApplyWantOk}-\ref{getLock:end}.
In \Line{getLock:ApplyWantOk}, $p$ changes \Apply[i] to the value \pair{\cOk}{s}, to prevent itself from getting promoted in future promote actions.
In \Line{getLock:returnX}, $p$ returns from its \lock{} call by returning the integer value stored in \X\ (a non-$\bot$ value indicating a successful \lock{} call).

If the value returned in \Line{getLock:IncCounter} is $2$, then $p$ does not become \king{\L} or \queen{\L}, and thus $p$ assumes the role of a pawn.
Process $p$ then waits for a notification of its own promotion, or, for the \ctr\ value to decrease from $2$, by spinning on \Apply[i] and \ctr\ in \Line{getLock:awaitAckOrCtrDecrease}.
When $p$ breaks out of this spin lock, it determines in \Line{getLock:ifBackpacked} whether it was promoted by checking whether the value of \Apply[i] was changed to \pair{\cOk}{s}.
A process is promoted only by a \king{\L}, \queen{\L} or a \ppawn{\L} during their \release{}{} call.
If $p$ finds that it was not promoted, then $p$ is said to have been \emph{missed} during a \ctr-cycle,
and thus $p$ repeats the role-loop.
We later show that a process gets missed during at most one \ctr-cycle.

If $p$ was promoted, then it  writes a constant value $\cPPawn = 3$ into $\Role[i]$ in \Line{getLock:RolePPawn} and becomes \ppawn{\L}.
Since $p$ has been promoted, $p$ knows that both \king{\L} and \queen{\L} are no longer executing their entry or Critical Section, and thus $p$ owns lock \L\ now.
Then $p$ goes on to break out of the role-loop in \Line{getLock:EndInnerLoop}, and proceeds to return from its \lock{} call by returning the special value $\infty$ (a non-$\bot$ value indicating a successful \lock{} call), since $p$ is \ppawn{\L}.

\bparagraph{The \release{}{} method}
Suppose $p$ executes a call to \release{i}{j} with an integer argument $j$.
We restrict the execution such that a process calls a \release{i}{j} method only after a call to a successful \lock{i} has been completed.

In \Line{release:setrNotCollected}, $p$ initializes the local variable $r$ to the boolean value $\False$.
Local variable $r$ is returned later in \Line{release:return} to indicate whether the integer $j$ was successfully written to \X\ during the release method call.
In lines~\ref{release:ifKing},~\ref{release:ifQueen} and~\ref{release:ifPPawn} process $p$ determines its role at the node and the action to perform.
In \Line{release:ApplyOkBot}, process $p$ deregisters itself from lock \L\ by swapping \pair{\bot}{\bot} into \Apply[i].
At the end of the method call a boolean is returned in \Line{release:return}, indicating whether the integer $j$ was written to \X.

If $p$ determines that it is \king{\L}, then it attempts to decrease \ctr\ from $1$ to $0$ in \Line{release:ctr10}.
This decrement operation will only fail if there exists a queen process at lock \L\ which increased the \ctr\ to $2$ during its \lock{} call.
If the decrement operation fails then $p$ has determined that there exists a queen process at lock \L\ and it now synchronizes with \queen{\L} to perform the collect action.
Recall that \CAS\ object \X\ is used by \king{\L} and \queen{\L} to determine which process performs a collect.
In \Line{release:setX}, $p$ attempts to swap integer $j$ into \X\ by executing a \X.\CAS{$\bot,j$} operation and stores the result of the operation in local variable $r$.
If $p$ is successful then it performs the collect action by executing a call to \doCollect{i} in line~\ref{release:doCollect}.
If $p$ is unsuccessful then it knows that \queen{\L} will perform a collect.
In \Line{release:callhRelease:King} $p$ calls the \helpRelease{i} method to synchronize the release of lock \L\ with \queen{\L}.
We describe the method \helpRelease{} shortly.

If $p$ determines that it is \queen{\L}, then it calls the \helpRelease{i} method call in \Line{release:callhRelease:Queen} to synchronize the release of lock \L\ with \king{\L}.

If $p$ determines that it is a promoted pawn, then it attempts to promote a waiting pawn by making a call to \doPromote{} in \Line{release:callPromote}.

\bparagraph{The \doCollect{} method}
Suppose a process $p$ executing a \doCollect{i} method call.
The collect action consists of reading the \Apply\ array (left to right), and creating a vector $A$ of $n$ values, where the $k$-th element is either $\bot$ (to indicate that the process with pseudo-ID $k$ is not a candidate for promotion) or an integer sequence number (to indicate that the process with pseudo-ID $k$ is a candidate for promotion).
The vector $A$ is stored in the \APArray{n} instance \PawnSet\ in \Line{collect:updateAll} using a \PawnSet.\collect{$A$} operation.
The \PawnSet.\collect{$A$} operation ensures that if the $k$-th element of \PawnSet\ has value $3 = \cAbort$ (written during a \PawnSet.\cUpdate{$k,\cdot$} operation), then the $k$-th element is not overwritten during the \PawnSet.\collect{$A$} operation.
This is required to ensure that processes that have expressed a desire to abort are not collected and subsequently promoted.

\bparagraph{The \helpRelease{} method}
Suppose \king{\L} calls \helpRelease{i} and \queen{\L} calls \helpRelease{k}.
During the course of these method calls, \king{\L} and \queen{\L} synchronize with each other in order to reset \CAS objects \X\ and \LSync, remove themselves from \PawnSet, promote a collected process and notify the promoted process.
If no process is found in \PawnSet\ that can be promoted, then the \PawnSet\ object is reset to its initial state and \ctr\ reset to $0$.
Recall that \CAS object \LSync\ is used as a synchronization primitive by \king{\L} and \queen{\L} to determine which process exits last among them, and thus performs all pending release work.
In \Line{hRelease:setT}, the process which swaps value $i$ or $k$ into \LSync\ by executing a successful \CAS operation, exits, and the other process performs the pending release work in lines~\ref{hRelease:readX} -~\ref{hRelease:end}.
Let us now refer to this other process as the releasing process. 
In lines~\ref{hRelease:readX} -~\ref{hRelease:resetX}, the releasing process resets \X\ to its initial value $\bot$.
In line~\ref{hRelease:readT}, the releasing process reads the pseudo-ID written to \LSync\ by the exited process (process that executed a successful \CAS operation on \X).
The pseudo-ID written to \LSync\ is required to remove the exited process from getting promoted in a future promote in case it was collected in \PawnSet.
In line~\ref{hRelease:collectFixOther}, the releasing process removes the exited process from \PawnSet.
\CAS object \LSync\ is reset to its initial value $\bot$ in \Line{hRelease:resetT}.
In \Line{hRelease:callPromote}, the releasing process calls \doPromote{} to promote a collected process.


\bparagraph{The \doPromote{} method}
Suppose $p$ executes a call to \doPromote{i}.
In line~\ref{promote:collectFixSelf}, $p$ removes itself from \PawnSet\ by executing a \PawnSet.\remove{$i$} operation.
It does so to prevent itself from getting promoted in case it was collected earlier.
In \Line{promote:FR12}, $p$ performs a promote action by executing a \PawnSet.\promote{} operation.
If a process was collected and the process has not aborted then its corresponding element ($k$-th element for a process with pseudo-ID $k$) in \PawnSet\ will have the value \pair{\cReg}{\cdot}.
If a process has aborted then its corresponding element in \PawnSet\ will have the value \pair{\cPro}{\cdot}.

If a successful \promote{} operation is executed then an element in \PawnSet\ is changed from \pair{\cReg}{s} to \pair{\cPro}{s}, where $s \in \N$, and the pair \pair{k}{s} is returned, where $k$ is the index of that element in \PawnSet.
In this case we say that process with pseudo-ID $k$ was \emph{promoted}.
If an unsuccessful \promote{} operation is executed, then no element in \PawnSet\ has the value \pair{\cReg}{s}, where $s \in \N$, and thus the special value \pair{\bot}{\bot } is returned.
We then say that no process was promoted.
The returned pair is stored in local variables \pair{j}{seq} in \Line{promote:FR12}.

If no process was promoted, then $p$ resets \PawnSet\ to its initial value in \Line{promote:resetBackpack} using the \reset{} operation, and decreases \ctr\ from $2$ to $0$ in \Line{promote:ctr20}.
If a process was found and promoted in \PawnSet, then that process is notified of its promotion, by swapping its corresponding \Apply\ array element's value from $\cWant$ to $\cOk$ using a \CAS operation in \Line{promote:ApplyWantOk}.

Recall that, while executing a \lock{} method call a process may receive a signal to abort.
Suppose a process $p$ receives a signal to abort while executing a \lock{i} method call.
If process $p$ is busy-waiting in lines~\ref{getLock:ApplyBotWant},~\ref{getLock:awaitAckOrCtrDecrease} or~\ref{getLock:awaitX}, then $p$ stops executing \lock{i}, and instead executes a call to \abort{i}.
If $p$ is poised to execute any line~\ref{getLock:ApplyWantOk} or~\ref{getLock:end} then it completes its call to \lock{i}.
If $p$ is poised to execute any other line then it continues executing \lock{i} until it begins to busy-wait in lines~\ref{getLock:ApplyBotWant},~\ref{getLock:awaitAckOrCtrDecrease} or~\ref{getLock:awaitX}, at which point it stops and calls \abort{i}.
If $p$ does not begin to busy-wait in lines~\ref{getLock:ApplyBotWant},~\ref{getLock:awaitAckOrCtrDecrease} or~\ref{getLock:awaitX} then it completes its \lock{i} call.

\bparagraph{The \abort{} method}
Suppose $p$ executes a call to \abort{i}.
Process $p$ first determines whether it quit \lock{i} while busy-waiting on \Apply[i] in \Line{getLock:ApplyBotWant}, and if so, $p$ returns $\bot$ in \Line{abort:returnbot}.
If not, then $p$ changes \Apply[i] to the value $\cOk$ in \Line{abort:ApplyWantOk}, to prevent itself from getting collected in future collect actions.
In \Line{abort:ifPawn}, process $p$ determines whether it quit \lock{i} while busy-waiting on \Apply[i] in line~\ref{getLock:awaitAckOrCtrDecrease} or~\ref{getLock:awaitX}, or  while busy-waiting on \X\ in \Line{getLock:awaitX}.
If $p$ quit while busy-waiting on \Apply[i] then clearly it is a pawn process, and if it quit while busy-waiting on \X\ then it is a queen process.

If process $p$ determines that it is a pawn then it attempts to remove itself from \PawnSet\ by executing a \PawnSet.\cUpdate{$i,s$} operation in \Line{abort:ifHead}, where $s$ was the sequence number returned in \Line{getLock:getSequenceNo}.
If $p$ has not been promoted yet, then the operation succeeds and $p$'s corresponding element in \PawnSet\ is changed to a value \pair{\cAbort}{s}, thus making sure that $p$ can not be collected or promoted anymore.
If $p$ has already been promoted then the operation fails and $p$ now knows that it is has been promoted, and assumes the role of a promoted pawn, and in \Line{abort:RolePPawn}, $p$ writes $\cPPawn$ into \Role$[i]$ and returns the special value $\infty$ in \Line{abort:returninfty}.

If process $p$ determines that it is \queen{\L} then it first attempts to swap a special value $\infty$ into \X\ in \Line{abort:setX} by executing a \X.\CAS{$\bot,\infty$} operation to indicate its desire to abort.
If $p$ is successful then $p$ has determined that it is the first (among \king{\L} and itself) to exit, and therefore $p$ performs the collect action by calling \doCollect{i} in \Line{abort:doCollect}.
Process $p$ then makes a call to \helpRelease{i} in \Line{abort:callhRelease} to help release lock \L\ by synchronizing with \king{\L}.

If $p$ was unsuccessful at swapping value $\infty$ into \X\ then it knows the \king{\L} is executing \release{}{}, and \king{\L} will eventually perform the collect action.
Then $p$ has determined that it is the current owner of lock \L, and returns the integer value stored in \X\ in \Line{abort:returnX}.

Process $p$ executes \Line{abort:ApplyOkBot} only if $p$ successfully aborted earlier in its \abort{} call, and thus it deregisters itself from lock \L\ by swapping \pair{\bot}{\bot} into \Apply[i].
Finally, in \Line{abort:returnr}, $p$ returns $\bot$ to indicate a successful abort (i.e., a failed \lock{} call).

\subsection{Analysis and Proofs of Correctness}\label{sec:Analysis}

Let $H$ be an arbitrary history of an algorithm that accesses an instance, \L, of object \ARMLockArray{n}, where the following safety conditions hold.

\begin{condition}
\label{cond:safety:ARLockArray}
\begin{enumerate}[(a)]
  \item No two \lock{i} calls are executed concurrently for the same $i$, where $i \in \Set{0,\ldots,n-1}$.
   \item If a process $p$ executes a successful \lock{i} call, then some process $q$ eventually executes a \release{i}{} call where the invocation of \release{i}{} happens after the response of \lock{i} (assuming the scheduler is such that $q$ continues to make progress until its \release{i}{} call happens).\label{condition:ifLockThenRelease}
  \item For every \release{i}{} call, there must exist a unique successful \lock{i} call that completed before the invocation of the \release{i}{} call.\label{condition:ifReleaseThenExistsLock}

\end{enumerate}
\end{condition}

Then the following claims hold for history $H$.

%
%



\begin{lemma} \label{cl:MethodsWaitfree}
Methods \release{i}{j}, \abort{i}, \helpRelease{i}, \doCollect{i}, \doPromote{i} are wait-free.
\end{lemma}
\begin{proof}
Follows from an inspection of these methods.
\end{proof}

\begin{claim} 
No two \release{i}{} calls where a shared memory step is pending, are executed concurrently for the same $i$, where $i \in \Set{0,\ldots,n-1}$.
\label{claim:releaseNotConcurrent}
\end{claim}
\begin{proof}
Assume for the purpose of a contradiction that two processes are executing a call to \release{i}{} concurrently for the first time at time $t$.
Then from Condition~\ref{cond:safety:ARLockArray}(b)-(c), it follows that two successful calls to \lock{i}{} were executed before $t$.
From condition~\ref{cond:safety:ARLockArray}(a) it follows that the two successful \lock{i} calls did not overlap.
Consider the first successful \lock{i} call executed by some process $p$.
Since the \lock{i} call returned a non-$\bot$ value, the method did not return from line~\ref{abort:returnbot}.
Then $p$ did not abort while busy-waiting in \Line{getLock:ApplyBotWant}, and thus \apply$[i]$ was set to a non-$\pair{\bot}{\bot}$ value in \Line{getLock:ApplyBotWant} during the first \lock{i} call.
Let $t'$ be the point in time when \apply$[i]$ was set to a non-$\pair{\bot}{\bot}$ value in \Line{getLock:ApplyBotWant}.
We now show that the $\apply[i] \neq \pair{\bot}{\bot}$ in the duration between $[t',t]$.
Suppose not, i.e., some process resets \apply$[i]$ to \pair{\bot}{\bot} during $[t',t]$.
Now, \apply$[i]$ is reset to a \pair{\bot}{\bot} value only in \Line{abort:ApplyOkBot} during \abort{i} or in \Line{release:ApplyOkBot} during \release{i}{}.

\textbf{Case a - } \apply$[i]$ reset to \pair{\bot}{\bot} in \Line{release:ApplyOkBot} during \release{i}{}.
Then the last shared memory step of the \release{i}{} has been executed, and the call has ended for the purposes of the claim.
Then the two \release{i}{} calls are not concurrent at $t$, a contradiction.

\textbf{Case b - } \apply$[i]$ reset to \pair{\bot}{\bot} in \Line{abort:ApplyOkBot} during \abort{i}.
Since the two \lock{i} calls are not concurrent it follows that $\apply[i] \neq \pair{\bot}{\bot}$ at the end of the first \lock{i} call, and thus \apply$[i]$ is reset to \pair{\bot}{\bot} in \Line{abort:ApplyOkBot} during the second successful \lock{i} call.
Now consider the second successful \lock{i} call executed by some process $q$.
Then $q$ would repeatedly fail the \apply$[i]$.\CAS{$\pair{\bot}{\bot},\cdot$} operation of \Line{getLock:ApplyBotWant}, and the only way $q$'s \lock{i} call could finish, is if $q$ aborts the busy-wait loop of \Line{getLock:ApplyBotWant}.
In which case $q$ executes \abort{i}, and satisfies the if-condition of \Line{abort:ifFlag} and return $\bot$ in \Line{abort:returnbot}.
Then the second \lock{i} does not reset \apply$[i]$ in \Line{abort:ApplyOkBot} during \abort{i} -- a contradiction.

Since $\apply[i] \neq \pair{\bot}{\bot}$ throughout $[t',t]$, it then follows from the same argument of \textbf{Case b}, that the second \lock{i} call is unsuccessful, and thus a contradiction.
\end{proof}

From Claim~\ref{claim:releaseNotConcurrent} and Condition~\ref{cond:safety:ARLockArray}(a) it follows that no two calls to \lock{p} or \release{p}{} are executed concurrently for the same $p$, where $p \in \Set{0,\ldots,n-1}$.
Then we can label the process executing a \lock{p} or \release{p}{} call, simply $p$, without loss of generality.
We do so to make the rest of the proofs easier to follow.


\bparagraph{Helpful claims based on variable usage}

\begin{claim} \label{cl:basic:Role}
  \begin{enumerate}[(a)]
   \item   $\Role[p]$ is changed by process $q$, only if $q=p$. \label{scl:RoleOnlyp}
  \item   $\Role[p]$ is unchanged during \release{p}{}. \label{scl:RoleUnchanged}
 \item $\Role[p]$ can be set to value \cKing, \cQueen\ or \cPawn\ only when $p$ executes \Line{getLock:IncCounter} during \lock{p}. \label{scl:RoleIncCounter}
  \item $\Role[p]$ is set to value \cPPawn\ only when $p$ executes \Line{getLock:RolePPawn} during \lock{p} or when $p$ executes \Line{abort:RolePPawn} during \abort{p}.    \label{scl:RolePPawn}
  \end{enumerate}
\end{claim}
\begin{proof}
 All claims follow from an inspection of the code. 
\end{proof}

\begin{claim} \label{cl:basic:Collect}
  \begin{enumerate}[(a)]
  \item The only operations on \PawnSet\ are \collect{$A$}, \promote{}, \remove{$i$}, \remove{$j$}, \cUpdate{$k,s$} and \reset{} (in lines~\ref{collect:updateAll},~\ref{promote:FR12},~\ref{promote:collectFixSelf},~\ref{hRelease:collectFixOther},~\ref{abort:ifHead} and~\ref{promote:resetBackpack}, respectively) where $A$ is a vector with values in $\Set{\bot} \cup \N$, and $i,j,k \in \Set{0,1,\ldots,n-1}$, and $s \in \N$.\label{scl:CollectOps}
  \item  The $i$-th entry of \PawnSet\ can be changed to $\pair{\cReg}{s} = \pair{1}{s}$,  where $s \in \N$, only when a process executes a \PawnSet.\collect{$A$} operation in \Line{collect:updateAll} where $A[i] = s$.\label{scl:CollectRegister}
  \item  The $i$-th entry of \PawnSet\ can be changed to $\pair{\cPro}{s} = \pair{2}{s}$, where $s \in \N$, only when a process executes a \PawnSet.\promote{} operation in \Line{promote:FR12}.  \label{scl:CollectPromote}
  \item  The $i$-th entry of \PawnSet\ can be changed to $\pair{\cAbort}{s} = \pair{3}{s}$, where $s \in \N$, only when a process executes a \PawnSet.\remove{$i$}, \PawnSet.\remove{$j$} or \PawnSet.\cUpdate{$k,s$} operation in lines ~\ref{promote:collectFixSelf},~\ref{hRelease:collectFixOther} or~\ref{abort:ifHead}, respectively. \label{scl:CollectDeregister}
  \end{enumerate}
\end{claim}
\begin{proof}
Part~\refC{scl:CollectOps} follows from an inspection of the code.
Parts~\refC{scl:CollectRegister},~\refC{scl:CollectPromote} and~\refC{scl:CollectDeregister} follow from Part~\refC{scl:CollectOps} and the semantics of type \APArray{n}.
\end{proof}

\begin{claim} \label{cl:basic:Apply}
Let $s \in \N$.
\begin{enumerate}[(a)]
  \item \Apply$[p]$ is changed from $\pair{\bot}{\bot}$ to a non-$\pair{\bot}{s}$ value only when process $p$ executes a successful \Apply$[p]$.\CAS{$\pair{\bot}{\bot},\pair{\cWant}{s}$} operation in \Line{getLock:ApplyBotWant}. \label{scl:ApplyRegister}
  \item \Apply$[p]$ is changed to value $\pair{\cWant}{s}$ only when process $p$ executes a successful \Apply$[p]$.\CAS{$\pair{\bot}{\bot},\pair{\cWant}{s}$} operation in \Line{getLock:ApplyBotWant}. \label{scl:ApplyRegister2}
  \item \Apply$[p]$ is changed to a $\pair{\bot}{\bot}$ value only when $p$ executes a successful \Apply$[p]$.\CAS{$\pair{\cOk}{s},\pair{\bot}{\bot}$} operation either in \Line{abort:ApplyOkBot} or \Line{release:ApplyOkBot}. \label{scl:ApplyDeregister}
\end{enumerate}
\end{claim}

\begin{proof}
Parts~\refC{scl:ApplyRegister},~\refC{scl:ApplyRegister2} and~\refC{scl:ApplyDeregister} follow from an inspection of the code.
\end{proof}



\bparagraph{Helpful Notations and Definitions}
We now establish a notion of time for our history $H$.
Let the $i$-th step in $H$ occur at time $i$.
Then every point in time during $H$ is in $\mathbb{N}$.

Let $t_p^{i}$ denote the point in time immediately after process $p$ has finished executing line $i$, and no process has taken a step since $p$ has executed the last operation of line $i$ (This operation can be the response of a method call made in line $i$).
Since some private methods are invoked from more than one place in the code, the point in time $t_p^{i}$, where $i$ is a line in the method, does not refer to a unique point in time in history $H$.
In those cases we make sure that it is clear from the context of the discussion, which point $t_p^{i}$ refers to.
Let $t_p^{i-}$ denote the point in time when $p$ is poised to execute line $i$, and no other process takes steps before $p$ executes line $i$.

Let $p$ be an arbitrary process and $s$ be an arbitrary integer.
We say process $p$ \emph{registers}, when it executes a successful \Apply$[p]$.\CAS{$\pair{\bot}{\bot},\pair{\cWant}{s}$} operation in \Line{getLock:ApplyWantOk}.
Process $p$ \emph{captures} and \emph{wins} lock \L\ when it returns from \lock{p} with a non-$\bot$ value.
Process $p$ is said to \emph{promote} another process $q$ if $p$ executes a \PawnSet.\promote{} operation in \Line{promote:FR12} that returns a value $\pair{q}{s}$, where $s \in \N$.
A process $p$ is said to be \emph{promoted} at lock $\L$, if some process $q$ executes a \PawnSet.\promote{} operation that returns value \pair{p}{s}, where $s \in \N$.

Process $p$ is said to \emph{hand over} lock \L\ to process $q$ if it executes a successful \CAS operation \L.\X.\CAS{$\bot,j$} in \Line{release:setX}, where $q$ is the process that last increased $\ctr$ from $1$ to $2$.
Process $p$ is said to have \emph{released} lock \L\ by executing a successful \ctr.\CAS{$1,0$} operation in \Line{release:ctr10}, or by executing a successful \ctr.\CAS{$2,0$} operation in \Line{promote:ctr20}.
Process $p$ either hands over, promotes a process, or releases lock \L\ during a call to \L.\release{p}{j} where $j$ is an arbitrary integer.
A process \emph{ceases to own} a lock either by releasing lock \L\ or by promoting another process, or by handing over lock \L\ to some other process.
Process $p$ is \emph{deregistered} when $p$ executes a successful \Apply$[p]$.\CAS{$\pair{\cOk}{s},\pair{\bot}{\bot}$} operation in \Line{abort:ApplyOkBot} or~\ref{release:ApplyOkBot}.
A process $p$ is said to be \emph{not registered in \PawnSet} if the $p$-th entry of \PawnSet\ is not value $\pair{\cReg}{s}$, where $s \in \N$.
The repeat-until loop starting at \lref{getLock:BeginInnerLoop} and ending at \lref{getLock:EndInnerLoop} is called \emph{role-loop}.

In some of the proofs we use represent an execution using diagrams, and the legend for the symbols used in the diagrams is given in Figure~\ref{fig:Legend}.

\begin{figure*}[!htbp]
\centering
\includegraphics[scale=0.8]{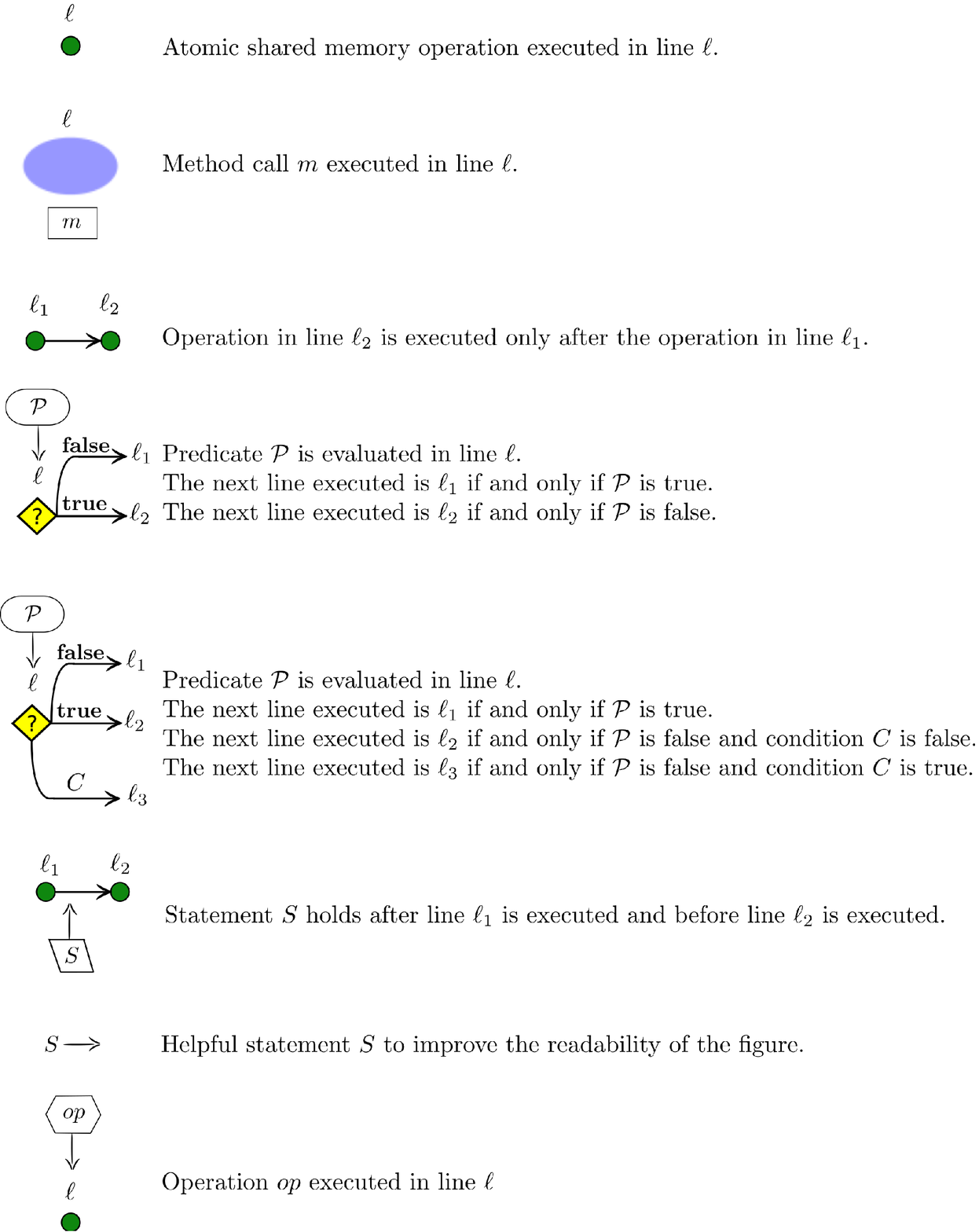}
\caption{Legend for Figures~\ref{fig:KsGetLock} to \ref{fig:PisRelease}}
\label{fig:Legend}
\end{figure*}




\bparagraph{Releasers of lock and Cease-release events}



A process $p$ becomes a \emph{releaser} of lock \L\ at time $t$ when
\begin{enumerate}[(R1)]
  \item $p$ increases \ctr\ to $1$ (i.e., \ctr.\inc{} returns $0 = \cKing$) or $2$ (i.e., \ctr.\inc{} returns $1 = \cQueen$), or when
  \item $p$ is promoted at lock \L\ by some process $q$ .
\end{enumerate}

\begin{claim}
\begin{enumerate}[(a)]
  \item $p$ executes a \ctr.\CAS{$1,0$} operation only in \Line{release:ctr10} during \release{p}{j}.\label{scl:def:phi}
  \item $p$ executes a \LSync.\CAS{$\bot,p$} operation only in \Line{hRelease:setT} during $p$'s call to \helpRelease{p}.\label{scl:def:tau}
  \item $p$ executes a \PawnSet.\promote{} operation only in \Line{promote:FR12} during $p$'s call to \doPromote{p}.\label{scl:def:pi}
  \item $p$ executes a  \ctr.\CAS{$2,0$} operation only in \Line{promote:ctr20} during $p$'s call to \doPromote{p}.\label{scl:def:theta}
\end{enumerate}
\label{cl:basic:ReleaseEvents}
\end{claim}
\begin{proof}
 All claims follows from an inspection of the code.
\end{proof}

We now define the following \emph{cease-release} events with respect to $p$ :
\begin{enumerate}
 \item[$\phi_p$:] $p$ executes a successful \ctr.\CAS{$1,0$} (at \pt{p}{release:ctr10} during \release{p}{j}).
 \item[$\tau_p$:] $p$ executes a successful \LSync.\CAS{$\bot,p$} (at \pt{p}{hRelease:setT} during \helpRelease{p}).
 \item[$\pi_p$:] $p$ promotes some process $q$ (at \pt{p}{promote:FR12}
during \doPromote{p}).
 \item[$\theta_p$:] $p$ executes an operation \ctr.\CAS{$2,0$} (at \pt{p}{promote:ctr20} during \doPromote{p}).
\end{enumerate}

Process $p$ \emph{ceases} to be a releaser of lock \L\ when one of $p$'s cease-release events occurs.
We say process $p$ is a releaser of lock \L\ at any point after it becomes a releaser and before it ceases to be a releaser.

\begin{claim}
\begin{enumerate}[(a)]
 \item Method \doCollect{p} is called only by process $p$ in lines~\ref{abort:doCollect} and~\ref{release:doCollect}. \label{scl:doCollect}
 \item Method \helpRelease{p} is called only by process $p$ in lines~\ref{release:callhRelease:King},~\ref{release:callhRelease:Queen} and ~\ref{abort:callhRelease}. \label{scl:hRelease}
 \item Method \doPromote{p} is called only by process $p$ in \Line{release:callPromote} and in \Line{hRelease:callPromote} (during \helpRelease{p}). \label{scl:promote}
 \item If cease-release event $\phi_p$ occurs then $p$ is executing \release{p}{j}. \label{scl:phi}
 \item If cease-release event $\tau_p$ occurs then $p$ is executing \helpRelease{p}. \label{scl:tau}
 \item If cease-release event $\pi_p$ or $\theta_p$ occurs then $p$ is executing \helpRelease{p} or \doPromote{p}. \label{scl:piOrTheta}
\end{enumerate}

\label{cl:basic}
\end{claim}

\begin{proof}
Parts~\refC{scl:doCollect},~\refC{scl:hRelease} and~\refC{scl:promote} follow from an inspection of the code.
By definition, cease-release event $\phi_p$ occurs when $p$ executes a successful \ctr.\CAS{$1,0$} operation in \Line{release:ctr10} during \release{p}{j}, and thus~\refC{scl:phi} follows immediately.
By definition, cease-release event $\tau_p$ occurs when $p$ executes a successful \LSync.\CAS{$\bot,p$} in \Line{hRelease:setT} during \helpRelease{p}, and thus~\refC{scl:tau} follows immediately.
By definition, cease-release event $\pi_p$ occurs only when $p$ executes a \PawnSet.\promote{} operation that returns a non-\pair{\bot}{\bot} value in \Line{promote:FR12}, and cease-release event $\theta_p$ occurs only when $p$ executes a \ctr.\CAS{$2,0$} operation in \Line{promote:ctr20}.
Then if cease-release event $\pi_p$ or $\theta_p$ occurs then $p$ is executing \doPromote{p}.
From~\refC{scl:promote}, $p$ could also call \doPromote{p} from \Line{hRelease:callPromote} during \helpRelease{p}.
Then if cease-release event $\pi_p$ or $\theta_p$ occurs then $p$ is executing \doPromote{p} or \helpRelease{p}.
Thus,~\refC{scl:piOrTheta} holds.
\end{proof}

\begin{claim} \label{cl:helpful:promoted}
Consider $p$'s $k$-th passage, where $k \in \N$. Note that $s=k$.
If $\Role[p] = \cPPawn$ at some point in time $t$ during $p$'s call to \lock{p}, then some process $q$ promoted $p$ at $\pt{q}{promote:FR12}$ and $p$ became releaser of \L\ by condition (R2) at $\pt{q}{promote:FR12} < t$.
\end{claim}
\begin{proof}
From \Claim{cl:basic:Role}\refC{scl:RolePPawn}, $p$ changes $\Role[p]$ to \cPPawn\ only in \Line{getLock:RolePPawn} or \Line{abort:RolePPawn}.

\textbf{Case a - } $p$ changed $\Role[p]$ to \cPPawn\ in \Line{abort:RolePPawn}:
Then $p$'s call to \PawnSet.\cUpdate{$p,s$} returned \false\ in \Line{abort:ifHead}.
From the semantics of the \APArray{n} object, it follows that the $p$-th entry of \PawnSet\ was set to value \pair{\cPro}{s} = \pair{2}{s}.
From \Claim{cl:basic:Collect}\refC{scl:CollectPromote}, the $p$-th entry of \PawnSet\ is set to value \pair{\cPro}{s} only when a \PawnSet.\promote{} operation returns $\pair{p}{s}$ in \Line{promote:FR12}.
Then some process $q$ promoted $p$ at \pt{q}{promote:FR12} and $p$ became a releaser of \L\ by condition (R2) at $\pt{q}{promote:FR12} < t$.

\textbf{Case b - }  $p$ changed $\Role[p]$ to \cPPawn\ in \Line{getLock:RolePPawn}:
Then $p$ broke out of the spin loop of \Line{getLock:ApplyBotWant}, and thus $\Apply[p] = \pair{\cWant}{s} \neq \pair{\cOk}{s}$ at \pt{p}{getLock:ApplyBotWant}.
Since $p$ satisfied the if-condition of \Line{getLock:ifBackpacked}, it follows that $\Apply[p] = \pair{\cOk}{s}$ at \pt{p}{getLock:ifBackpacked}.
Since $p$ does not change \Apply$[p]$ to value \pair{\cOk}{s} during $[\pt{p}{getLock:ApplyBotWant},\pt{p}{getLock:ifBackpacked}]$ it follows that some other process changed \Apply$[p]$ to value \pair{\cOk}{s}.
Now, \Apply$[p]$ is changed to value \pair{\cOk}{s} by some other process (say $q$) only in \Line{promote:ApplyWantOk} and thus, from the code structure, $q$ also executed a \PawnSet.\promote{} operation that returned \pair{p}{s} in \Line{promote:FR12}.
Then  $q$ promoted $p$ at \pt{q}{promote:FR12} and $p$ became a releaser of \L\ by condition (R2) at $\pt{q}{promote:FR12} < t$.
\end{proof}

\begin{claim}
\label{cl:eventE1HasNotOccured}
Consider $p$'s $k$-th passage, where $k \in \N$.
If $t \in \{$ $[\ptB{p}{abort:ifFlag},\pt{p}{abort:returnr}]$, $[\ptB{p}{release:setX},\pt{p}{release:callhRelease:King}]$, $[\ptB{p}{release:callhRelease:Queen},\pt{p}{release:callhRelease:Queen}]$, $[\ptB{p}{release:callPromote},\pt{p}{release:callPromote}]$ $\}$, then cease-release  event $\phi_p$ does not occur before time $t$.
\end{claim}
\begin{proof}
By definition, cease-release event $\phi_p$ occurs when $p$ executes a successful \ctr.\CAS{$1,0$} operation in \Line{release:ctr10}.
From \Claim{cl:basic}\refC{scl:phi} cease-release event $\phi_p$ occurs only during \release{p}{j}.

\textbf{Case a - } $t \in [\ptB{p}{abort:ifFlag},\pt{p}{abort:returnr}]$:
Then $p$ is executing \abort{p} and has not yet executed a call to \release{p}{}.
Since cease-release event $\phi_p$ can occur only during \release{p}{}, cease-release event $\phi_p$ did not occur before time $t$.

\textbf{Case b - } $t \in [\ptB{p}{release:setX},\pt{p}{release:callhRelease:King}]$:
Then $p$ must have failed the if-condition of \Line{release:ctr10}, and thus $p$ executed an unsuccessful \ctr.\CAS{$1,0$} operation in \Line{release:ctr10}, and cease-release event $\phi_p$ did not occur before time $t$.

\textbf{Case c - } $t \in \{$ $[\ptB{p}{release:callhRelease:Queen},\pt{p}{release:callhRelease:Queen}]$, $[\ptB{p}{release:callPromote},\pt{p}{release:callPromote}]$ $\}$:
From \Claim{cl:TableOfRoles}, $\Role[p] \in \Set{\cQueen,\cPPawn}$ at $t$.
Since $\Role[p]$ is unchanged during \release{p}{} (\Claim{cl:basic:Role}\refC{scl:RoleUnchanged}), it follows that $\Role[p] \neq \cKing$ at \ptB{p}{release:ifKing}.
Then $p$ fails the if-condition of \Line{release:ifKing}, and does not execute \Line{release:ctr10} and thus cease-release event $\phi_p$ did not occur before time $t$.
\end{proof}

The proof of the following claim has been moved to Appendix~\ref{sec:Appendix:remainingProofs} since the proof is long and straight forward.
\begin{claim}
The value of \Role$[p]$ at various points in time during $p$'s $k$-th passage, where $k \in \N$, is as follows.\

\label{cl:TableOfRoles}
\begin{tabular}{ l | l}
    \hline
    Time & Value of $\Role[p]$  \\ \hline
    $\pt{p}{getLock:IncCounter}$ & \Set{\bot,\cKing,\cQueen,\cPawn} \\ \hline
    $[\pt{p}{getLock:awaitAckOrCtrDecrease},\pt{p}{getLock:ifBackpacked}]$ &  \cPawn \\ \hline
    $\pt{p}{getLock:RolePPawn}$ &  \cPPawn\ \\ \hline
    $\ptB{p}{getLock:ifQueen}$ & \Set{\cKing,\cQueen,\cPPawn}  \\ \hline
    $\pt{p}{getLock:awaitX}$ & \cQueen \\ \hline
    $[\pt{p}{getLock:ApplyWantOk},\pt{p}{getLock:ifRolePQueen}]$ & \Set{\cKing,\cQueen,\cPPawn} \\ \hline
    $[\pt{p}{abort:ApplyWantOk},\ptB{p}{abort:ifPawn}]$ & \Set{\cQueen,\cPawn}   \\ \hline
    $\pt{p}{abort:ifHead}$ & \cPawn  \\ \hline
    $[\pt{p}{abort:RolePPawn},\pt{p}{abort:returninfty}]$ & \cPPawn\ \\ \hline
    $[\ptB{p}{abort:setX},\pt{p}{abort:callhRelease}]$ &  \cQueen \\ \hline
\end{tabular}
\begin{tabular}{ l | l}
    \hline
    Time & Value of $\Role[p]$  \\ \hline
    $[\ptB{p}{release:safetyCheck},\ptB{p}{release:ifKing}]$ & \Set{\cKing,\cQueen,\cPPawn} \\ \hline
    $[\ptB{p}{release:ctr10},\pt{p}{release:callhRelease:King}]$ & \cKing\ \\ \hline
    $\ptB{p}{release:callhRelease:Queen}$ & \cQueen\ \\ \hline
    $\ptB{p}{release:callPromote}$ & \cPPawn\ \\ \hline
    $[\ptB{p}{release:ApplyOkBot},\pt{p}{release:return}]$ &  $\Set{\cKing,\cQueen,\cPPawn}$ \\ \hline
    $[\ptB{p}{collect:collectLoop},\pt{p}{collect:updateAll}]$ &  \Set{\cKing,\cQueen} \\ \hline
    $[\ptB{p}{hRelease:setT},\pt{p}{hRelease:end}]$ &  \Set{\cKing,\cQueen} \\ \hline
    $[\ptB{p}{promote:FR12},\pt{p}{promote:end}]$ & \Set{\cKing,\cQueen,\cPPawn}  \\ \hline
\end{tabular}
\end{claim}

\begin{claim}   \label{cl:releaseMethodCalledAtMostOnce}
Consider $p$'s $k$-th passage, where $k \in \N$.
\begin{enumerate}[(a)]
 \item If process $p$ calls \helpRelease{p} or \doPromote{p} during \abort{p} then it does not call \release{p}{j}.
\label{scl:abortOrRelease}
 \item Process $p$ calls \helpRelease{p} at most once. \label{scl:hReleaseOnce}
 \item Process $p$ calls \doPromote{p} at most once. \label{scl:promoteOnce}
\end{enumerate}
\end{claim}
\begin{proof}

\textbf{Proof of~\refC{scl:abortOrRelease}: }
The following observations follow from an inspection of the code.
If $p$ executes \doPromote{p} during \abort{p}, then it does so during a call to \helpRelease{p} in \Line{hRelease:callPromote}.
If $p$ executes \helpRelease{p} during \abort{p}, then it does so by executing \Line{abort:callhRelease}.
Then $p$ calls \helpRelease{p} or \doPromote{p} during \abort{p} in \Line{abort:callhRelease} and goes on to return value $\bot$ in \Line{abort:returnr}.
Then $p$'s call to \lock{p} returns value $\bot$ and $p$ does not call \release{p}{} (follows from conditions~\ref{condition:ifLockThenRelease} and~\ref{condition:ifReleaseThenExistsLock}).

\textbf{Proof of~\refC{scl:hReleaseOnce}: }
From Part~\refC{scl:abortOrRelease}, if \helpRelease{p} is executed during \abort{p} then \release{p}{j} is not executed.
Then to prove our claim we need to show that \helpRelease{p} is called at most once during \abort{p} and \release{p}{j}, respectively.
From \Claim{cl:basic}\refC{scl:hRelease}, method \helpRelease{p} is called by $p$ only in lines~\ref{release:callhRelease:King},~\ref{release:callhRelease:Queen} and~\ref{abort:callhRelease}.
Since \helpRelease{p} is called only once during \abort{p} (specifically in \Line{abort:callhRelease}), it follows immediately that $p$ executes \helpRelease{p} at most once during \abort{p}.
From \Claim{cl:TableOfRoles}, $\Role[p] \in \Set{\cKing,\cQueen,\cPPawn}$ at \ptB{p}{release:safetyCheck}.
Since $\Role[p]$ is unchanged during \release{p}{} (\Claim{cl:basic:Role}\refC{scl:RoleUnchanged}), it follows that $p$ satisfies exactly one of the if-conditions of lines~\ref{release:ifKing},~\ref{release:ifQueen} and~\ref{release:ifPPawn}, and thus $p$ does not execute both lines~\ref{release:callhRelease:King} and~\ref{release:callhRelease:Queen}.
Then $p$ executes \helpRelease{p} at most once during \release{p}{j}.

\textbf{Proof of~\refC{scl:promoteOnce}: }
From Part~\refC{scl:abortOrRelease}, if \doPromote{p} is executed during \abort{p} then \release{p}{j} is not executed.
Then to prove our claim we need to show that \doPromote{p} is called at most once during \abort{p} and \release{p}{j}, respectively.
From \Claim{cl:basic}\refC{scl:promote}, method \doPromote{p} is called by $p$ only in \Line{release:callPromote} and in \Line{hRelease:callPromote} (during \helpRelease{p}).

\textbf{Case a - } $p$ called \doPromote{p} in \Line{hRelease:callPromote} (during \helpRelease{p}).
Then $p$ is executing \helpRelease{p}.
From \Claim{cl:basic}\refC{scl:hRelease}, method \helpRelease{p} is called by $p$ only in lines~\ref{release:callhRelease:King},~\ref{release:callhRelease:Queen} and~\ref{abort:callhRelease}.
Then $p$ called \helpRelease{p} either in line~\ref{release:callhRelease:King},~\ref{release:callhRelease:Queen} or~\ref{abort:callhRelease}

\textbf{Case a(i) - } $p$ called \helpRelease{p} in line~\ref{release:callhRelease:King} or~\ref{release:callhRelease:Queen} (during \release{p}{j}).
Then $p$ is executing \release{p}{}, and since $p$ called \helpRelease{p} in lines~\ref{release:callhRelease:King} or~\ref{release:callhRelease:Queen}, $p$ satisfied the if-conditions of lines~\ref{release:ifKing} or~\ref{release:ifQueen}, and thus $\Role[p] = \cKing$ at \ptB{p}{release:ifKing} or $\Role[p] = \cQueen$ at \ptB{p}{release:ifQueen}, respectively.
Since $\Role[p]$ is unchanged during \release{p}{} (\Claim{cl:basic:Role}\refC{scl:RoleUnchanged}), it follows that $\Role[p] \in \Set{\cKing,\cQueen}$ during \release{p}{}.
Then $p$ fails the if-condition of \Line{release:ifPPawn} and does not execute \doPromote{p} in \Line{release:callPromote}.
Hence, $p$ executes \doPromote{p} at most once during \release{p}{}.

\textbf{Case a(ii) - } $p$ called \helpRelease{p} in line~\ref{abort:callhRelease}.
Then $p$ is executing \abort{p} and it goes on to return value $\bot$ in \Line{abort:returnr}.
Then $p$'s call to \lock{p} returns value $\bot$ and $p$ does not call \release{p}{j} (follows from conditions~\ref{condition:ifLockThenRelease} and~\ref{condition:ifReleaseThenExistsLock}).
Hence, $p$ executes \doPromote{p} at most once during \abort{p}.

\textbf{Case b - } $p$ called \doPromote{p} in \Line{release:callPromote}.
Then $p$ is executing \helpRelease{p} and $p$ satisfied the if-condition of lines~\ref{release:ifPPawn}, and thus $\Role[p] = \cPPawn$ at \ptB{p}{release:ifPPawn}.
Since $\Role[p]$ is unchanged during \release{p}{} (\Claim{cl:basic:Role}\refC{scl:RoleUnchanged}), it follows that $\Role[p] = \cPPawn$ during \release{p}{}.
Then $p$ failed the if-condition of lines~\ref{release:ifKing} and~\ref{release:ifQueen} and $p$ did not execute \helpRelease{p} in lines~\ref{release:callhRelease:King} and~\ref{release:callhRelease:Queen}.
Hence, $p$ executes \doPromote{p} at most once during \release{p}{}.
\end{proof}

\begin{claim}
\label{cl:ifReleasinglockThenReleaser}
Consider $p$'s $k$-th passage, where $k \in \N$.
Let $t$ be a point in time at which either $p$ is poised to execute \release{p}{j},
or $t \in \{$ $[\ptB{p}{abort:setX},\pt{p}{abort:doCollect}]$, $[\ptB{p}{release:setX},\pt{p}{release:doCollect}]$, $[\ptB{p}{collect:collectLoop},\pt{p}{collect:updateAll}]$ , \ptB{p}{hRelease:setT},$ [\ptB{p}{hRelease:readX},\ptB{p}{hRelease:callPromote}]$, \ptB{p}{promote:FR12}, $[\ptB{p}{promote:resetBackpack},\ptB{p}{promote:ctr20}]$ $\}$.
Then 
\begin{enumerate}[(a)]
 \item none of $p$'s cease-release events have occurred before time $t$, and  \label{scl:releaseEventsDidNotOccur}
 \item $p$ is a releaser of lock \L\ at time $t$ \label{scl:isReleaser}
\end{enumerate}
\end{claim}

\begin{proof}
\textbf{Proof of~\refC{scl:releaseEventsDidNotOccur}: }
First note that if $t \in [\ptB{p}{collect:collectLoop},\pt{p}{collect:updateAll}]$ then $p$ is executing \doCollect{}.
From \Claim{cl:basic}\refC{scl:doCollect}, $p$ calls \doCollect{} only in lines~\ref{abort:doCollect} and~\ref{release:doCollect}.
Then if $t \in [\ptB{p}{collect:collectLoop},\pt{p}{collect:updateAll}]$ then $t \in [\ptB{p}{abort:doCollect},\pt{p}{abort:doCollect}]$ or $t \in [\ptB{p}{release:doCollect},\pt{p}{release:doCollect}]$.
Therefore, assume now $t \in [\ptB{p}{abort:setX},\pt{p}{abort:doCollect}]$ or $t \in [\ptB{p}{release:setX},\pt{p}{release:doCollect}]$.

\textbf{Case a - } $t \in \Set{[\ptB{p}{abort:setX},\pt{p}{abort:doCollect}],[\ptB{p}{release:setX},\pt{p}{release:doCollect}]}$:
If $t \in [\ptB{p}{abort:setX},\pt{p}{abort:doCollect}]$ then from a code inspection, $p$ is executing \abort{p} and $p$ did not execute a call to \doPromote{p} or \helpRelease{p} before time $t$.
If $t \in [\ptB{p}{release:setX},\pt{p}{release:doCollect}]$ then $p$ is executing \release{p}{j} and then from a code inspection and \Claim{cl:releaseMethodCalledAtMostOnce}\refC{scl:abortOrRelease} it follows that $p$ did not execute a call to \doPromote{p} or \helpRelease{p} before time $t$.
Then from Claims~\ref{cl:basic}\refC{scl:tau} and~\ref{cl:basic}\refC{scl:piOrTheta} it follows that events $\tau_p$, $\pi_p$ and $\theta_p$ did not occur before time $t$.
Since $t \in [\ptB{p}{abort:setX},\pt{p}{abort:doCollect}]$ or $t \in [\ptB{p}{release:setX},\pt{p}{release:doCollect}]$, it follows from \Claim{cl:eventE1HasNotOccured} that cease-release event $\phi_p$ did not occur before time $t$.

\textbf{Case b - } $t \in \{$ \ptB{p}{hRelease:setT}, $[\ptB{p}{hRelease:readX},\ptB{p}{hRelease:callPromote}]$ $\}$:
Then $p$ is executing \helpRelease{p}.
Then from Claim~\ref{cl:basic}\refC{scl:hRelease} it follows that $p$ is executing a call to \helpRelease{p} in line~\ref{release:callhRelease:King},~\ref{release:callhRelease:Queen} or~\ref{abort:callhRelease}.
Then from \Claim{cl:eventE1HasNotOccured} it follows that cease-release event $\phi_p$ did not occur before time $t$.
From \Claim{cl:releaseMethodCalledAtMostOnce}\refC{scl:hReleaseOnce}, it follows that this is $p$'s only call to \helpRelease{p}.
From Claim~\ref{cl:basic}\refC{scl:promote}, $p$ calls \doPromote{p} only in \Line{release:callPromote} and in \Line{hRelease:callPromote} (during \helpRelease{p}).
Since $p$ has not yet executed \Line{release:callPromote} and this is the only call to \helpRelease{p}, $p$ has not called \doPromote{p} before time $t$.
Then from Claim~\ref{cl:basic}\refC{scl:piOrTheta} it follows that events $\pi_p$ and $\theta_p$ did not occur before time $t$.
By definition, cease-release event $\tau_p$ occurs when $p$ executes a successful \LSync.\CAS{$\bot,p$} in \Line{hRelease:setT}.
If $t = \ptB{p}{hRelease:setT}$, then clearly cease-release event $\tau_p$ did not occur before time $t$.
If $t \in [\ptB{p}{hRelease:readX},\ptB{p}{hRelease:callPromote}]$, then $p$ satisfied the if-condition of \Line{hRelease:setT}, and thus $p$ executed an unsuccessful \LSync.\CAS{$\bot,p$} operation in \Line{hRelease:setT}, and thus cease-release event $\tau_p$ did not occur before time $t$.

\textbf{Case c - } $t \in \{$ \ptB{p}{promote:FR12}, $[\ptB{p}{promote:resetBackpack},\ptB{p}{promote:ctr20}]$ $\}$:
Then $p$ is executing \doPromote{p}.
From \Claim{cl:releaseMethodCalledAtMostOnce}\refC{scl:promoteOnce}, it follows that this is the only call to \doPromote{p}.
By definition, cease-release event $\theta_p$ occurs only when $p$ executes a \ctr.\CAS{$2,0$} operation in \Line{promote:ctr20} of \doPromote{p}.
Event $\theta_p$ did not occur before time $t$ since $t < \pt{p}{promote:ctr20}$ and this is $p$'s only call to \doPromote{p}.
By definition, cease-release event $\pi_p$ occurs only when $p$ executes a \PawnSet.\promote{} operation that returns a non-\pair{\bot}{\bot} value in \Line{promote:FR12} of \doPromote{p}.
If $t = \ptB{p}{promote:FR12}$, then cease-release event $\pi_p$ did not occur before time $t$ since $\ptB{p}{promote:FR12} < \pt{p}{promote:FR12}$ (and since this is $p$'s only call to \doPromote{p}).
If $t \in [\ptB{p}{promote:resetBackpack},\ptB{p}{promote:ctr20}]$, then $p$ satisfied the if-condition of \Line{promote:FR12}, and thus $p$'s \PawnSet.\promote{} operation returned value \pair{\bot}{\bot}, and thus cease-release event $\pi_p$ did not occur before time $t$.
Since $p$ calls \doPromote{p} only in line~\ref{release:callPromote} and \Line{hRelease:callPromote} (during \helpRelease{p}), $p$ is executing line~\ref{release:callPromote},~\ref{release:callhRelease:King},~\ref{release:callhRelease:Queen} or~\ref{abort:callhRelease}.
Then from \Claim{cl:eventE1HasNotOccured} it follows that cease-release event $\phi_p$ did not occur before time $t$.

We now show that cease-release event $\tau_p$ did not occur before time $t$ thus completing the proof.

\textbf{Subcase c(i) - } $p$ called \doPromote{p} during \helpRelease{p}:
Then $p$ satisfied the if-condition of \Line{hRelease:setT}, and thus $p$ executed an unsuccessful \LSync.\CAS{$\bot,p$} operation in \Line{hRelease:setT}, and cease-release event $\tau_p$ did not occur before time $t$.

\textbf{Subcase c(ii) - } $p$ called \doPromote{p} in line~\ref{release:callPromote}:
From \Claim{cl:TableOfRoles}, $\Role[p] \in \cPPawn$ at $\ptB{p}{release:callPromote}$.
Since $\Role[p]$ is unchanged during \release{p}{} (\Claim{cl:basic:Role}\refC{scl:RoleUnchanged}), it follows that $\Role[p] = \cPPawn$ at \ptB{p}{release:ifKing} and \ptB{p}{release:ifQueen}.
Then $p$ fails the if-conditions of lines~\ref{release:ifKing} and~\ref{release:ifQueen}, and does not execute a call to \helpRelease{p} before time $t$.
Then from \Claim{cl:basic}\refC{scl:tau} it follows that cease-release event $\tau_p$ did not occur before time $t$.

\textbf{Proof of~\refC{scl:isReleaser}: }
From Part~\refC{scl:releaseEventsDidNotOccur}, $p$ does not cease to be releaser of \L\ before $t$.
Therefore, to prove our claim we need to show that $p$ becomes a releaser of \L\ at some point $t' < t$.
We first show that $\Role[p] \in \{$ \cKing, \cQueen, \cPPawn\ $\}$ at time $t$.
Let $t'$ be the point when $p$ is poised to execute \release{p}{j}.
From the inspection of the various points in time chosen for $t$ (including $\ptB{p}{release:safetyCheck}$, but excluding $t'$) and the table in \Claim{cl:TableOfRoles}, it follows that $\Role[p] \in \{$ \cKing, \cQueen, \cPPawn\ $\}$ at time $t$ (including $\ptB{p}{release:safetyCheck}$, but excluding $t'$).
Clearly $\Role[p]$ is unchanged during $[t',\ptB{p}{release:safetyCheck}]$.
Then the value of $\Role[p]$ at $t'$ is the same as that at $\ptB{p}{release:safetyCheck}$, i.e., $\Role[p] \in \{$ \cKing, \cQueen, \cPPawn\ $\}$.

\textbf{Case a - } $\Role[p] \in \Set{\cKing,\cQueen}$ at time $t$:
From \Claim{cl:basic:Role}\refC{scl:RoleIncCounter}, $\Role[p]$ is set to \cKing\ or \cQueen\ only when $p$ executes \Line{getLock:IncCounter}.
Then  $p$ changed $\Role[p]$ to $\cKing$ or $\cQueen$ at $\pt{p}{getLock:IncCounter}$, and thus $p$ became a releaser of lock \L\ by condition (R1) at $\pt{p}{getLock:IncCounter} = t' < t$.

\textbf{Case b - } $\Role[p] = \cPPawn$ at time $t$:
From \Claim{cl:helpful:promoted}, it follows that some process $q$ promoted $p$ at \pt{q}{promote:FR12} and $p$ became a releaser of \L\ by condition (R2) at $\pt{q}{promote:FR12} = t' < t$.
\end{proof}

%

\begin{claim}
\label{cl:releaserAtCeaseReleaseEvents}
Consider $p$'s $k$-th passage, where $k \in \N$.
If any of process $p$'s cease-release events occurs at time $t$ then $p$ ceases to be the releaser of lock \L\ at time $t$.
\end{claim}
\begin{proof}
To prove our claim we need to show that $p$ is a releaser of \L\ immediately before time $t$, since by definition $p$ ceases to be a releaser of \L\ when any of $p$'s cease-release events occurs.
By definition, cease-release event $\phi_p$ occurs when $p$ executes a successful \ctr.\CAS{$1,0$} operation in \Line{release:ctr10}, cease-release event $\tau_p$ occurs when $p$ executes a successful \LSync.\CAS{$\bot,p$} in \Line{hRelease:setT}, cease-release event $\pi_p$ occurs only when $p$ executes a \PawnSet.\promote{} operation that returns a non-\pair{\bot}{\bot} value in \Line{promote:FR12}, cease-release event $\theta_p$ occurs only when $p$ executes a \ctr.\CAS{$2,0$} operation in \Line{promote:ctr20}.
From \Claim{cl:ifReleasinglockThenReleaser}\refC{scl:isReleaser}, $p$ is a releaser of \L\ at \ptB{p}{release:ctr10}, \ptB{p}{hRelease:setT}, \ptB{p}{promote:FR12} and \ptB{p}{promote:ctr20}.
Hence, the claim follows.
\end{proof}

%


We say a process has \emph{write-access} to objects \X\ and \LSync, respectively, if the process can write a value to \X\ and \LSync, respectively.
We say a process has \emph{registration-access} to object \PawnSet, if the process can execute an operation on \PawnSet\ that can write values in \Set{\pair{a}{b} | a \in \Set{0,1,2} = \Set{0,\cReg,\cPro}, b \in \N} to some entry of \PawnSet.
We say a process has \emph{deregistration-access} to object \PawnSet, if the process can execute an operation on \PawnSet\ that can write value \pair{\cAbort}{s} = \pair{3}{s}, where $s \in \N$, to some entry of \PawnSet.
Object \PawnSet\ is said to be \emph{candidate-empty} if no entry of \PawnSet\ has value \pair{\cReg}{\cdot} or \pair{\cPro}{\cdot}.

\begin{claim}
\label{cl:variablesChangedByReleaserOnly}
Only releasers of \L\ have write-access to  \X, \LSync\ and registration-access to \PawnSet.
\end{claim}
\begin{proof}
The following observations follow from an inspection of the code.
A value can be written to \X\ only in lines~\ref{abort:setX},~\ref{release:setX} and~\ref{hRelease:resetX}.
A value can be written to \LSync\ only in lines~\ref{hRelease:setT} and~\ref{hRelease:resetT}.
From the semantics of the \APArray{n} object, only operations  \collect{}, \promote{}, and \reset{} can write values in \Set{\pair{a}{b} | a \in \Set{0,\cReg,\cPro} = \Set{0,1,2}, b \in \N}  to \PawnSet.
From \Claim{cl:basic:Collect}\refC{scl:CollectOps}, the operations \collect{}, \promote{},  and \reset{} are executed on \PawnSet\ only in lines~\ref{collect:updateAll},~\ref{promote:FR12}, and~\ref{promote:resetBackpack}, respectively.

Suppose an arbitrary process $p$ writes a value to \X\ or \LSync, or a value in \Set{\pair{a}{b} | a \in \Set{0,1,2}, b \in \N} to an entry of \PawnSet.
From \Claim{cl:ifReleasinglockThenReleaser}\refC{scl:isReleaser}, $p$ is a releaser of \L\ at \ptB{p}{abort:setX}, \ptB{p}{release:setX}, \ptB{p}{hRelease:resetX}, \ptB{p}{hRelease:setT}, \ptB{p}{hRelease:resetT}, \ptB{p}{promote:FR12}, \ptB{p}{promote:resetBackpack} and \ptB{p}{collect:updateAll}.
Hence, the claim follows.
\end{proof}

\begin{claim}
\label{cl:iThCollectChangedByPAndReleaser}
The $i$-entry of \PawnSet\ can be changed only by process $i$ or a releaser of \L.
\end{claim}
\begin{proof}
The values that can be written to \PawnSet\ are in \Set{\pair{a}{b} | a \in \Set{0,1,2,3}, b \in \N}.
A process that can write values in  \Set{\pair{a}{b} | a \in \Set{0,1,2}, b \in \N} to any entry of \PawnSet\ is said to have registration-access to \PawnSet.
From \Claim{cl:variablesChangedByReleaserOnly} it follows that only a releaser of \L\ has registration-access to \PawnSet, therefore only a releaser of \L\ can write values in \Set{\pair{a}{b} | a \in \Set{0,1,2}, b \in \N} to the $i$-th entry of \PawnSet.
From \Claim{cl:basic:Collect}\refC{scl:CollectDeregister} the value \pair{\cAbort}{s} = \pair{3}{s}, where $s \in \N$, can be written to the $i$-th entry of \PawnSet\ only when a process executes a \remove{$i$}, \remove{$i$} or \PawnSet.\cUpdate{$i,s$} operation in line~\ref{promote:collectFixSelf},~\ref{hRelease:collectFixOther} or~\ref{abort:ifHead}, respectively.
From \Claim{cl:ifReleasinglockThenReleaser}\refC{scl:isReleaser}, it follows that a process executing lines~\ref{promote:collectFixSelf} and ~\ref{hRelease:collectFixOther} is a releaser of \L.
Since a \PawnSet.\cUpdate{$i,s$} operation in \Line{abort:ifHead} is executed only by process $i$, our claim follows.
\end{proof}

\begin{claim}
\label{cl:TChangedByReleaserOnly}
\LSync\ is changed to a non-$\bot$ value only by a releaser of \L\ (say $r$) in \Line{hRelease:setT} which triggers the cease-release event $\tau_r$.
\end{claim}
\begin{proof}
By definition, cease-release event $\tau_p$ occurs when $p$ executes a successful \LSync.\CAS{$\bot,p$} in \Line{hRelease:setT}.
From a code inspection, \LSync\ is changed to a non-$\bot$ value only when some process (say $r$) executes a successful \LSync.\CAS{$\bot,r$} operation in \Line{hRelease:setT}.
From \Claim{cl:variablesChangedByReleaserOnly} it follows that $\LSync$ is changed only by a releaser of \L.
Then $r$ is a releaser of \L\ when it changes \LSync\ to a non-$\bot$ value in \Line{hRelease:setT} and doing so triggers the cease-release event $\tau_r$.
\end{proof}

\begin{claim}
\label{cl:PromoteOpByReleaserOnly}
A \PawnSet.\promote{} operation is executed only by a releaser of \L\ (say $r$), and if the value returned is non-\pair{\bot}{\bot} the cease-release event $\pi_r$ is triggered.
\end{claim}
\begin{proof}
By definition, cease-release event $\pi_p$ occurs only when $p$ executes a \PawnSet.\promote{} operation that returns a non-\pair{\bot}{\bot} value in \Line{promote:FR12}.
From a code inspection, a \PawnSet.\promote{} operation is executed only when some process (say $r$) executes \Line{hRelease:setT}.
From \Claim{cl:variablesChangedByReleaserOnly} it follows that $\PawnSet$ is changed only by a releaser of \L.
Then $r$ is a releaser of \L\ when it executes a \PawnSet.\promote{} operation, and if the operation returns a non-\pair{\bot}{\bot} value then the cease-release event $\pi_r$ is triggered.
\end{proof}

\begin{claim}
\label{cl:EventsPromote}
During an execution of \doPromote{p} exactly one of the events $\pi_{p}$ and $\theta_p$ occurs.
\end{claim}

\begin{proof}
By definition, cease-release event $\pi_p$ occurs when $p$ executes a \PawnSet.\promote{} operation in \Line{promote:FR12} that returns a non-\pair{\bot}{\bot} value, and cease-release event $\theta_p$ occurs when $p$ executes a \ctr.\CAS{$2,0$} operation in \Line{promote:ctr20} during \doPromote{p}.

\textbf{Case a - } the \PawnSet.\promote{} operation in \Line{promote:FR12} returns a non-\pair{\bot}{\bot} value, and thus cease-release event $\pi_p$ occurs:
Then $p$ fails the if-condition of \Line{promote:ifNotPromoted} and \Line{promote:ctr20} is not executed.
Therefore, cease-release event $\theta_p$ does not occur.

\textbf{Case b - } the \PawnSet.\promote{} operation in \Line{promote:FR12} returns \pair{\bot}{\bot}, and thus cease-release event $\pi_p$ does not occur:
Then $p$ satisfies the if-condition of \Line{promote:ifNotPromoted}, and executes a \ctr.\CAS{$2,0$} operation in \Line{promote:ctr20}.
Hence, cease-release event $\theta_p$ occurs.
\end{proof}

\begin{claim}
\label{cl:EventsHelpRelease}
During an execution of \helpRelease{p} exactly one of the events $\tau_{p},\pi_{p}$ and $\theta_p$ occurs.
\end{claim}

\begin{proof}
By \Claim{cl:basic:ReleaseEvents}, events $\pi_p$ and $\theta_p$ can only occur during $p$'s call to \doPromote{p}, and cease-release event $\tau_p$ occurs when $p$ executes a successful \LSync.\CAS{$\bot,p$} operation in \Line{hRelease:setT}.

\textbf{Case a - } $p$ executes a successful \LSync.\CAS{$\bot,p$} operation in \Line{hRelease:setT}, and thus cease-release event $\tau_p$ occurs:
Then  $p$ fails the if-condition of \Line{hRelease:setT}, and returns immediately from its call to \helpRelease{i}.
Therefore, events $\pi_{p}$ and $\theta_p$ do not occur.

\textbf{Case b - } $p$ executes an unsuccessful \LSync.\CAS{$\bot,p$} operation in \Line{hRelease:setT}, and thus cease-release event $\tau_p$ does not occur.
Then  $p$ satisfies the if-condition of \Line{hRelease:setT}, and calls \doPromote{p} in \Line{hRelease:callPromote}.
From \Claim{cl:EventsPromote}, exactly one of the events $\pi_p$ and $\theta_p$ occurs during $p$'s call to \doPromote{p}.
\end{proof}

\begin{claim}
 The value of \ctr\ can change only when a \ctr.\inc{}, \ctr.\CAS{$2,0$} or \ctr.\CAS{$1,0$} operation is executed in lines~\ref{getLock:IncCounter},~\ref{promote:ctr20} or~\ref{release:ctr10}.
\label{cl:basic:ctr}
\end{claim}
\begin{proof}
From the semantics of the \RCASCounter{2} object, if \ctr\ is increased to value $i$ by a \ctr.\inc{} operation, then its value was $i-1$ immediately before the operation was executed.
Then all claims follow from an inspection of the code.
\end{proof}

\begin{claim} \label{cl:ctrValues}
If the value of \ctr\ changes, it either increases by $1$ or decreases to $0$.
Moreover its values are in \Set{0,1,2}.
\end{claim}
\begin{proof}
From the semantics of the \RCASCounter{2} object, a \ctr.\inc{} operation changes the value of \ctr\  from $i$ to $i+1$ only if $i \in \Set{0,1}$.
From Claims~\ref{cl:basic:ctr}, the value of \ctr\ can change only when a \ctr.\inc{}, \ctr.\CAS{$2,0$} or \ctr.\CAS{$1,0$} operation is executed (in lines~\ref{getLock:IncCounter},~\ref{promote:ctr20} or~\ref{release:ctr10}).
Then it follows that the values of \ctr\ are in\Set{0,1,2}.
It also follows that the value of \ctr\ either changes from $0$ to $1$ and back to $0$, or it changes from $0$ to $1$ to $2$ and back to $0$.
\end{proof}

\bparagraph{\ctr-Cycle Interval $T$}
\label{sec:IntervalT}
Let $T = [t_s,t_e)$ be a time interval where $t_s$ is a point when \ctr\ is $0$ and $t_e$ is the next point in time when \ctr\ is decreased to $0$.
For $i \in \Set{0,1,2}$ let \I{i} = $\{ t \in T | \ctr = i$ at $t \}$ and let time \Ib{i} = \minimum{\I{i}} and time \Ia{i} = \maximum{\I{i}}.
From \Claim{cl:ctrValues}, it follows immediately that during $T$ the set $\I{i}, i \in \Set{0,1,2}$, forms an interval $[\Ib{i},\Ia{i}]$, and $\I{2} = \varnothing$ if and only if \ctr\ is never increased to $2$ during $T$.
Moreover, $t_s =  \Ib{0}$ and \I{0} \emph{is immediately followed by} \I{1} (i.e., \minimum{\I{1}} = \maximum{\I{0}} + 1).
If $\I{2} \neq \varnothing$ then \I{2} follows immediately after \I{1}.
The \ctr-cycle interval $T$ ends either at time \Ia{1} if $\I{2} = \varnothing$, or at time \Ia{2} if $\I{2} \neq \varnothing$.
 

Then it also follows that exactly one process changes \ctr\ from $0$ to $1$ during $T$, and it does so at time \Ib{1}.
Let \K\ be the process that increases \ctr\ to $1$ at time \Ib{1}.
And if $\I{2} \neq \varnothing$ then exactly one process changes \ctr\ from $1$ to $2$ during $T$, and it does so at time \Ib{2}.
If $\I{2} \neq \varnothing$ let \Q\ be the process that increases \ctr\ to $2$ at time \Ib{2}.
Let $R(t)$ denote the set of processes that are the releasers of lock \L\ at time $t \in T$.

\begin{claim} \label{cl:proofProperties1}
If $R(\Ib{0}) = \varnothing$ and at $\Ib{0}$, $\X = \LSync = \bot$ and \PawnSet\ is candidate-empty, then the following holds:

\begin{enumerate}[(a)]
    \item $\forall_{t \in \I{0}}: R(t) = \varnothing$ and throughout \I{0}, $\X = \LSync = \bot$ and \PawnSet\ is candidate-empty.\label{I0}
    \item $R(\Ib{1}) = \Set{\K}$ and at time \Ib{1}, $\X = \LSync = \bot$ and \PawnSet\ is candidate-empty.\label{@I1+}
    \item \K\ executes  lines of code of \lock{\K} starting with \Line{getLock:ApplyBotWant} as depicted in Figure~\ref{fig:KsGetLock}.\label{KsGetLock} (A legend for the figure is given in Figure~\ref{fig:Legend}.)

\begin{figure}[!htb]
\centering
\includegraphics[scale=0.6]{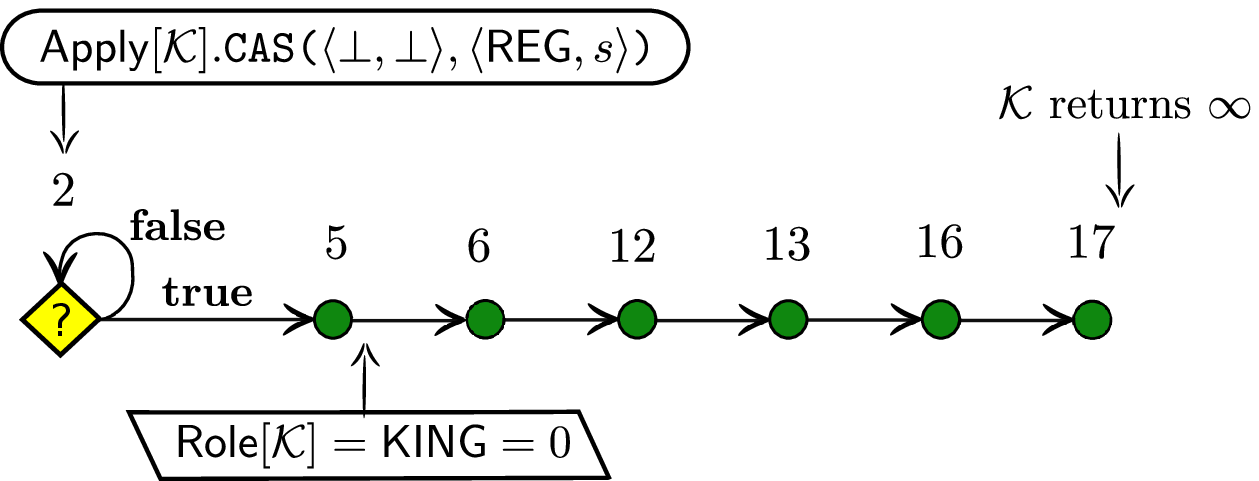}
\caption{\K's call to \lock{\K}}
\label{fig:KsGetLock}
\end{figure}

    \item \K's call to \lock{\K} returns $\infty$ and $\Role[\K] = \cKing$ throughout $[\pt{\K}{getLock:IncCounter},\pt{\K}{getLock:returninfty}]$. \label{KDoesNotStarve} \label{Role=King}
    \item \K\ executes a \ctr.\CAS{$1,0$} operation in \Line{release:ctr10} during $T$, and \K\ does not change $\X,\LSync$ or $\PawnSet$ throughout $[\Ib{1},\pt{\K}{release:ctr10}]$.
\label{aExecutesctr10}
    \item $\forall_{t \in \I{1}}: R(t) =\Set{\K}$. \label{R@I1}
    \item Throughout \I{1}, $\X = \LSync = \bot$ and \PawnSet\ is candidate-empty. \label{XT@I1}
\end{enumerate}
\end{claim}


\begin{proof}
\textbf{Proof of~\refC{I0}: }
Consider the claim $R(t) = \varnothing$ where $t \in \I{0}$.
Since $R(\Ib{0}) = \varnothing$ holds  by assumption, the claim holds at $t =  \Ib{0}$.
For the purpose of a contradiction assume the claim fails to hold for the first time at some point $t'$ during \I{0}.
Then some process $p$ becomes a releaser of lock \L\ at time $t'$.
Process $p$ cannot become a releaser of \L\ by $p$ increasing \ctr\ to $1$ or $2$ (condition (R1)) at time $t'$, since $\ctr = 0$ throughout \I{0}.
Therefore, assume it becomes a releaser of \L\ when some process $q$ promotes $p$ (condition (R2)) at $t'$.
By \Claim{cl:releaserAtCeaseReleaseEvents}, $q$ ceases to be a releaser of lock \L\ at $t'$.
This is a contradiction to our assumption that $p$ is the first process during \I{0} to become a releaser of \L.

By assumption the variables $\X, \LSync$ and \PawnSet\ are at their initial value at $\Ib{0}$.
Since the values of these variables are only changed by a releaser of lock \L\ (by \Claim{cl:variablesChangedByReleaserOnly}) and for all $t \in \I{0}$, $R(t) = \varnothing$, it follows that the variables are unchanged throughout \I{0}.

\textbf{Proof of~\refC{@I1+}: }
At time \Ib{1} \ctr\ is increased from $0$ to $1$, and thus the only operation executed is a \ctr.\inc{} operation  by process \K.
Then \K\ becomes a releaser of lock \L\ at time \Ib{1} by condition (R1).
Since for all $t \in \I{0}$, $R(t) = \varnothing$ (Part~\refC{I0}), it follows that $R(\Ib{1}) = \Set{\K}$.
Since $\X = \LSync = \bot$ and \PawnSet\ is candidate-empty throughout \I{0} (Part~\refC{I0}), and the only operation at time \Ib{1} is the \ctr.\inc{} operation, it follows that $\X = \LSync = \bot$ and \PawnSet\ is candidate-empty at time \Ib{1}  .

\textbf{Proof of~\refC{KsGetLock} and~\refC{KDoesNotStarve}: }
Since \K\ is the process that increased \ctr\ from $0$ to $1$ at time \Ib{1}, and since \K\ can increase \ctr\ only by executing a \ctr.\inc{} operation in \Line{getLock:IncCounter} (by Claim~\ref{cl:basic:ctr}) \K\ set $\Role[\K] =  0 = \cKing$  at \pt{\K}{getLock:IncCounter}.
Then from the code structure, \K\ does not execute lines~\ref{getLock:awaitAckOrCtrDecrease}-\ref{getLock:RolePPawn}, and does not repeat the role-loop, and does not busy-wait in the spin loop of \Line{getLock:awaitX}; instead \K\ proceeds to execute lines~\ref{getLock:ApplyWantOk} -~\ref{getLock:returninfty} and returns value $\infty$ in \Line{getLock:returninfty}.
Since \K\ does not change $\Role[\K]$  during $[\pt{\K}{getLock:IncCounter},\pt{\K}{getLock:returninfty}]$, $\Role[\K] = \cKing$ throughout $[\pt{\K}{getLock:IncCounter},\pt{\K}{getLock:returninfty}]$.

\textbf{Proof of~\refC{aExecutesctr10}: }
Since \K\ is the process that increased \ctr\ from $0$ to $1$ at time \Ib{1}, and since \K\ can increase \ctr\ only by executing a \ctr.\inc{} operation in \Line{getLock:IncCounter} (by Claim~\ref{cl:basic:ctr}) \K\ set $\Role[\K] =  0 = \cKing$  at \pt{\K}{getLock:IncCounter}.
From Part~\refC{KDoesNotStarve}, \K\ returns from \lock{\K} with value $\infty$ in \Line{getLock:returninfty}, and thus \K\ consequently calls \release{\K}{j} (follows from conditions~\refC{condition:ifLockThenRelease} and~\refC{condition:ifReleaseThenExistsLock}).
Note that \K\ has not executed any operations on $\X,\LSync$ and \PawnSet\ in the process.
Then $\Role[\K] = \cKing$  at \ptB{\K}{release:ifKing} and thus $p$ satisfies the if-condition of \Line{release:ifKing} and executes the \ctr.\CAS{$1,0$} operation in \Line{release:ctr10} during $T$ without having executed any operations on $\X,\LSync$ and \PawnSet\ in the process.
Thus \K\ did not change $\X,\LSync$ or $\PawnSet$ during $[\Ib{1},\pt{\K}{release:ctr10}]$.

\textbf{Proof of~\refC{R@I1}: }
Since $R(\Ib{1}) = \Set{\K}$ (Part~\refC{@I1+}), to prove our claim we need to show that during \I{1} \K\ does not cease to be a releaser and no process becomes a releaser.
Suppose not, i.e., the claim $R(t) = \Set{\K}$ fails to hold for the first time at some point $t'$ in \I{1}.

\textbf{Case a - } Process \K\ ceases to be a releaser of \L\ at $t'$:
By definition, cease-release event $\phi_{\K}$ occurs when $\K$ executes a successful \ctr.\CAS{$1,0$} operation in \Line{release:ctr10},
From Part~\refC{aExecutesctr10}, \K\ executes a \ctr.\CAS{$1,0$} operation in \Line{release:ctr10}.
If \K\ executes a successful \ctr.\CAS{$1,0$} operation in \Line{release:ctr10} then, by definition, cease-release event $\phi_{\K}$ occurs and by \Claim{cl:releaserAtCeaseReleaseEvents} \K\ ceases to be the releaser of \L.
Thus, $t' = \pt{\K}{release:ctr10}$ and \ctr\ changes to value of $0$ at $t'$.
But since $t' \in \I{1}$ and $\ctr = 1$ throughout \I{1}, we have a contradiction.
If \K\ executes an unsuccessful \ctr.\CAS{$1,0$} operation in \Line{release:ctr10}, then $\ctr \neq 1$ at \ptB{\K}{release:ctr10}.
Since $p$ did not cease to a releaser at \ptB{\K}{release:ctr10}, $\ptB{\K}{release:ctr10} < t'$.
Since $\Ib{1} = \pt{\K}{getLock:IncCounter} < \pt{\K}{release:ctr10} < t' <  \Ia{1}$ and $\ctr = 1$ throughout \I{1}, $\ctr = 1$ at \ptB{\K}{release:ctr10}, and thus we have a contradiction.


\textbf{Case b - } Some process $q$ becomes a releaser of \L\ at $t'$:
Since \ctr\ is not increased during \I{1}, it follows from conditions (R1) and (R2) that some process $r$ promoted $q$ at time $t'$.
Then by definition, event $\pi_{r}$ occurs at $t'$, and thus from \Claim{cl:releaserAtCeaseReleaseEvents} it follows that $r$ is a releaser of \L\ immediately before $t'$.
Since \K\ is the only releaser immediately before $t'$, $r=\K$.
Then cease-release event $\pi_{\K}$ occurred at $t'$ and \K\ ceases to be a releaser at $t'$.
As was shown in \textbf{Case a}, this leads to a contradiction.

\textbf{Proof of~\refC{XT@I1}: }
At time \Ib{1} the claim $\X = \LSync = \bot$ and \PawnSet\ is candidate-empty holds by Part~\refC{@I1+}.
Suppose some process $p$ changes \LSync\ or \X\ or \PawnSet\ for the first time at some point $t'$ during \I{1}.
From \Claim{cl:variablesChangedByReleaserOnly} it follows that $p$ is a releaser of lock \L\ at time $t'$.
Since for all $t \in \I{1}$, $R(t) = \Set{\K}$ (Part~\refC{R@I1}), it follows that $p = \K$.
From Part~\refC{aExecutesctr10}, \K\ does not change any of the variables before the point when it executes a \ctr.\CAS{$1,0$} operation in \Line{release:ctr10}, i.e., $\ptB{\K}{release:ctr10} < t'$.
If \K\ executes a successful \ctr.\CAS{$1,0$} operation in \Line{release:ctr10} then the interval \I{1} ends and clearly $t' \notin \I{1}$, hence a contradiction.
If \K\ executes an unsuccessful \ctr.\CAS{$1,0$} operation in \Line{release:ctr10} then $\ctr \neq 1$ at \ptB{\K}{release:ctr10}.
Since $\Ib{1} = \ptB{\K}{getLock:IncCounter} < \ptB{\K}{release:ctr10} < t' <  \Ia{1}$ and $\ctr = 1$ throughout \I{1}, we have a contradiction.
\end{proof}

\begin{claim} \label{cl:proofProperties2}
If $\I{2} \neq \varnothing$ and $R(\Ib{0}) = \varnothing$ and at $\Ib{0}$, $\X = \LSync = \bot$ and \PawnSet\ is candidate-empty, then the following claims hold:
\begin{enumerate}[(a)]
      \item $R(\Ib{2}) = \Set{\K,\Q}$ and at time \Ib{2}, $\X = \LSync = \bot$ and \PawnSet\ is candidate-empty. \label{@I2-}
      \item \K\ and \Q\ are the first two releasers of \L.  \label{firstTwoReleasers}
      \item During $(\Ib{2},\Ia{2}]$ a process can become a releaser of \L\ only if it gets promoted by a releaser of \L. \label{onlyPromotions}
      \item If \K\ takes enough steps, \K\ executes  lines of code of \release{\K}{} starting with \Line{release:safetyCheck}  as depicted in Figure~\ref{fig:KsRelease}. \label{KsRelease}

      \item  \K\ executes an unsuccessful \ctr.\CAS{$1,0$} operation in \Line{release:ctr10}, and calls \helpRelease{\K} in \Line{release:callhRelease:King} such that $\Ib{2} < \ptB{\K}{release:ctr10} < \ptB{\K}{release:callhRelease:King}$. \label{aExecutesFailedctr10} \label{aExecutesHelpRelease}

      \item If \K\ and \Q\ take enough steps, \Q\ finishes \lock{\Q} during $T$.
 \label{QDoesNotStarve}
      \item If \K\ and \Q\ take enough steps, \Q\ executes  lines of code of \lock{\Q} starting with \Line{getLock:ApplyBotWant}  as depicted in Figure~\ref{fig:QsGetLock}. \label{QsGetLock}

      \item If \Q\ calls \release{\Q}{}, it executes lines of code of \release{\Q}{} starting with \Line{release:safetyCheck}  as depicted in Figure~\ref{fig:QsRelease}. \label{QsRelease}

      \item \Q\ calls \helpRelease{\Q} either in \Line{abort:callhRelease} or in \Line{release:callhRelease:Queen}, after time \Ib{2}. \label{QExecutesHelpRelease}  

\end{enumerate}

\end{claim}

\begin{figure*}[!htbp]
\centering
\includegraphics[scale=0.6]{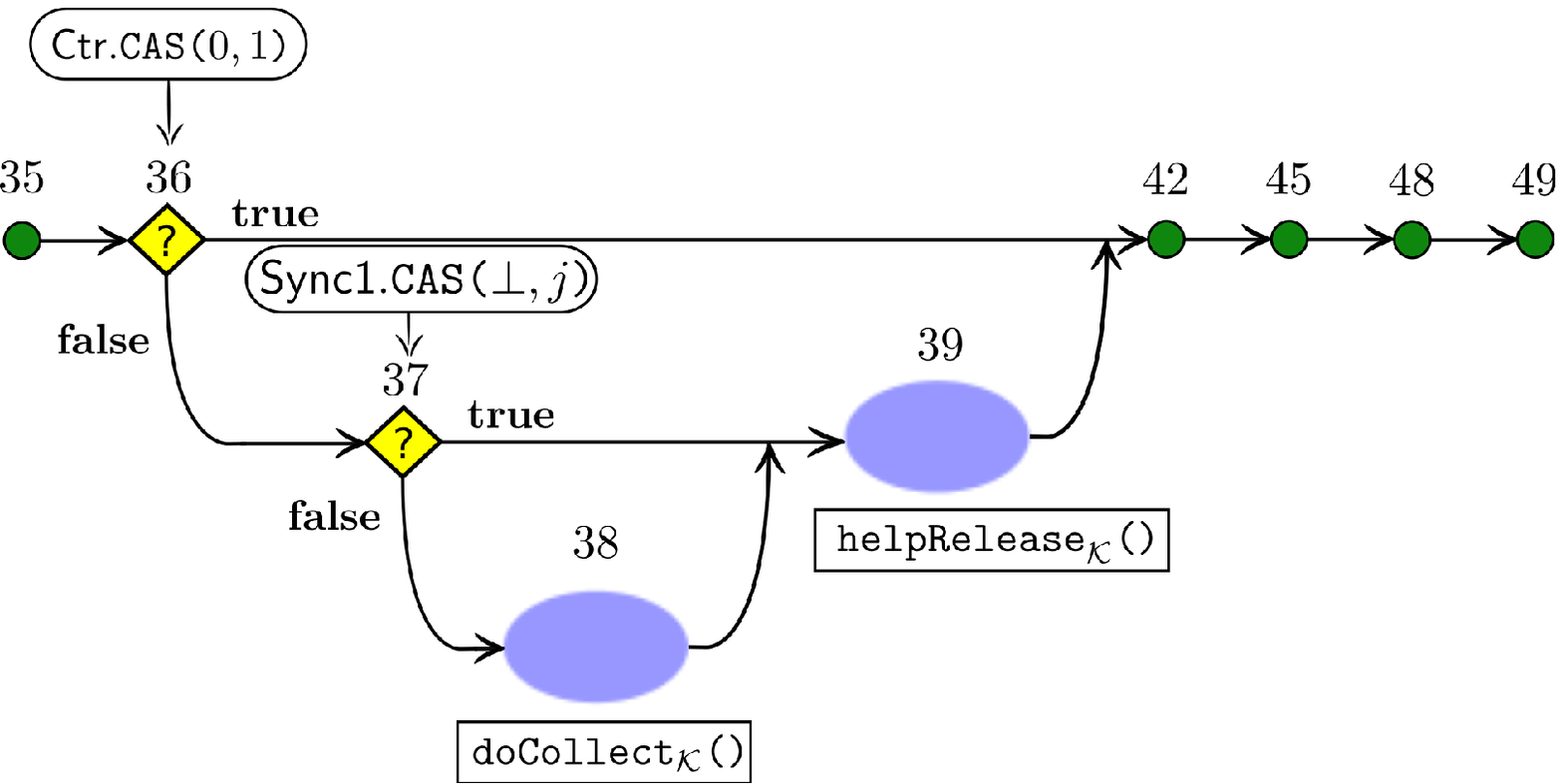}
\caption{\K's call to \release{\K}{j}}
\label{fig:KsRelease}
\end{figure*}

\begin{figure*}[!htbp]
\centering
\includegraphics[scale=0.6]{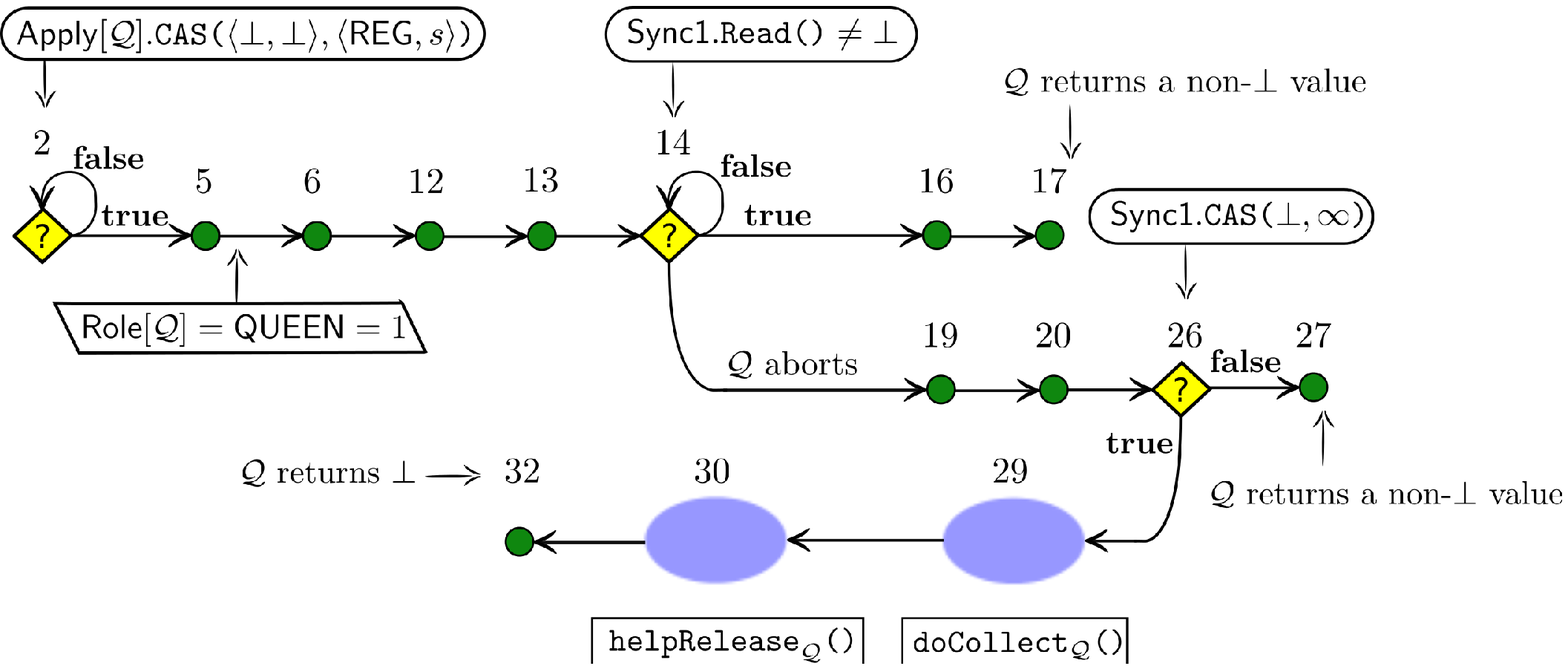}
\caption{\Q's call to \lock{\Q}}
\label{fig:QsGetLock}
\end{figure*}

\begin{figure}[!htbp]
\centering
\includegraphics[scale=0.6]{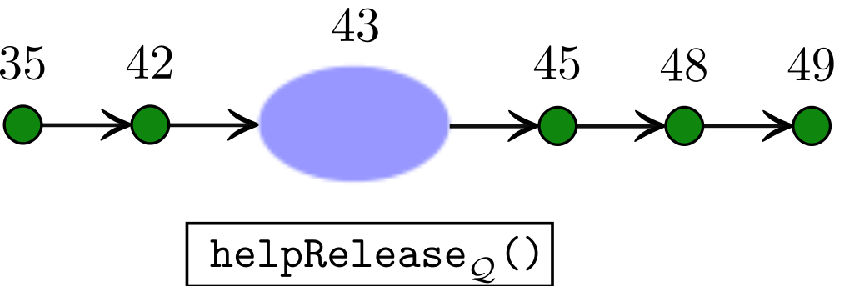}
\caption{\Q's call to \release{\Q}{}}
\label{fig:QsRelease}
\end{figure}

\begin{proof}

\textbf{Proof of~\refC{@I2-} and~\refC{firstTwoReleasers}: }
Since \Q\ is the process that increases \ctr\ from $1$ to $2$ at time \Ib{2}, and since \Q\ can increase \ctr\ only by executing a \ctr.\inc{} operation in \Line{getLock:IncCounter} (by Claim~\ref{cl:basic:ctr}) \Q\ becomes a releaser of lock \L\ by condition (R1) at $\Ib{2} = \pt{\Q}{getLock:IncCounter}$.
Since for all $t \in \I{1}$, $R(t) = \Set{K}$ (Claim~\ref{cl:proofProperties1}\refC{R@I1}), it follows that $R(\Ib{2}) = \Set{\K,\Q}$.
By claim~\ref{cl:proofProperties1}\refC{XT@I1}, throughout \I{1}, $\X = \LSync = \bot$ and \PawnSet\ is candidate-empty, and since the only operation executed at time \Ib{2} is \ctr.\inc{}, it follows that at time \Ib{2}, $\X = \LSync = \bot$ and \PawnSet\ is candidate-empty.
Hence Part~\refC{firstTwoReleasers} holds.
Clearly \K\ and \Q\ are the first two releasers of \L, hence Part~\refC{firstTwoReleasers} holds.

\textbf{Proof of~\refC{onlyPromotions}: }
From conditions (R1) and (R2), a process can become a releaser of \L\ either by increasing \ctr\ to $1$ or $2$ or by getting promoted.
Since \ctr\ is not increased during $(\Ib{2},\Ia{2}]$, it follows that during $(\Ib{2},\Ia{2}]$ a process becomes a releaser of \L\ only if it gets promoted.
By definition, a process can be promoted only when a \PawnSet.\promote{} operation is executed in \Line{promote:FR12} and from \Claim{cl:PromoteOpByReleaserOnly} only a releaser of \L\ can execute this operation.
Then during $(\Ib{2},\Ia{2}]$ a process becomes a releaser of \L\ only if it gets promoted by a releaser of \L.

\textbf{Proof of~\refC{KsRelease} and~\refC{aExecutesFailedctr10}: }
From \Claim{cl:proofProperties1}\refC{aExecutesctr10}, \K\  executes the \ctr.\CAS{$1,0$} operation in \Line{release:ctr10}  during $T$.
If \K's \ctr.\CAS{$1,0$} operation is successful then the value of \ctr\ decreases from $1$ to $0$ and the \ctr-cycle interval $T$ ends and thus $\I{2} = \varnothing$, which is a contradiction to our assumption that $\I{2} \neq \varnothing$.
Then \K's \ctr.\CAS{$1,0$} operation is unsuccessful.

Since \K\ executes an unsuccessful \ctr.\CAS{$1,0$} operation in \Line{release:ctr10}, \K\ satisfies the if-condition of \Line{release:ctr10}, executes lines~\ref{release:setX}-\ref{release:doCollect} and calls \helpRelease{\K} in \Line{release:callhRelease:King}, and then executes lines~\ref{release:ApplyOkBot}-\ref{release:return}.

Since \K\ executes an unsuccessful \ctr.\CAS{$1,0$} operation in \Line{release:ctr10}, it follows that \ctr\ was changed from $1$ to $2$ at time \Ib{2} (by definition), and thus $\Ib{2} < \pt{\K}{release:ctr10}$.
Since $\pt{\K}{release:ctr10} < \pt{\K}{release:callhRelease:King}$, it follows that $\Ib{2} < \pt{\K}{release:ctr10} < \pt{\K}{release:callhRelease:King}$.

\textbf{Proof of~\refC{QDoesNotStarve},~\refC{QsGetLock} and~\refC{QsRelease}: }
Since \Q\ is the process that increases \ctr\ from $1$ to $2$ at time \Ib{2}, and since \Q\ can increase \ctr\ only by executing a \ctr.\inc{} operation in \Line{getLock:IncCounter} (by Claim~\ref{cl:basic:ctr}) \Q\ set $\Role[\K] =  1 = \cQueen$  at $\pt{\Q}{getLock:IncCounter}=\Ib{2}$.
Then from the code structure, \Q\ does not execute lines~\ref{getLock:awaitAckOrCtrDecrease}-\ref{getLock:RolePPawn}, and does not repeat the role-loop, instead, it proceeds to \Line{getLock:ifQueen} and then proceeds to busy-wait in the spin loop of \Line{getLock:awaitX}.
Then \Q\ does not finish \lock{\Q} only if it spins indefinitely in \Line{getLock:awaitX} and does not receive a signal to abort.

For the purpose of a contradiction assume that \Q\ does not finish \lock{\Q}.
Then \Q\ reads the value $\bot$ from \X\ in \Line{getLock:awaitX} indefinitely.
From Part~\refC{aExecutesHelpRelease} it follows that \K\ executes a \X.\CAS{$\bot,j$} operation in \Line{release:setX} during $(\Ib{2},\Ia{2}]$.
Since $\X = \bot$ at time \Ib{2} (Part~\refC{@I2-}), and only a releaser can change \X\ (Claim~\ref{cl:variablesChangedByReleaserOnly}), and \Q\ is busy-waiting in \Line{getLock:awaitX}, it follows that the only other releaser, \K, executed a successful \X.\CAS{$\bot,j$} operation in \Line{release:setX} during $(\Ib{2},\Ia{2}]$ and changed \X\ to a non-$\bot$ value.
Then for \Q\ to read $\bot$ from \X\ in \Line{getLock:awaitX} indefinitely, some process must reset \X\ to $\bot$ before \Q\ reads \X\ again.

\textbf{Case a - } \K\ resets \X\ in \Line{hRelease:resetX} before \Q\ reads \X\ again:
For \K\ to reset \X\ in \Line{hRelease:resetX}, \K\ must satisfy the if-condition of \Line{hRelease:setT} and thus \K\ must execute an unsuccessful \LSync.\CAS{$\bot,\K$} operation in \Line{hRelease:setT}.
Since $\LSync = \bot$ at time \Ib{2} (Part~\refC{@I2-}), and only a releaser can change \LSync\ (Claim~\ref{cl:variablesChangedByReleaserOnly}), and \Q\ is busy-waiting in \Line{getLock:awaitX}, it follows that $\LSync = \bot$ at \ptB{\K}{hRelease:setT}.
Thus \K's \LSync.\CAS{$\bot,\K$} operation in \Line{hRelease:setT} is successful and we get a contradiction.

\textbf{Case b - } some other process becomes a releaser and resets \X\ before \Q\ reads \X\ again:
From Part~\refC{onlyPromotions} it follows that during $(\Ib{2},\Ia{2}]$ a process can become a releaser of \L\ only if it is promoted (by condition (R2)).
Since a process is promoted only by a releaser of \L\, and \K\ is the only other releaser of \L\ apart from \Q, it follows that \K\ promotes some process before \Q\ reads \X\ again.
As argued in \textbf{Case a}, \K\ executes a successful \LSync.\CAS{$\bot,\K$} operation in \Line{hRelease:setT}.
Then from the code structure, \K\ does not call \doPromote{\K} in \Line{hRelease:callPromote}, and thus \K\ does not promote any process.
Hence, we have a contradiction.

\textbf{Proof of~\refC{QExecutesHelpRelease}: }
Since \Q\ is the process that increases \ctr\ from $1$ to $2$ at time \Ib{2}, and since \Q\ can increase \ctr\ only by executing a \ctr.\inc{} operation in \Line{getLock:IncCounter} (by Claim~\ref{cl:basic:ctr}) \Q\ set $\Role[\K] =  1 = \cQueen$  at \pt{\Q}{getLock:IncCounter}.
Then from the code structure, \Q\ does not execute lines~\ref{getLock:awaitAckOrCtrDecrease}-\ref{getLock:RolePPawn}, and does not repeat the role-loopp; instead, it proceeds to \Line{getLock:ifQueen} and then proceeds to busy-wait in the spin loop of \Line{getLock:awaitX}.

\textbf{Case a - } \Q\ does not receive a signal to abort while busy-waiting in \Line{getLock:awaitX}:
From Part~\refC{QDoesNotStarve}, \Q\ does not busy-wait indefinitely in \Line{getLock:awaitX} and eventually breaks out.
Since \Q\ breaks out of the spin loop of \Line{getLock:awaitX} it reads non-$\bot$ from \X\ and then from the code structure it follows that \Q\ goes on to return that non-$\bot$ value in \Line{getLock:returninfty}.
Consequently \Q\ calls \release{\Q}{j} (follows from conditions~\ref{condition:ifLockThenRelease} and~\ref{condition:ifReleaseThenExistsLock}).
Consider \Q's call to \release{\Q}{j}.
Since \Q\ last changed $\Role[\Q]$ only in \Line{getLock:IncCounter}, $\Role[\Q] = \cQueen$ at \ptB{\Q}{release:safetyCheck}.
Since $\Role[\Q]$ is unchanged during \release{\Q}{} (\Claim{cl:basic:Role}\refC{scl:RoleUnchanged}), it follows that $\Role[\Q] = \cQueen$ throughout \release{\Q}{}.
Then from the code structure it follows that \Q\ executes only lines~\ref{release:safetyCheck}-\ref{release:ifKing},~\ref{release:ifQueen}-\ref{release:ifPPawn} and ~\ref{release:ApplyOkBot}-\ref{release:return}.
Then \Q\ calls \helpRelease{\Q} only in \Line{release:callhRelease:Queen}, and since $\Ib{2} = \pt{\Q}{getLock:IncCounter} < \pt{\Q}{release:callhRelease:Queen}$, our claim holds.

\textbf{Case b - } \Q\ receives a signal to abort while busy-waiting in \Line{getLock:awaitX}:
Then \Q\ calls \abort{\Q}, and from the code structure \Q\ executes lines~\ref{abort:ifFlag}-\ref{abort:ifPawn}, and then line~\ref{abort:setX}.
If \Q\ fails the \X.\CAS{$\bot,\infty$} operation of line~\ref{abort:setX}, then $\X \neq \bot$ at \ptB{\Q}{abort:setX}.
From Claim~\ref{cl:variablesChangedByReleaserOnly}, only a releasers of \L\ can change \X\ to a non-$\bot$ value, and since \K\ and \Q\ are the only releasers of \L, it follows that \K\ changed \X\ to a non-$\bot$ value.
Then \Q\ satisfies the if-condition of \Line{abort:setX} and returns the non-$\bot$ value written by \K\ to \X\ in \Line{abort:returnX}.
Consequently \Q\ calls \release{\Q}{j} (follows from conditions~\ref{condition:ifLockThenRelease} and~\ref{condition:ifReleaseThenExistsLock}), and as argued in \textbf{Case a}, \Q\ executes only lines~\ref{release:safetyCheck}-\ref{release:ifKing},~\ref{release:ifQueen}-\ref{release:ifPPawn} and~\ref{release:ApplyOkBot}-\ref{release:return}, and \Q\ calls \helpRelease{\Q} only in \Line{release:callhRelease:Queen}.
Since $\Ib{2} = \pt{\Q}{getLock:IncCounter} < \pt{\Q}{release:callhRelease:Queen}$, our claim holds.

If \Q's \X.\CAS{$\bot,\infty$} operation is successful, then \Q\ goes on to call \doCollect{\Q} in \Line{abort:doCollect}, calls \helpRelease{\Q} in \Line{abort:callhRelease}, then executes lines~\ref{abort:ApplyOkBot}-\ref{abort:returnr}, and finally returns $\bot$ in \Line{abort:returnr}.
Since $\Ib{2} = \pt{\Q}{getLock:IncCounter} < \pt{\Q}{abort:callhRelease}$, our claim holds.
\end{proof}

\label{notations:lambda,gamma}
Define $\lambda$ to be the first point in time when \LSync\ is changed to a non-$\bot$ value, and if \LSync\ is never changed to non-$\bot$ then $\lambda = \infty$.
Define $\gamma$ to be the first point in time when a \PawnSet.\promote{} operation is executed, and if a \PawnSet.\promote{} operation is never executed then $\gamma = \infty$.
From Claims~\ref{cl:proofProperties2}\refC{aExecutesHelpRelease} and~\ref{cl:proofProperties2}\refC{QExecutesHelpRelease}, both \K\ and \Q\ execute \helpRelease{\K} and \helpRelease{\Q}, respectively, after time \Ib{2}.
Let $\A \in \Set{\K,\Q}$ be the first process among them to execute \Line{hRelease:setT}, and let $\B \in \Set{\K,\Q} - \Set{\A}$ be the other process, i.e., $\pt{\A}{hRelease:setT} < \pt{\B}{hRelease:setT}$.

\begin{claim} \label{cl:proofProperties3}
If $\I{2} \neq \varnothing$ and $R(\Ib{0}) = \varnothing$ and at $\Ib{0}$, $\X = \LSync = \bot$ and \PawnSet\ is candidate-empty, then the following claims hold:
\begin{enumerate}[(a)]
      \item $\Ib{2} < \lambda = \pt{\A}{hRelease:setT}$ and for all $t \in [\Ib{2},\lambda)$, $R(t) = \Set{\K,\Q}$ and $\LSync = \bot$ throughout $[\Ib{2},\lambda)$, and cease-release event $\tau_{\A}$ occurs at $\lambda$. \label{I2-,L}
      \item If \K\ and \Q\ take enough steps, then \A\ executes lines of code of \helpRelease{\A} starting with \Line{hRelease:setT}  as depicted in Figure~\ref{fig:AsHelpRelease}. \label{AsHelpRelease}

\begin{figure}[!htbp]
\centering
\includegraphics[scale=0.6]{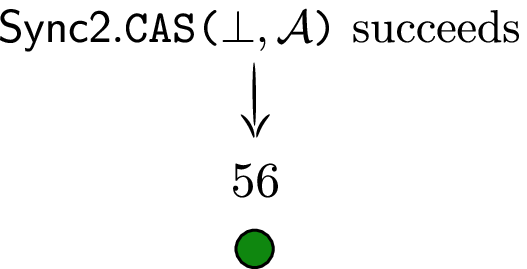}
\caption{\A's call to \helpRelease{\A}}
\label{fig:AsHelpRelease}
\end{figure}

      \item If \K\ and \Q\ take enough steps, then \B\ executes lines of code of \helpRelease{\B} and \doPromote{\B} as depicted in Figures~\ref{fig:BsHelpRelease} and~\ref{fig:BsPromote}, respectively. \label{ABhRelease} \label{BsHelpRelease} \label{BsPromote}

\begin{figure}[!htb]
\centering
\includegraphics[scale=0.6]{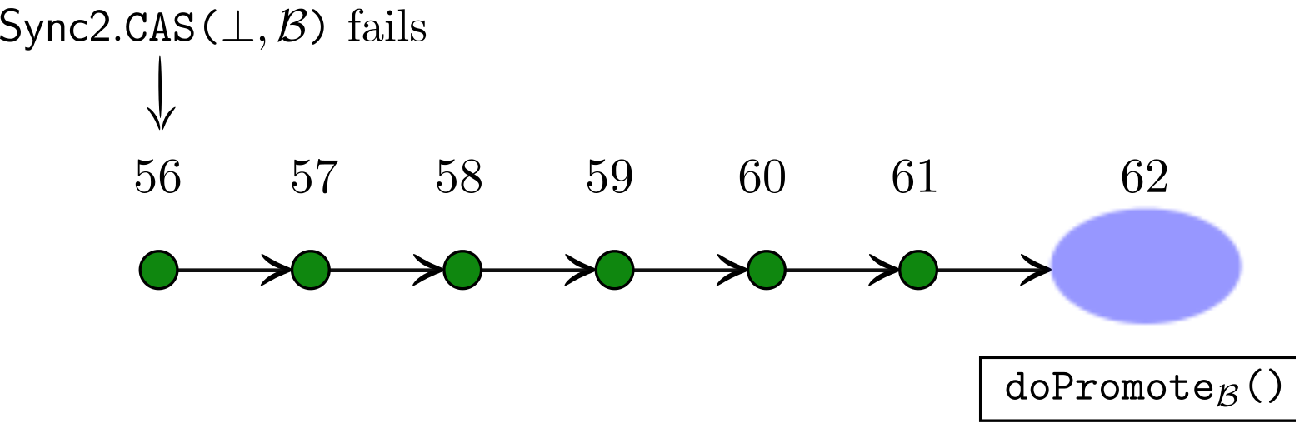}
\caption{\B's call to \helpRelease{\B}}
\label{fig:BsHelpRelease}
\end{figure}

\begin{figure}[!htb]
\centering
\includegraphics[scale=0.6]{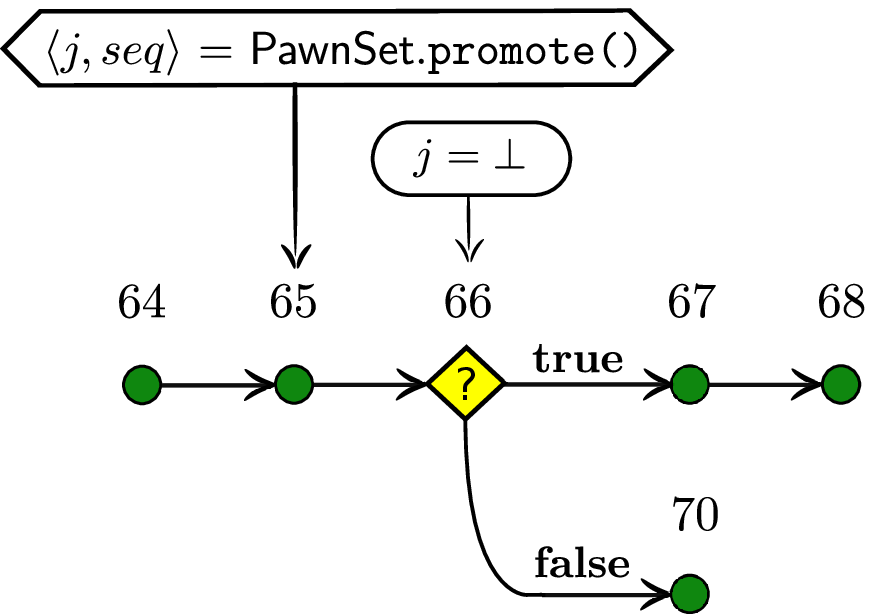}
\caption{\B's call to \doPromote{\B}}
\label{fig:BsPromote}
\end{figure}

      \item $\lambda < \gamma = \pt{\B}{promote:FR12}$. \label{GExists}
      \item $\forall_{t \in [\lambda,\gamma)}$, $R(t) = \Set{\B}$. \label{L,G}
      \item At time $\gamma$, $\X= \LSync = \bot$. \label{@G}
      \item No promotion event occurs at lock \L\ during $[\Ib{2},\gamma)$. \label{noPromotion@I0-,G}
      \item The \PawnSet.\promote{} operation at time $\gamma$ does not return a value in \Set{\pair{a}{b} | a \in \Set{\K,\Q}, b \in \N}.\label{abNotPromoted}
      \item If the \PawnSet.\promote{} operation at time $\gamma$ returns a non-\pair{\bot}{\bot} value  then \B's cease-release event $\pi_{\B}$ occurs at time $\gamma$.
\label{promotion@G}
      \item If the \PawnSet.\promote{} operation at time $\gamma$ returns value \pair{\bot}{\bot} then \B's cease-release event $\theta_{\B}$ occurs at $t' = \pt{\B}{promote:ctr20} \geq \gamma$, and throughout $[\gamma,t']$ no process is promoted, and $\forall_{t \in [\gamma,t')}$, $R(t) = \Set{\B}$. \label{noPromotion@G}
      \item Either \K\ or \Q\ calls \doCollect{}, specifically during $[\Ib{2},\gamma]$. \label{aOrbCollect}
\end{enumerate}

\end{claim}

\begin{proof}
\textbf{Proof of~\refC{I2-,L}: }
We first show that for all $t \in [\Ib{2},\ptB{\A}{hRelease:setT}]$, $R(t) = \Set{\K,\Q}$ and then show that $\lambda = \pt{\A}{hRelease:setT}$.
From Claims~\ref{cl:proofProperties2}\refC{aExecutesHelpRelease} \K\ calls \helpRelease{\K} in \Line{release:callhRelease:King} after time \Ib{2}.
From Claim~\ref{cl:proofProperties2}\refC{@I2-}, \K\ is a releaser of \L\ at time \Ib{2}.
From an inspection of Figures~\ref{fig:KsGetLock} and~\ref{fig:KsRelease}, throughout $[\I{1},\ptB{K}{release:callhRelease:King}]$ \K\ does not execute a call to \helpRelease{\K} or \doPromote{\K}.
Also from an inspection, \K\ fails to decrease $\ctr$ from $1$ to $0$ at \pt{\K}{release:ctr10}, thus \K's cease-release event $\phi_{\K}$ does not occur.
Since \K's cease-release events $\tau_{\K}, \pi_{\K}$ and $\theta_{\K}$ only occur during \helpRelease{\K} or \doPromote{\K} (Claims~\ref{cl:basic:ReleaseEvents}\refC{scl:tau} and~\ref{cl:basic:ReleaseEvents}\refC{scl:piOrTheta}), it follows that \K\ is a releaser of \L\ throughout $[\Ib{1},\ptB{\K}{hRelease:setT}]$.

From Claim~\ref{cl:proofProperties2}\refC{QExecutesHelpRelease}, \Q\ calls \helpRelease{\Q}, respectively either in \Line{abort:callhRelease} or \Line{release:callhRelease:Queen}, after time \Ib{2}.
From Claim~\ref{cl:proofProperties2}\refC{@I2-}, \Q\ is a releaser of \L\ at time \Ib{2}.
From an inspection of Figures~\ref{fig:QsGetLock} and~\ref{fig:QsRelease}, throughout $[\I{2},\ptB{Q}{hRelease:setT}]$ \Q\ does not execute a call to \helpRelease{\Q} or \doPromote{\Q}.
Also from an inspection, \Q\ does not execute a \ctr.\CAS{$1,0$} operation in \Line{release:ctr10}, and  thus \Q's cease release event $\phi_{\Q}$ does not occur.
Since \Q's cease-release events $\tau_{\Q}, \pi_{\Q}$ and $\theta_{\Q}$ only occur during \helpRelease{\Q} or \doPromote{\Q} (Claims~\ref{cl:basic:ReleaseEvents}\refC{scl:tau} and~\ref{cl:basic:ReleaseEvents}\refC{scl:piOrTheta}), it follows that \Q\ is a releaser of \L\ throughout $[\Ib{2},\ptB{\Q}{hRelease:setT}]$.

Then for all $t \in [\Ib{2},\ptB{\A}{hRelease:setT}]$, $\Set{\K,\Q} \subseteq R(t)$ since $\Ib{1} < \Ib{2}$ and $\ptB{\A}{hRelease:setT} = \minimum{\ptB{\K}{hRelease:setT}, \ptB{\Q}{hRelease:setT}}$.
From Claim~\ref{cl:proofProperties2}\refC{onlyPromotions}, it follows that a process can become a releaser during \I{2} only if it is promoted by a releaser of \L.
Then to show that for all $t \in [\Ib{2},\ptB{\A}{hRelease:setT}]$, $R(t) = \Set{\K,\Q}$, we need to show that no process is promoted by \K\ or \Q\ during $[\Ib{2},\ptB{\A}{hRelease:setT}]$.
If a process was promoted by \K\ or \Q\ during $[\Ib{2},\ptB{\A}{hRelease:setT}]$ then by definition cease-release events $\pi_{\K}$ or $\pi_{\Q}$ would have occurred during $[\Ib{2},\ptB{\A}{hRelease:setT}]$, but as shown above this does not happen.

From Claim~\ref{cl:proofProperties2}\refC{@I2-}, $\LSync = \bot$ at time \Ib{2}.
From a code inspection, \LSync\ is changed to a non-$\bot$ value only in \Line{hRelease:setT} (during \helpRelease{}), moreover only by a releaser of \L\ (from \Claim{cl:variablesChangedByReleaserOnly}).
Since for all $t \in [\Ib{2},\ptB{\A}{hRelease:setT}]$, $R(t) = \Set{\K,\Q}$ and $\ptB{\A}{hRelease:setT} = \minimum{\ptB{\K}{hRelease:setT}, \ptB{\Q}{hRelease:setT}}$, it follows then that $\LSync = \bot$ throughout $[\Ib{2},\ptB{\A}{hRelease:setT}]$ and $\A$ executes a successful \LSync.\CAS{$\bot,\A$} operation in \Line{hRelease:setT}.
Thus $\A$'s cease-release  event $\tau_{\A}$ occurs at $\pt{\A}{hRelease:setT}$.

Since $\LSync = \bot$ throughout $[\Ib{0},\Ia{1}]$ (Claims~\ref{cl:proofProperties1}\refC{I0} and~\ref{cl:proofProperties1}\refC{XT@I1}) and throughout $[\Ib{2},\ptB{\A}{hRelease:setT}]$, it follows that \LSync\ was changed to a non-$\bot$ value for the first time at \pt{\A}{hRelease:setT}, thus $\lambda = \pt{\A}{hRelease:setT}$.
Then it follows for all $t \in [\Ib{2},\lambda)$, $R(t) = \Set{\K,\Q}$, and $\LSync = \bot$ throughout $[\Ib{2},\lambda)$

\textbf{Proof of~\refC{AsHelpRelease}: }
From Part~\refC{I2-,L}, \A's cease-release event $\tau_{A}$ occurs at $\lambda = \pt{A}{hRelease:setT}$, and thus \A's \LSync.\CAS{$\bot,\A$} operation in \Line{hRelease:setT} succeeds.
Then from the code structure \A\ does not satisfy the if-condition on \Line{hRelease:setT} and returns from its call to \helpRelease{\A}.
Thus, \Figure{fig:AsHelpRelease} follows.

\textbf{Proof of~\refC{BsHelpRelease},~\refC{GExists},~\refC{L,G},~\refC{@G},~\refC{noPromotion@I0-,G},~\refC{abNotPromoted},~\refC{promotion@G} and~\refC{noPromotion@G}: }
From Part~\refC{I2-,L}, $\lambda = \pt{\A}{hRelease:setT}$ and for all $t \in [\Ib{2},\lambda)$, $R(t) = \Set{\K,Q}$ and $\LSync = \bot$ throughout $[\Ib{2},\lambda)$ and cease-release event $\tau_{\A}$ occurs at $\lambda$.
Then $\A$ ceases to be a releaser of \L\ at $\lambda$, and thus $R(\lambda) = \Set{\B}$ and $\LSync = \A \neq \bot$ at $\lambda$.
From \Claim{cl:proofProperties2}\refC{onlyPromotions} it follows that \B\ will continue to be the only releaser of \L\ until the point when \B\ ceases to be a releaser of \L\ or promotes another process.
Let $t > \lambda$ be the point in time when \B\ ceases to be a releaser of \L.
Since \B\ ceases to be a releaser of \L\ if it promotes another process (by definition of cease-release event $\pi_{\B}$), it follows that \B\ is the only releaser of \L\ throughout $[\lambda,t)$.
Then from Claim~\ref{cl:variablesChangedByReleaserOnly} it follows that \B\ has exclusive write-access to $\X,\LSync$ and exclusive registration-access to $\PawnSet$ throughout $[\lambda,t)$.

Now consider $\B$'s \helpRelease{\B} call.
Since $\lambda = \pt{\A}{hRelease:setT} < \pt{\B}{hRelease:setT}$ and $\LSync \neq \bot$ at $\lambda$ and \B\ has exclusive write-access to \LSync\ throughout $[\lambda,t)$, $\B$ fails the \LSync.\CAS{$\bot,\B$} operation at $\pt{\B}{hRelease:setT}$, and thus satisfies the if-condition of \Line{hRelease:setT}.
It then executes lines~\ref{hRelease:readX} -~\ref{hRelease:callPromote}, and calls \doPromote{\B} in \Line{hRelease:callPromote}.
Then Figures~\ref{fig:BsHelpRelease} and~\ref{fig:BsPromote} and Part~\refC{BsHelpRelease} follows immediately.

We now show that $\gamma = \pt{\B}{promote:FR12} \leq t$.
Since $\lambda = \pt{\A}{hRelease:setT}$ and $\pt{\A}{hRelease:setT} < \pt{\B}{hRelease:setT} < \pt{\B}{promote:FR12}$, it would follow that $\lambda < \gamma$, and hence we would have proved Part~\refC{GExists}.
And since \B\ is the only releaser of \L\ throughout $[\lambda,t)$, we would have proved Part~\refC{L,G} as well, i.e., \B\ is the only releaser throughout $[\lambda,\gamma)$.

During \doPromote{\B}, \B\ executes a \PawnSet.\promote{} operation in \Line{promote:FR12}.
Since \K\ and \Q\ are the first two releasers of \L\ during $T$ (Claim~\ref{cl:proofProperties2}\refC{firstTwoReleasers}), and only a releaser executes a \PawnSet.\promote{} operation (\Claim{cl:PromoteOpByReleaserOnly}), and $\A$ ceased to be a releaser at $\pt{\A}{hRelease:setT} < \pt{\B}{promote:FR12}$, it follows that \B's \PawnSet.\promote{} operation in \Line{promote:FR12} is the first \PawnSet.\promote{} operation, and thus $\gamma = \pt{\B}{promote:FR12}$.
Since  none of \B's cease-release events occur during $[\pt{\B}{hRelease:setT},\ptB{\B}{promote:FR12}]$, $t \geq \pt{\B}{promote:FR12}$.

During $[\pt{\B}{hRelease:setT},\pt{\B}{promote:FR12}]$, \B\ resets \X\ and \LSync\ in lines~\ref{hRelease:resetX} and~\ref{hRelease:resetT}, respectively, and since \B\ has exclusive write-access to \X\ and \LSync\ throughout $[\pt{\B}{hRelease:setT},\pt{\B}{promote:FR12}]$, at time $\gamma = \pt{\B}{promote:FR12}$, $\X = \LSync = \bot$.
Thus, Part~\refC{@G} follows.

By definition $\gamma$ is the point in time when the first \PawnSet.\promote{} operation occurs.
Since a promotion event occurs only when a \PawnSet.\promote{} operation returns a non-\pair{\bot}{\bot} value, it follows that no promotion event occurs during $[\Ib{0},\gamma)$.
Hence, Part~\refC{noPromotion@I0-,G} follows.

Since \B\ has exclusive write-access to \LSync\ throughout $[\lambda,\pt{\B}{promote:FR12}]$, and  $\LSync = \A$ at $\lambda > \pt{\B}{hRelease:setT}$, $\B$ reads the value $\A$ from \LSync\ in \Line{hRelease:readT} and executes a \PawnSet.\remove{$\A$} operation in \Line{hRelease:collectFixOther}.
Since $\B$ executes \PawnSet.\remove{$\A$} and \PawnSet.\remove{$\B$} in lines~\ref{hRelease:collectFixOther} and~\ref{promote:collectFixSelf} during $[\lambda,\gamma)$ and \B\ has exclusive registration-access to \PawnSet\ during $[\lambda,\gamma)$, it follows from the semantics of the \APArray{n}\ object that $\B$'s \PawnSet.\promote{} operation at time $\gamma$ does not return values in $\Set{\pair{a}{b} | a \in \Set{\A,\B} = \Set{\K,\Q}, b \in \N}$.
Hence, Part~\refC{abNotPromoted} follows.

\textbf{Case a - } $\B$'s \PawnSet.\promote{} operation returns a non-\pair{\bot}{\bot} value:
Then $\B$'s cease-release  event $\pi_{\B}$ occurs at $\pt{\B}{promote:FR12} = \gamma$ (\Claim{cl:PromoteOpByReleaserOnly}), and thus Part~\refC{promotion@G} holds.


\textbf{Case b - } $\B$'s \PawnSet.\promote{} operation in \Line{promote:FR12} returns \pair{\bot}{\bot}.
Then $\B$ did not find any process to promote, and thus cease-release event $\pi_{\B}$ did not occur.
From the code structure $\B$ goes on to execute a \ctr.\CAS{$2,0$} operation in \Line{promote:ctr20}.
Since $\ctr = 2$ throughout \I{2}, it follows that $\B$'s \ctr.\CAS{$2,0$} operation succeeds, and thus $\B$'s cease-release  event $\theta_{\B}$ occurs at $\pt{\B}{promote:ctr20}$ and the intervals \I{2} and $T$ end.
Therefore $t' = \pt{\B}{promote:ctr20} > \pt{\B}{promote:FR12} = \gamma$.
Clearly, \B\ does not promote any process in $[\pt{\B}{promote:FR12},\pt{\B}{promote:ctr20}]$ = $[\gamma,t']$, and thus Part~\refC{noPromotion@G} holds.


\textbf{Proof of~\refC{aOrbCollect}: }
From an inspection of \Figure{fig:KsRelease}, \K\ executes a \X.\CAS{$\bot,\K$} operation in \Line{release:setX}.
Since $\Ib{2} < \pt{\K}{release:ctr10} < \pt{\K}{release:setX} < \pt{\K}{hRelease:setT} < \gamma$ (from Parts~\refC{I2-,L} and~\refC{GExists}), it follows that $\pt{\K}{release:setX} \in [\Ib{2},\gamma]$.
From an inspection of Figures~\ref{fig:QsGetLock} and~\ref{fig:QsGetLock}, \Q\ may or may not execute a \X.\CAS{$\bot,\infty$} operation in \Line{abort:setX}.
If \Q\ executes a \X.\CAS{$\bot,\infty$} operation in \Line{abort:setX}, since $\Ib{2} < \pt{\Q}{abort:setX} < \pt{\Q}{hRelease:setT} < \gamma$, it follows that $\pt{\Q}{abort:setX} \in [\Ib{2},\gamma]$.

Since for all $t \in [\Ib{2},\gamma]$, $R(t) \subseteq \Set{\K,\Q}$ (from Parts~\refC{I2-,L} and~\refC{L,G}), and only releasers of \L\ have write-access to \X\ (Claim~\ref{cl:variablesChangedByReleaserOnly}), and $\X = \bot$ at $\Ib{2}$ (Claim~\ref{cl:proofProperties2}\refC{@I2-}), it follows that either \K\ or \Q\ executes a successful \CAS{} operation on \X.
Then from the code structure it follows that either \K\ or \Q\ executed a call to \doCollect{} in lines~\ref{release:doCollect} or~\ref{abort:doCollect}, respectively.
Since $\pt{\K}{release:doCollect} < \ptB{\K}{release:callhRelease:King} = \ptB{\K}{hRelease:setT} < \gamma$ and  $\pt{\Q}{abort:doCollect} < \ptB{\Q}{hRelease:setT} < \gamma$, \K\ or \Q\ executed a call to \doCollect{} during $[\Ib{2},\gamma]$.
\end{proof}

\begin{claim}
If a process $p$ is promoted at time $t' \in T$ and a \PawnSet.\reset{} has not been executed during $[\Ib{0},t']$, then $p$ did not execute a \PawnSet.\cUpdate{$p,s$} operation during $[\Ib{0},t']$, where $s \in \N$.
\label{cl:ifPromotedThenWaiting}
\end{claim}
\begin{proof}
Suppose not, i.e., $p$ executed a \PawnSet.\cUpdate{$p,s$} operation at time $t < t'$.
Since $p$ has not been promoted before $t' > t$ it follows that a \PawnSet.\promote{} operation that returns \pair{p}{\cdot} has not been executed before $t$.
Then from \Claim{cl:basic:Collect}\refC{scl:CollectOps} and the semantics of \PawnSet, it follows that the $p$-th entry of \PawnSet\ is not at value \pair{\cPro}{s} = \pair{2}{s} throughout $[\Ib{0},t]$.
Then $p$'s \PawnSet.\cUpdate{$p,s$} operation at $t$ succeeds, and thus $p$ writes value \pair{\cAbort}{s} = \pair{3}{s} to the $p$-entry of \PawnSet.
Then for $p$ to be promoted at $t' > t$, it follows from the semantics of \PawnSet\ and \Claim{cl:basic:Collect}\refC{scl:CollectOps}, that during $[t,t')$ a \PawnSet.\reset{} operation and then a \PawnSet.\collect{$A$} operation where $A[p] = s$, must be executed, followed by a \PawnSet.\promote{} at $t'$ that returns \pair{p}{s}.
This is a contradiction to the assumption that a \PawnSet.\reset{} is not executed during $[\I{0},t']$.
\end{proof}

\label{notations:Omega,P}
Let $\ell$ be the number of times a promotion occurs during $T$.
For all $i \in \Set{1,\ldots,\ell}$, define $\Omega_{i}$ to be the $i$-th interval $[\Ob{i},\Oa{i}]$ that begins when the $i$-th promotion occurs during $T$ and ends when the promoted process ceases to be a releaser of \L.
Let $\P_i$ be the process promoted at \Ob{i}.

\begin{claim} \label{cl:proofProperties4}

If $\I{2} \neq \varnothing$ and $R(\Ib{0}) = \varnothing$ and at time \Ib{0}, $\X = \LSync = \bot$ and \PawnSet\ is candidate-empty, then the following claims hold for all $i \in \Set{1,\ldots,\ell}$:
\begin{enumerate}[(a)]
      \item If $\ell \geq 1$, then  $\gamma = \Ob{1}$ and $R(\Ob{1}) = \Set{\P_1}$, and $\X = \LSync  = \bot$ at \Ob{1}, and no \PawnSet.\reset{} operation has been executed during $[\Ib{0},\Ob{1}]$.  \label{@O1-}
      \item If  $R(\Ob{i}) = \Set{\P_{i}}$, then for all $t \in [\Ob{i},\Oa{i})$, $R(t) = \Set{\P_{i}}$. (i.e., $\P_i$ is the only releaser throughout $\Omega_i$) \label{help1}
      \item If $i \neq \ell$ and $R(\Ob{i}) = \Set{\P_{i}}$, then $\Oa{i} = \Ob{i+1}$ and $R(\Ob{i+1}) = \Set{\P_{i+1}}$. (i.e., $\P_{i+1}$ is the only releaser at $\Ob{i+1}$) \label{help2}
      \item If $i \neq \ell$, then $\Oa{i} = \Ob{i+1}$ and $R(\Ob{i+1}) = \Set{\P_{i+1}}$.  \label{Oi+=Oi+1-}
      \item For all $t \in [\Ob{i},\Oa{i})$, $R(t) = \Set{\P_i}$. (i.e., $\P_i$ is the only releaser throughout $\Omega_i$) \label{R@Oi-,Oi+}
\end{enumerate}

\end{claim}

\begin{proof}

\textbf{Proof of~\refC{@O1-}: }
If the \PawnSet.\promote{} operation at time $\gamma$ returns value \pair{\bot}{\bot}, then from Claims~\ref{cl:proofProperties3}\refC{noPromotion@I0-,G} and~\ref{cl:proofProperties3}\refC{noPromotion@G} it follows that no promotion occurs during $T$, which is a contradiction to $\ell \geq 1$.
Thus, the \PawnSet.\promote{} operation at time $\gamma$ returns a non-\pair{\bot}{\bot} value.
By definition $\gamma$ is the point when the first \PawnSet.\promote{} operation occurs, and $\Ob{1}$ is the point when the first promotion occurs and $\P_1$ is the process promoted at \Ob{1}.
Then $\gamma = \Ob{1}$, and $\P_1$ is the first promoted process.
From Claim~\ref{cl:proofProperties3}\refC{L,G}, \B\ is the only releaser of \L\ at the point in time immediately before time $\gamma$.
Then from Claim~\ref{cl:proofProperties3}\refC{promotion@G} it follows that \B\ promotes $\P_1$ at time $\gamma = \Ob{1}$, and \B\ ceases to be a releaser of \L\ at $\gamma$, therefore $R(\gamma) = \Set{\P_1}$.
From Claim~\ref{cl:proofProperties3}\refC{@G} it follows that $\X = \LSync  = \bot$ at \Ob{1}.

From an inspection of the code, a \PawnSet.\reset{} is executed only in \Line{promote:resetBackpack}, and it can be executed only after a \PawnSet.\promote{} is executed in \Line{promote:FR12}.
Since $\gamma$ is the first point when a \PawnSet.\promote{} is executed, it follows that no \PawnSet.\reset{} operation was executed during $[\Ib{0},\gamma]$.

\textbf{Proof of~\refC{help1}: }
Since $R(\Ob{i}) = \Set{\P_{i}}$, and \Oa{i} is the point when $\P_i$ ceases to be a releaser of \L, for all $t \in [\Ob{i},\Oa{i})$, $\Set{\P_i} \subseteq R(t)$.
To show that for all $t \in [\Ob{i},\Oa{i})$, $R(t) = \Set{\P_i}$, we need to show that no other process becomes a releaser of \L, during $[\Ob{i},\Oa{i})$.
Suppose some process $q \neq \P_i$ becomes a releaser of \L\ some time during that interval.
Since $\Ob{i} > \Ob{1} = \gamma > \Ib{2}$, from Claim~\ref{cl:proofProperties2}\refC{onlyPromotions} it follows that $\P_i$ promotes $q$ during $[\Ob{i},\Oa{i})$.
Then from \Claim{cl:PromoteOpByReleaserOnly}, $\P_i$'s cease-release event $\pi_{\P_i}$ occurs during $[\Ob{i},\Oa{i})$, and thus $\P_i$ ceases to be a releaser of \L\ during $[\Ob{i},\Oa{i})$.
Hence a contradiction.

\textbf{Proof of~\refC{help2}: }
Since $i < \ell$, it follows that there exists a process $\P_{i+1}$ that becomes a releaser of \L\ during $T$.
By definition, $\P_i$ and $\P_{i+1}$ are the $i$-th and $(i+1)$-th promoted processes during $T$, respectively.
Since $\Ob{i+1} > \Ob{i} > \Ob{1} = \gamma > \Ib{2}$, from \Claim{cl:proofProperties2}\refC{onlyPromotions} it follows that no other process becomes a releaser after $\P_i$ became a releaser and before $\P_{i+1}$ becomes a releaser, i.e., during $[\Ob{i},\Ob{i+1}]$.
Moreover, since $R(\Ob{i}) = \Set{\P_{i}}$, it follows that the next process to be promoted, i.e., $\P_{i+1}$, is promoted by the only releaser of \L, $\P_{i}$.
Then from \Claim{cl:PromoteOpByReleaserOnly}, it follows that $\P_i$ promotes $\P_{i+1}$ by executing a \PawnSet.\promote{} in \Line{promote:FR12} that returns \pair{\P_{i+1}}{s}, where $s \in \N$, and event $\pi_{\P_i}$ occurs at \pt{\P_i}{promote:FR12}.
Then $\P_i$ ceases to be a releaser of \L\ at \pt{\P_i}{promote:FR12} and thus $\Oa{i} = \pt{\P_i}{promote:FR12}$.
Since \Ob{i+1} is the point when $\P_{i+1}$ becomes a releaser of \L, it follows that $\Oa{i} = \Ob{i+1}$, and thus $R(\Ob{i+1}) = \Set{\P_{i+1}}$.

\textbf{Proof of~\refC{Oi+=Oi+1-}: }
We prove by induction that for all $k < \ell$, $R(\Ob{k+1}) = \Set{\P_{k+1}}$ and $\Oa{k} = \Ob{k+1}$.

\textbf{Basis ($k = 1$)}
From Part~\refC{@O1-}, $\P_1$ is the only releaser of \L\ at \Ob{1}, and clearly $\ell > k = 1$.
Then from Part~\refC{help2}, $\Oa{1} = \Ob{2}$ and $R(\Ob{2}) = \Set{\P_2}$.

\textbf{Induction step ($k > 1$)}
By definition $\P_{k}$ is the promoted process at \Ob{k}, and since $|R(\Oa{k-1})| = 1$ and $\Oa{k-1} = \Ob{k}$ (by the induction hypothesis), it follows that $\P_{k}$ is the only releaser of \L\ at \Ob{k}.
Then from Part~\refC{help2}, $\Oa{k} = \Ob{k+1}$ and $R(\Ob{k+1}) = \Set{\P_{k+1}}$.


\textbf{Proof of~\refC{R@Oi-,Oi+}: }
From Part~\refC{@O1-}, $R(\Ob{1}) = \Set{\P_1}$, and thus from Part~\refC{help1},  for all $t \in [\Ob{1},\Oa{1})$, $R(t) = \Set{\P_{1}}$.
From Part~\refC{Oi+=Oi+1-}, for all $i >1$, $R(\Ob{i}) = \Set{\P_i}$, and thus from Part~\refC{help1}, for all $t \in [\Ob{i},\Oa{i})$, $R(t) = \Set{\P_{i}}$.
Hence, our claim follows.
\end{proof}

\begin{claim} \label{cl:proofProperties5}

If $\I{2} \neq \varnothing$ and $R(\Ib{0}) = \varnothing$ and at time \Ib{0}, $\X = \LSync = \bot$ and \PawnSet\ is candidate-empty, then the following claims hold for all $i \in \Set{1,\ldots,\ell}$:
\begin{enumerate}[(a)]
      \item A \PawnSet.\reset{} operation is not executed during $[\Ib{0},\Ob{i}]$. \label{noReset@I0-,Oi-}
      \item $\P_i$ executes lines of code of \lock{\P_i} starting with \Line{getLock:ApplyBotWant} as depicted in Figure~\ref{fig:PisGetLock}.\label{PisGetLock}

\begin{figure*}[!htb]
\centering
\includegraphics[scale=0.6]{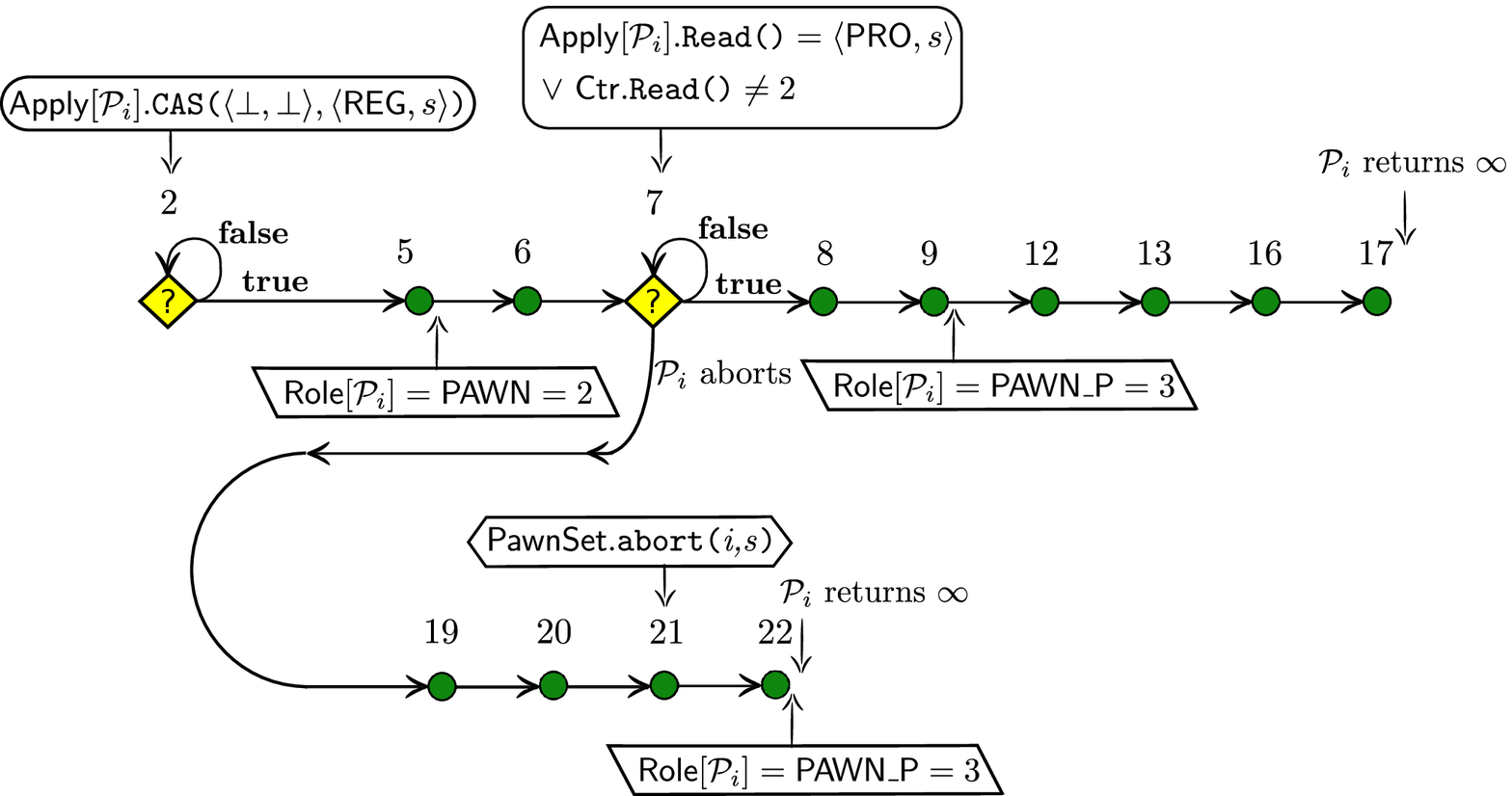}
\caption{$\P_i$'s call to \lock{\P_i}}
\label{fig:PisGetLock}
\end{figure*}

      \item $\P_i$'s call to \lock{\P_i} returns $\infty$, and $\P_i$ finishes \lock{\P_i} during $T$, and $\Role[\P_i] = \cPPawn$ when $\P_i$'s call to \lock{\P_i} returns.\label{PiDoesNotStarve}\label{Role=PPawn}

      \item Exactly one cease-release event among $\pi_{\P_i}$ and $\theta_{\P_i}$ occurs during $\P_i$'s call to \doPromote{\P_i}.\label{PisCeaseReleaseEvent}

      \item $\P_i$ executes lines of code of \release{\P_i}{} starting with \Line{release:safetyCheck} as depicted in Figure~\ref{fig:PisRelease}.\label{PisRelease}

\begin{figure}[!htb]
\centering
\includegraphics[scale=0.6]{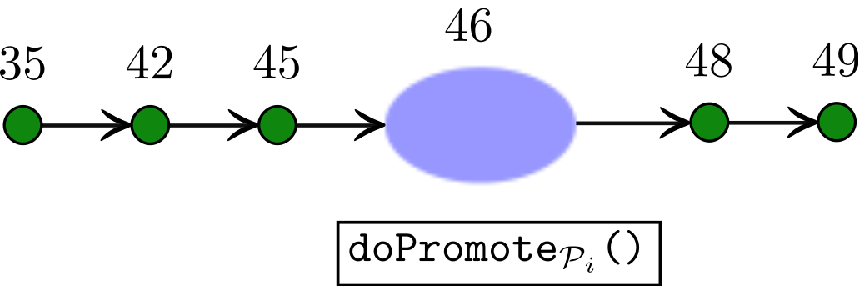}
\caption{$\P_i$'s call to \release{\P_i}{}}
\label{fig:PisRelease}
\end{figure}

      \item $\P_i$ does not write to \X\ or \LSync\ during $[\Ob{i},\Oa{i}]$. \label{Oi-,Oi+helper}

      \item $\pt{\P_i}{getLock:ApplyBotWant} < \Ob{i} < \ptB{\P_i}{release:safetyCheck} < \Oa{i} < \pt{\P_i}{release:ApplyOkBot}$ and $\Oa{i} \leq \Ia{2}$.\label{whenOioccurs}

      \item  If $i \neq \ell$, then a \PawnSet.\reset{} operation is not executed during $[\Ib{0},\Oa{i}]$. \label{Oi-,Oi+helper2}
      \item Throughout $[\gamma,\Oa{\ell}]$, $\X = \LSync = \bot$. \label{XTC@Oi-,Oi+}
      \item If $\ell > 1$, $\Ia{2} = \Oa{\ell} = \pt{\P_{\ell}}{promote:ctr20}$. \label{I2+=Ol+}
      \item For all $t \in [\gamma,\Ia{2})$, $|R(t)| = 1$. \label{[G,I2+]}
      \item $R(\Ia{2}) = \varnothing$ and at \Ia{2}, $\X= \LSync = \bot$ and \PawnSet\ is candidate-empty. \label{@I2+}
\end{enumerate}

\end{claim}

\begin{proof}

\textbf{Proof of~\refC{noReset@I0-,Oi-}-\refC{Oi-,Oi+helper2}: }
We prove Parts~\refC{noReset@I0-,Oi-}-\refC{Oi-,Oi+helper2} by induction on $i$.
First, we prove Part~\refC{noReset@I0-,Oi-} for $i=1$.
Second, we show that if Part~\refC{noReset@I0-,Oi-} is true for a fixed $i$, then Parts~\refC{PisGetLock}-\refC{Oi-,Oi+helper2} are true for $i$.
Finally, we show that if Parts~\refC{noReset@I0-,Oi-}-\refC{Oi-,Oi+helper2} are true for $i$, then Part~\refC{noReset@I0-,Oi-} is true for $i+1$, thus completing the proof.

From Claim~\ref{cl:proofProperties4}\refC{@O1-}, no \PawnSet.\reset{} operation has been executed during $[\Ib{0},\Ob{1}]$.
Hence,  Part~\refC{noReset@I0-,Oi-} for $i=1$ holds.

Now we show that if Part~\refC{noReset@I0-,Oi-} is true for a fixed $i$, then Parts~\refC{PisGetLock}-\refC{Oi-,Oi+helper2} follow for $i$.

\textbf{Proof of Parts~\refC{PisGetLock} and~\refC{PiDoesNotStarve} if Part~\refC{noReset@I0-,Oi-} for $i$ is true: }
Let $q$ be the process that promotes $\P_i$ at $\Ob{i}$.
Then $q$'s \PawnSet.\promote{} operation in \Line{promote:FR12} returned value \pair{\P_i}{s}, where $s \in \N$, and $\Ob{i} = \pt{q}{promote:FR12}$.
Then from the semantics of the \PawnSet\ object it follows that the $\P_i$-th entry of \PawnSet\ was \pair{\cReg}{s} = \pair{1}{s} immediately before $\Ob{i}$.
Then from \Claim{cl:basic:Collect}\refC{scl:CollectRegister} it follows that some process (say $r$) executed a \PawnSet.\collect{$A$} operation in \Line{collect:updateAll} where $A[\P_i] = s$.
Then from the code structure, $r$ read $\Apply[\P_i] = \pair{\cWant}{s}$ in \Line{collect:ApplyRead}.
By \Claim{cl:basic:Apply}\refC{scl:ApplyRegister} \Apply$[\P_i]$ is set to value $\pair{\cWant}{s}$ only by process $\P_i$ when it executes a successful \Apply$[\P_i]$.\CAS{$\pair{\bot}{\bot},\pair{\cWant}{s}$}, therefore $\P_i$ executed the same and broke out of the spin loop of \Line{getLock:ApplyBotWant}.
Note that $\pt{\P_i}{getLock:ApplyBotWant} < \pt{r}{collect:ApplyRead} < \Ob{i} = \pt{q}{promote:FR12}$.

Since $\ctr = 0$ throughout \I{0}, $\ctr = 1$ throughout \I{1} and $\ctr = 2$ throughout \I{2}, it follows that \ctr\ is increased only at points \Ib{1} and \Ib{2} during $T$.
Since \K\ and \Q\ are the first two releasers of \L\ and they increased \ctr\ to $1$ and $2$, respectively, at \Ib{1} and \Ib{2}, respectively, it follows that no other process apart from \K\ and \Q\ increases the value of \ctr\ during $T$.
Since $\Ob{i} \geq \Ob{1} = \gamma > \Ib{2}$ (by Claims~\ref{cl:proofProperties3}\refC{I2-,L} and~\ref{cl:proofProperties3}\refC{GExists} and~\ref{cl:proofProperties4}\refC{@O1-}), $\P_i$ becomes a releaser of \L\ only after \Ib{2} (the point at which \Q\ became a releaser of \L).
Thus, $\P_i$ is not among the first two releasers of \L, thus $\P_i \notin \Set{\K,\Q}$.
Then it follows that $\P_i$ does not increase \ctr.
Therefore $\P_i$'s \ctr.\inc{} operation in \Line{getLock:IncCounter} returns value $2 = \cPawn$, and thus $\P_i$ sets $\Role[\P_i]$ to $2 = \cPawn$ in \Line{getLock:IncCounter}.
Then from the code structure $\P_i$ satisfies the if-condition of \Line{getLock:ifSoldier} and proceeds to spin in \Line{getLock:awaitAckOrCtrDecrease}.

\textbf{Case a - $\P_i$ receives a signal to abort while busy-waiting in \Line{getLock:awaitAckOrCtrDecrease}}:
Then $\P_i$ stops spinning in \Line{getLock:awaitAckOrCtrDecrease} and executes \abort{\P_i}.
Since $\P_i$ last set $\Role[\P_i]$ to \cPawn\ in \Line{getLock:IncCounter}, it then follows from the code structure that $\P_i$ proceeds to execute lines~\ref{abort:ifFlag}-\ref{abort:ifPawn}, and satisfies the if-condition of \Line{abort:ifPawn}, and then executes a \PawnSet.\cUpdate{$\P_i,s$} operation in \Line{abort:ifHead}.

Since a \PawnSet.\reset{} operation has not been executed during $[\Ib{0},\Ob{i}]$, from \Claim{cl:ifPromotedThenWaiting}, it follows that $\P_i$ did not execute a \PawnSet.\cUpdate{$\P_i,s$} operation in \Line{abort:ifHead} during $[\Ib{0},\Ob{i}]$, thus $\pt{\P_i}{abort:ifHead} > \Ob{i}$.
Since $\P_i$ has exclusive-registration access to \PawnSet\ during $[\Ob{i},\Oa{i}]$, and $p$ has not executed any of its cease-release events or reset \PawnSet\ during $[\pt{\P_i}{getLock:ApplyBotWant},\pt{\P_i}{abort:ifHead}]$, and $\pt{\P_i}{getLock:ApplyBotWant} < \Ob{i}$, it then follows that \PawnSet\ was not reset during $[\Ob{i},\pt{\P_i}{abort:ifHead}]$.
Then since the $\P_i$-th entry of \PawnSet\ was last changed to \pair{\cPro}{s} = \pair{2}{s} at \Ob{i}, it remains \pair{\cPro}{s} throughout $[\Ob{i},\pt{\P_i}{abort:ifHead}]$.
Then $\P_i$'s \PawnSet.\cUpdate{$\P_i,s$} operation in \Line{abort:ifHead} returns \false\ by the semantics of the \PawnSet\ object.
Then $p$ satisfies the if-condition of \Line{abort:ifHead}, proceeds to set $\Role[\P_i]$ to \cPPawn\ in \Line{abort:RolePPawn}, and then returns $\infty$ from its call to \abort{\P_i} and \lock{\P_i}.

\textbf{Case b - $\P_i$ does not receive a signal to abort while busy-waiting in \Line{getLock:awaitAckOrCtrDecrease}}:

Recall that process $q$ promotes $\P_i$ at $\Ob{i}$ by executing a \PawnSet.\promote{} operation in \Line{promote:FR12} that returns value \pair{\P_i}{s}, where $s \in \N$.
Since processes in the system continue to take steps, process $q$ sets its local variable $j$ to value $\P_i$ in \Line{promote:FR12}, and proceeds to fail the if-condition of \Line{promote:ifNotPromoted}, and then executes line~\ref{promote:ApplyWantOk} where $\pair{j}{seq} = \pair{\P_i}{s}$.
Then $q$ executes a \Apply$[\P_i]$.\CAS{$\pair{\cWant}{s},\pair{\cOk}{s}$} operation in \Line{promote:ApplyWantOk}.

Recall that process $r$ read value $\Apply[\P_i] = \pair{\cWant}{s}$ in \Line{collect:ApplyRead} and $\pt{\P_i}{getLock:ApplyBotWant} < \pt{r}{collect:ApplyRead} < \Ob{i} = \pt{q}{promote:FR12}$.
From an inspection of the code, \Apply[$\P_i$] can change from value \pair{\cWant}{s} only to value \pair{\cOk}{s} and from value \pair{\cOk}{s} only to value \pair{\bot}{\bot}.
Also, \Apply$[\P_i]$ can be changed from \pair{\cOk}{s} to \pair{\bot}{\bot}, only if $p$ executes line~\ref{abort:ApplyOkBot} or~\ref{release:ApplyOkBot}.
Since $p$ is spinning in \Line{getLock:awaitAckOrCtrDecrease} it follows that a \Apply$[\P_i]$.\CAS{$\pair{\cOk}{s},\pair{\bot}{\bot}$} operation is not executed during $(\Ob{i},\pt{q}{promote:ApplyWantOk})$, and thus $\Apply[\P_i] = \pair{\cWant}{s}$ throughout $(\Ob{i},\pt{q}{promote:ApplyWantOk})$.
Therefore, $q$ executes a successful \Apply$[\P_i]$.\CAS{$\pair{\cWant}{s},\pair{\cOk}{s}$} operation in \Line{promote:ApplyWantOk}, and thus $\Apply[\P_i] = \pair{\cOk}{s}$ at \pt{q}{promote:ApplyWantOk}.

Since $\P_i$ is busy-waiting in \Line{getLock:awaitAckOrCtrDecrease} for $\Apply[\P_i]$ to change to \pair{\cOk}{s}, it then follows that $\P_i$ busy-waits throughout $(\Ob{i},\pt{q}{promote:ApplyWantOk})$, and reads $\Apply[\P_i] = \pair{\cOk}{s}$ when it executes \Line{getLock:awaitAckOrCtrDecrease} for the first time after \pt{q}{promote:ApplyWantOk}.
Then $\P_i$ breaks out of the spin loop, and then from the code structure, $\P_i$ proceeds to set $\Role[\P_i]$ to \cPPawn\ in \Line{getLock:RolePPawn}, breaks out of the role-loop in \Line{getLock:EndInnerLoop}, executes line~\ref{getLock:ifQueen} and fails the if-condition of \Line{getLock:ifQueen}, and executes lines~\ref{getLock:ApplyWantOk}-\ref{getLock:end}, and returns from \lock{\P_i} in \Line{getLock:returninfty} with value $\infty$.
Note that $\Ob{i} < \pt{\P_i}{getLock:RolePPawn}$.

\textbf{Proof of Parts~\refC{PisCeaseReleaseEvent},~\refC{PisRelease} and~\refC{Oi-,Oi+helper} if Part~\refC{noReset@I0-,Oi-} for $i$ is true: }
Since $\P_i$ is the only releaser of \L\ throughout [\Ob{i},\Oa{i}) (Claim~\ref{cl:proofProperties4}\refC{R@Oi-,Oi+}), it follows from Claim~\ref{cl:variablesChangedByReleaserOnly} that $\P_i$ has exclusive write-access to objects \X\ and \LSync\ and exclusive registration-access to  \PawnSet\ throughout [\Ob{i},\Oa{i}).

Since $\P_i$ returns from its call to \lock{\P_i} with value $\infty$ (by Part~\refC{PiDoesNotStarve}), $\P_i$ executes a call to \release{\P_i}{} (follows from conditions~\ref{condition:ifLockThenRelease} and~\ref{condition:ifReleaseThenExistsLock}).

Since $\Role[\P_i] = \cPPawn$ when $\P_i$'s call to \lock{\P_i} returns (by Part~\refC{PiDoesNotStarve}), $\Role[\P_i] = \cPPawn$ at \ptB{\P_i}{release:safetyCheck}.
Since $\Role[\P_i]$ is unchanged during $[\pt{\P_i}{release:safetyCheck},\ptB{\P_i}{release:ApplyOkBot}]$ (follows from Claim~\ref{cl:basic:Role}\refC{scl:RoleUnchanged}), it follows from the code structure that during $\P_i$'s call to \release{\P_i}{j}, $\P_i$ only executes lines~\ref{release:safetyCheck}-\ref{release:ifKing},~\ref{release:ifQueen} and~\ref{release:ifPPawn}-\ref{release:return}.
Then Figure~\ref{fig:PisRelease} follows.

From an inspection of Figures~\ref{fig:PisGetLock} and~\ref{fig:PisRelease}, $\P_i$ does not execute a call to \helpRelease{\P_i} or execute a \ctr.\CAS{$1,0$} operation in \Line{release:ctr10} during \release{\P_i}{}.
Then from Claims~\ref{cl:basic:ReleaseEvents}\refC{scl:def:phi}  and~\ref{cl:basic:ReleaseEvents}\refC{scl:def:tau} $\P_i$'s cease-release events $\phi_{\P_i}$ and $\tau_{\P_i}$ do not occur.
Since $\P_i$ executes a call to \doPromote{\P_i} only in \Line{release:callPromote}, it follows from \Claim{cl:EventsPromote} that exactly one cease-release event among $\pi_{\P_i}$ and $\theta_{\P_i}$ occurs during $\P_i$'s call to \doPromote{\P_i}.
Hence, Part~\refC{PisCeaseReleaseEvent} follows.
Then \Oa{i} is the point when cease-release event $\pi_{\P_i}$ or $\theta_{\P_i}$ occurs.
From an inspection of Figures~\ref{fig:PisGetLock} and~\ref{fig:PisRelease} and the code, it is clear that $\P_i$ does not change \X\ or \LSync\ during \lock{\P_i} and \release{\P_i}.
Therefore, $\P_i$ does not change \X\ or \LSync\ during $[\Ob{i},\Oa{i}]$.


\textbf{Proof of Part~\refC{whenOioccurs} if Part~\refC{noReset@I0-,Oi-} for $i$ is true: }
As argued in Part~\refC{PisGetLock} and~\refC{PiDoesNotStarve}, $\pt{\P_i}{getLock:IncCounter} < \Ob{i}$, and $\Ob{i} < \pt{\P_i}{getLock:RolePPawn}$ or $\Ob{i} < \pt{\P_i}{abort:ifHead}$.
Since $\pt{\P_i}{getLock:RolePPawn} < \pt{\P_i}{release:safetyCheck}$ and $ \pt{\P_i}{abort:ifHead} < \pt{\P_i}{release:safetyCheck}$, it then follows that $\pt{\P_i}{getLock:IncCounter} < \Ob{i} < \pt{\P_i}{release:safetyCheck}$.

From Part~\refC{PisCeaseReleaseEvent}, exactly one cease-release event among $\pi_{\P_i}$ and $\theta_{\P_i}$ occurs during $\P_i$'s call to \doPromote{\P_i}.
If cease-release event $\theta_{\P_i}$ occurs then \Oa{i} is the point when $\P_i$'s cease-release event $\theta_{\P_i}$ occurs,i.e, $\Oa{i} = \pt{\P_i}{promote:ctr20}$.
Then $\P_i$ changes \ctr\ to $0$ and the \ctr-cycle interval $T$ ends at $\Oa{i} = \pt{\P_i}{promote:ctr20} = \Ia{2}$.

If cease-release event $\pi_{\P_i}$ occurs then \Oa{i} is the point when $\P_i$'s cease-release event $\pi_{\P_i}$ occurs,i.e,  $\Oa{i} = \pt{\P_i}{promote:FR12} < \Ia{2}$.

Since $\P_i$ calls \doPromote{\P_i} only in \Line{release:callPromote} (by inspection of \Figure{fig:PisRelease}), it then follows that $\Oa{i} \in \Set{\pt{\P_i}{promote:FR12},\pt{\P_i}{promote:ctr20}} < \pt{\P_i}{release:ApplyOkBot}$.
Thus, Part~\refC{whenOioccurs} holds.

\textbf{Proof of Part~\refC{Oi-,Oi+helper2} if Part~\refC{noReset@I0-,Oi-} for $i$ is true: }
As argued in Part~\refC{Oi-,Oi+helper}, exactly one cease-release event among $\pi_{\P_i}$ and $\theta_{\P_i}$ occurs during $\P_i$'s call to \doPromote{\P_i}.
If cease-release event $\theta_{\P_i}$ occurs then \Oa{i} is the point when $\P_i$'s cease-release event $\theta_{\P_i}$ occurs,i.e, $\Oa{i} = \pt{\P_i}{promote:ctr20}$.
Then $\P_i$ changes \ctr\ to $0$ and the \ctr-cycle interval $T$ ends at $\Oa{i} = \pt{\P_i}{promote:ctr20}$, and thus $\ell = i$.
This is a contradiction to the assumption $i \neq \ell$, hence  $\P_i$'s cease-release event $\pi_{\P_i}$ occurs during $\P_i$'s call to \doPromote{\P_i}.
Then \Oa{i} is the point when $\P_i$'s cease-release event $\pi_{\P_i}$ occurs,i.e, $\Oa{i} = \pt{\P_i}{promote:FR12}$.
From an inspection of Figures~\ref{fig:PisGetLock} and~\ref{fig:PisRelease} and the code, it follows that $\P_i$ does not execute a \PawnSet.\reset{} operation during $[\pt{\P_i}{getLock:ApplyBotWant},\ptB{\P_i}{release:callPromote}]$, and $\P_i$ calls \doPromote{\P_i} only in \Line{release:callPromote}.
Since $\Oa{i} = \pt{\P_i}{promote:FR12}$, from an inspection of the code of \doPromote{\P_i}, $\P_i$ does not execute a \PawnSet.\reset{} operation during $[\ptB{\P_i}{promote:FR12},\ptB{\P_i}{promote:ctr20}]$.
Then $\P_i$ does not execute a \PawnSet.\reset{} operation during $[\Ob{i},\Oa{i}]$.

Since $\P_i$ is the only releaser of \L\ throughout [\Ob{i},\Oa{i}) (Claim~\ref{cl:proofProperties4}\refC{R@Oi-,Oi+}), it follows from Claim~\ref{cl:variablesChangedByReleaserOnly} that $\P_i$ has exclusive registration-access to  \PawnSet\ throughout [\Ob{i},\Oa{i}).
Then since no \PawnSet.\reset{} operation was executed during $[\Ib{0},\Ob{i}]$, and $\P_i$ does not execute a \PawnSet.\reset{} operation during $[\Ob{i},\Oa{i}]$, it follows that no \PawnSet.\reset{} operation is executed during $[\Ib{0},\Oa{i}]$.
Hence, Part~\refC{Oi-,Oi+helper2} holds.

Finally, we show that if Parts~\refC{noReset@I0-,Oi-}-\refC{Oi-,Oi+helper2} are true for $i$, then Part~\refC{noReset@I0-,Oi-} is true for $i+1$, thus completing the proof.
From Part~\refC{Oi-,Oi+helper2} for $i$, no \PawnSet.\reset{} operation has been executed during $[\Ib{0},\Oa{i}]$.
From Claim~\ref{cl:proofProperties4}\refC{Oi+=Oi+1-}, $\Oa{i} = \Ob{i+1}$.
Then Part~\refC{noReset@I0-,Oi-} for $i+1$ holds.

\textbf{Proof of~\refC{XTC@Oi-,Oi+}: }
From Claim~\ref{cl:proofProperties4}\refC{@O1-}, $\X = \LSync = \bot$ at \Ob{1} = $\gamma$.
From Claims~\ref{cl:proofProperties4}\refC{@O1-} and~\ref{cl:proofProperties4}\refC{Oi+=Oi+1-}, it follows that $\gamma = \Ob{1} < \Oa{1} = \Ob{2} < \Oa{2} = \Ob{3} \ldots < \Oa{\ell-1} = \Ob{\ell} < \Oa{\ell}$.

From Claim~\ref{cl:proofProperties4}\refC{R@Oi-,Oi+}, for all $t \in [\Ob{i},\Oa{i})$, $R(t) \in \Set{\P_i}$.
Then $\P_i$ has exclusive write-access to \X\ and \LSync\ throughout $[\Ob{i},\Oa{i})$.
Since $\P_i$ does not change \X\ or \LSync\ during $[\Ob{i},\Oa{i}]$ (Part~\refC{Oi-,Oi+helper}), it then follows that $\X= \LSync = \bot$ throughout $[\Ob{1},\Oa{\ell}] = [\gamma,\Oa{\ell}]$.

\textbf{Proof of~\refC{I2+=Ol+}: }
As argued in Part~\refC{Oi-,Oi+helper}, exactly one cease-release event among $\pi_{\P_{\ell}}$ and $\theta_{\P_{\ell}}$ occurs during $\P_{\ell}$'s call to \doPromote{\P_{\ell}}.
If cease-release event $\pi_{\P_{\ell}}$ occurs then $\P_{\ell}$ promotes some process, and thus the number of processes that get promoted during $T$ is larger than $\ell$, which contradicts the definition of $\ell$.
Hence, cease-release event $\theta_{\P_{\ell}}$ occurs during \doPromote{\P_{\ell}} and \Oa{{\ell}} is the point when cease-release event $\theta_{\P_{\ell}}$ occurs,i.e, $\Oa{\ell} = \pt{\P_{\ell}}{promote:FR12}$.
Since \ctr\ is changed from $2$ to $0$ when $\theta_{\P_{\ell}}$ occurs, the \ctr-cycle interval $T$ ends at $\Oa{\ell} = \pt{\P_{\ell}}{promote:ctr20}$, and thus $\Ia{2} = \Oa{\ell} = \pt{\P_{\ell}}{promote:ctr20}$.

\textbf{Proof of~\refC{[G,I2+]} and~\refC{@I2+}: }

\textbf{Case a - } $\ell = 0$ :
Consider the first \PawnSet.\promote{} operation at $\gamma$.
Since $\ell = 0$, the \PawnSet.\promote{} operation at $\gamma$ returns value \pair{\bot}{\bot}.
Then from \Claim{cl:proofProperties3}\refC{noPromotion@G}, it follows that \B's cease-release event $\theta_{\B}$ occurs at $t' = \pt{\B}{promote:ctr20} \geq \gamma$, and throughout $[\gamma,t']$ no process is promoted, and for all $t \in [\gamma,t')$, $R(t) = \Set{\B}$.
Since \ctr\ is changed from $2$ to $0$ when $\theta_{\B}$ occurs, the \ctr-cycle interval $T$ ends at $t' =\pt{\B}{promote:ctr20}$, and thus $\Ia{2} = \pt{\B}{promote:ctr20} = t'$.
Then for all $t \in [\gamma,t') = [\gamma,\Ia{2})$, $|R(t)| = 1$.

From an inspection of Figure~\ref{fig:BsPromote} and the code, it follows that $\B$ executed a \PawnSet.\reset{} operation in \Line{promote:resetBackpack} during $[\gamma,t']$, and thus \PawnSet\ is candidate-empty immediately after.
Since for all $t \in [\gamma,t')$, $R(t) = \Set{\B}$, \B\ has exclusive registration-access to \PawnSet\ throughout $[\gamma,t')$ (follows from \Claim{cl:variablesChangedByReleaserOnly}).
Then it follows that \PawnSet\ is candidate-empty at $t' = \Ia{2}$.

Since for all $t \in [\gamma,t')$, $R(t) = \Set{\B}$, \B\ has exclusive write-access to \X\ and \LSync\ throughout $[\gamma,t')$ (follows from \Claim{cl:variablesChangedByReleaserOnly}).
Since $\X= \LSync = \bot$ at $\gamma$ (Claim~\ref{cl:proofProperties3}\refC{@G}), and \B\ does not write to \X\ and \LSync\ during $[\gamma,t']$, it follows that $\X= \LSync = \bot$ throughout $[\gamma,t'] = [\gamma,\Ia{2}]$.

\textbf{Case b - } $\ell \geq 1$ :
From Part~\refC{I2+=Ol+}, $\Ia{2} = \Oa{\ell} =  \pt{\P_{\ell}}{promote:ctr20}$.
Then from Part~\refC{XTC@Oi-,Oi+}, it follows that $\X = \LSync = \bot$ throughout $[\gamma,\Ia{2}]$, and from Claim~\ref{cl:proofProperties4}\refC{R@Oi-,Oi+}, it follows that for all $t \in [\Ob{1},\Oa{\ell}] = [\gamma,\Ia{2})$, $|R(t)| = 1$.
Since $\P_{\ell}$ ceases to be a releaser of \L\ at \Oa{\ell}, $R(\Ia{2}) = \varnothing$.

Since $\Oa{\ell} =  \pt{\P_{\ell}}{promote:ctr20}$, $\P_i$ executed \Line{promote:ctr20} and before that \Line{promote:resetBackpack}.
Hence, $\P_{\ell}$ executed a \PawnSet.\reset{} operation at $\pt{\P_{\ell}}{promote:resetBackpack} < \Oa{\ell}$.
Since $\pt{\P_{\ell}}{promote:resetBackpack} > \ptB{\P_{\ell}}{release:safetyCheck}$ and
$\ptB{\P_{\ell}}{release:safetyCheck}> \Ob{\ell}$ (by Part~\refC{whenOioccurs}), it follows that $\pt{\P_{\ell}}{promote:resetBackpack} > \Ob{\ell}$.
Hence, $\P_{\ell}$ executed a \PawnSet.\reset{} operation at $\pt{\P_{\ell}}{promote:resetBackpack} \in [\Ob{\ell},\Oa{\ell}]$.
Since $\P_{\ell}$ is the only releaser of \L\ throughout [\Ob{\ell},\Oa{\ell}) (Claim~\ref{cl:proofProperties4}\refC{R@Oi-,Oi+}), it follows from Claim~\ref{cl:variablesChangedByReleaserOnly} that $\P_{\ell}$ has exclusive registration-access to \PawnSet\ throughout [\Ob{\ell},\Oa{\ell}).
Then it follows that \PawnSet\ is candidate-empty at $\Ob{\ell} = \Ia{2}$.
\end{proof}

\begin{claim} \label{cl:@I0-}
$R(\Ib{0}) = \varnothing$ and at $\Ib{0}$, $\X = \LSync = \bot$ and \PawnSet\ is candidate-empty for any \ctr-cycle interval $T$ during history $H$.
\end{claim}
\begin{proof}
Let $T^k$ denote the $k$-th \ctr-cycle interval $T$ during history $H$.
We give a proof by induction over the integer $k$.
\textbf{Basis - }
At $\Ib{0}$ for $T^1$, the claim holds trivially since all variables are at their initial values ($\X = \LSync = \bot$ and \PawnSet\ is candidate-empty).

\textbf{Induction Step - }
By the induction hypothesis, at $\Ib{0}$ for $T^{k-1}$, $R(\Ib{0}) = \varnothing$, and $\X = \LSync = \bot$ and \PawnSet\ is candidate-empty.
Since $T^k$ begins immediately after $T^{k-1}$ ends, to prove our claim we need to show that, when $T^{k-1}$ ends, there are no releasers of \L\ and $\X = \LSync = \bot$ and \PawnSet\ is candidate-empty.
The time interval $T^{k-1}$ ends  either at time \Ia{1} or time \Ia{2}.

\textbf{Case a - $T^{k-1}$ ends at time \Ia{1}: }
Then $\I{2} = \varnothing$.
From Claim~\ref{cl:proofProperties1}\refC{R@I1} it follows that \K\ is the only releaser of \L\ during \I{1}.
Since $\I{2} = \varnothing$, it then follows from Claim~\ref{cl:proofProperties1}\refC{aExecutesctr10}, that \K's \ctr.\CAS{$1,0$} operation in \Line{release:ctr10} is successful, and the interval $\I{1}$ as well as $T^{k-1}$ ends at time \pt{\K}{release:ctr10}.
Then \K's cease-release event $\phi_{\K}$ occurs at $\pt{\K}{release:ctr10} = \Ia{1}$, and thus there are no releasers of \L\ immediately after $T^{k-1}$ ends.
And from Claim~\ref{cl:proofProperties1}\refC{XT@I1}, it follows that $\X = \LSync = \bot$ and \PawnSet\ is candidate-empty when $T^{k-1}$ ends.

\textbf{Case b - $T^{k-1}$ ends at time \Ia{2}: }
Then $\I{2} \neq \varnothing$.
Then our proof obligation follows immediately from Claim~\ref{cl:proofProperties5}\refC{@I2+}.
\end{proof}

Note that in the following claims, notations \I{0}, \I{1}, \I{2}, $\lambda, \gamma$, $\Omega_i$, \K, \Q\ and $\P_i$ are defined relative to a \ctr-cycle interval, as was defined previously in pages \pageref{sec:IntervalT}, \pageref{notations:lambda,gamma} and \pageref{notations:Omega,P}.
The exact \ctr-cycle interval is clear from the context of the discussion.

\begin{lemma} \label{cl:ARMLock:mutualExclusion}
\label{cor:ARMLock:mutualExclusion}
The mutual exclusion property holds during history $H$.
\end{lemma}
\begin{proof}
For the purpose of a contradiction assume that at time $t$, two processes (say $p$ and $q$) are poised to execute a call to \L.\release{}{}.
From \Claim{cl:ifReleasinglockThenReleaser}\refC{scl:isReleaser}, it follows that both $p$ and $q$ are releasers of \L\ at $t$.
Consider the \ctr-cycle interval $T$ such that $t \in T$.

From \Claim{cl:@I0-} it follows that at $\Ib{0}$, $\X = \LSync = \bot$ and \PawnSet\ is candidate-empty, and $R(\Ib{0}) = \varnothing$.
Then from  Claims~\ref{cl:proofProperties1}\refC{I0},~\ref{cl:proofProperties1}\refC{R@I1}, ~\ref{cl:proofProperties3}\refC{I2-,L},~\ref{cl:proofProperties3}\refC{L,G} and~\ref{cl:proofProperties3}\refC{[G,I2+]}, it follows that during $T$, lock \L\ has two releasers only during $[\Ib{2},\lambda)$.
Then $t \in [\Ib{2},\lambda)$.
Also from  Claim~\ref{cl:proofProperties3}\refC{I2-,L}, for all $t \in [\Ib{2},\lambda)$, $R(t) = \Set{\K,\Q}$.
Then \Set{p,q} = \Set{\K,\Q}.
Let $p = \K$ and $q = \Q$ without loss of generality.

Recall that \Ib{2} is the point in time when \Q\ increases \ctr\ from $1$ to $2$ and sets \Role[\Q] to \cQueen\ in \Line{getLock:IncCounter}.
Since $q$'s call to \lock{} returned a non-$\bot$ value, it follows from an inspection of \Figure{fig:QsGetLock}, that \Q\ returned either in \Line{getLock:returnX} or \Line{abort:returnX}.
Then \Q\ either read a non-$\bot$ value from \X\ in \Line{getLock:awaitX} or \Q\ failed the \X.\CAS{$\bot,\infty$} operation in \Line{abort:setX}.
Since $\X = \bot$  at \Ib{2} (by Claim~\ref{cl:proofProperties2}\refC{@I2-}), and $\Ib{2} = \pt{\Q}{getLock:IncCounter}$, it then follows that \X\ is changed to a non-$\bot$ value during $[\Ib{2},t]$.
Clearly, \Q\ does not change \X\ during $[\Ib{2},t]$.

Recall that \Ib{1} is the point in time when \K\ increases \ctr\ from $0$ to $1$ and sets $\Role[\K]$ to \cKing\ in \Line{getLock:IncCounter}.
It follows from an inspection of \Figure{fig:KsGetLock}, that \K\ does not change \X\ during \lock{\K}, and thus during $[\Ib{1},t]$.
Since \X\ is changed to a non-$\bot$ only by a releaser of \L\ (by Claim~\ref{cl:variablesChangedByReleaserOnly}) and $\X=\bot$ at \Ib{2}, and the only releasers of \L\ during $[\Ib{2},t]$ do not change \X, it then follows that $\X=\bot$ throughout $[\Ib{2},t]$.
Hence, a contradiction.
\end{proof}


%

\begin{claim} \label{cl:DoggedPGetsPromoted}
Consider an arbitrary \ctr-cycle interval $T$.
\begin{enumerate}[(a)]
 \item If $p$ is collected during $T$ and $p$ does not abort, then $p$ is promoted and notified during $T$. \label{scl:DoggedPGetsPromoted1}
  \item If $\Apply[p] =\pair{\cWant}{s}$  at \Ib{0}, where $s \in N$, and $p$ does not abort and $p$ does not increase \ctr, then $p$ is notified during $T$. \label{scl:DoggedPGetsPromoted2}
\end{enumerate}
\end{claim}
\begin{proof}
\textbf{Proof of~\refC{scl:DoggedPGetsPromoted1}: }
From \Claim{cl:@I0-} it follows that at $\Ib{0}$, $\X = \LSync = \bot$ and \PawnSet\ is candidate-empty, and $R(\Ib{0}) = \varnothing$.
Then from \Claim{cl:proofProperties3}\refC{aOrbCollect}, it follows that exactly one call to \doCollect{} is executed during $T$ by a process $q \in \Set{\K,\Q}$.
Since processes are collected only during a call to \doCollect{}, $q \in \Set{\K,\Q}$ collects $p$ during \doCollect{q} during $T$.
And $q$ does so by executing a \PawnSet.\collect{$A$} operation in \Line{collect:updateAll}, where $A[p] = s  \in \N$, and sets the $p$-th entry of \PawnSet\ to \pair{\cReg}{s}.
Since a \PawnSet.\promote{} that returns \pair{\bot}{\bot} is executed at \pt{\P_{\ell}}{promote:FR12} during $T$, it then follows from the semantics of the \PawnSet\ object that $p$ was promoted during $T$.
Then $p = \P_i$, for some $i \leq \ell$.
Note that $T$ does not end during $[\Ob{i},\Oa{i})$.

We now show that $p$ is also notified of its promotion during $T$.
The process (say $r$) that promoted $p$ by executing a \PawnSet.\promote{} operation in \Line{promote:FR12}, also goes on to notify $p$ of its promotion by executing a \Apply$[p]$.\CAS{\pair{\cWant}{s},\pair{\cOk}{s}} operation in \Line{promote:ApplyWantOk}.
Since $p$ does not abort, it follows from an inspection of Figure~\ref{fig:PisGetLock} and the code, that $p$ spins on \Apply$[p]$ in \Line{getLock:awaitAckOrCtrDecrease} until its notification.
Then $p$ executes \Line{getLock:RolePPawn} at $\pt{p}{getLock:RolePPawn} > \pt{r}{promote:ApplyWantOk} > \pt{r}{promote:FR12} = \Ob{i}$.
Since $\pt{p}{getLock:RolePPawn} < \Oa{i}$ and $T$ does not end before \Oa{i}, it follows that $p$ is notified during $T$.

\textbf{Proof of~\refC{scl:DoggedPGetsPromoted2}: }
Since $p$ does not increase \ctr\ it follows that $p$ reads $\ctr = 2$  every time it executes a \ctr.\inc{} operation in \Line{getLock:IncCounter}, and sets $\Role[p] = \cPawn$ in \Line{getLock:IncCounter}.
Then $p$ satisfies the if-condition of \Line{getLock:ifSoldier} and spins on variables \Apply$[p]$ and \ctr\ in line~\ref{getLock:awaitAckOrCtrDecrease}.
Since \ctr\ is only changed to $0$ at the end of $T$, it follows that $\ctr=2$ throughout $[\pt{p}{getLock:IncCounter},\Ia{2})$.
Then $p$ busy-waits in the spin loop of \Line{getLock:awaitAckOrCtrDecrease} until the end of $T$, or if it reads value \pair{\cOk}{s}, for some $s \in \N$, from \Apply$[p]$ in \Line{getLock:awaitAckOrCtrDecrease} during $T$.
Now, \Apply$[p]$ is changed to value \pair{\cOk}{s} by some process other than $p$, only if that process notifies $p$, i.e., executes a successful \Apply$[p]$.\CAS{\pair{\cWant}{s},\pair{\cOk}{s}} operation in \Line{promote:ApplyWantOk}.
We now show that $p$ is notified during $T$.

From \Claim{cl:@I0-} it follows that at $\Ib{0}$, $\X = \LSync = \bot$ and \PawnSet\ is candidate-empty, and $R(\Ib{0}) = \varnothing$.
Then from \Claim{cl:proofProperties3}\refC{aOrbCollect}, it follows that exactly one call to \doCollect{} is executed during $T$ by a process $q \in \Set{\K,\Q}$.
Consider the point when $q$ reads $\Apply[p]$ in \Line{collect:ApplyRead}.
If $q$ reads a value different from $\pair{\cWant}{s}$, then some process must have notified $p$ during $[\pt{p}{getLock:ApplyBotWant},\pt{q}{collect:ApplyRead}]$, and since $\Ib{0} < \pt{p}{getLock:ApplyBotWant}$ and $\pt{q}{collect:ApplyRead} \in T$, our claim holds.
If $q$ reads the value $\pair{\cWant}{s}$ from \Apply[p], then $q$ collects $p$ during $T$ by executing a \PawnSet.\collect{$A$} operation, where $A[p] = s$, in \Line{collect:updateAll} during $T$.
Thus, our claim follows from Part~\refC{scl:DoggedPGetsPromoted1}.
%
%
%
%
%
%
\end{proof}

%

\begin{claim} \label{cl:conditional:getLock}
If $p$ registered itself in \Line{getLock:ApplyBotWant}, and incurred \Order{1} RMRs in the process, and $p$ does not abort, and all processes in the system continue to take steps, then
\begin{enumerate}[(a)]
 \item $p$ finishes its call to \lock{p} and returns a non-$\bot$ value.  \label{scl:conditional:pDoesNotStarve}
 \item $p$ incurs \Order{1} RMRs in expectation during its call to \lock{p}.
 \label{scl:conditional:eRMRcomplexity}
\end{enumerate}
\end{claim}
\begin{proof}
\textbf{Proof of~\refC{scl:pDoesNotStarve} and~\refC{scl:eRMRcomplexity}: }
From an inspection of the code of \lock{p}, $p$ incurs a constant number of RMRs while executing all other lines of \lock{p} except while busy-waiting in lines~\ref{getLock:ApplyBotWant},~\ref{getLock:awaitAckOrCtrDecrease} and~\ref{getLock:awaitX}.

Consider $p$'s call to \lock{p}.
By assumption of the claim, $p$ registered itself in \Line{getLock:ApplyBotWant} by executing a successful \Apply$[p]$.\CAS{\pair{\bot}{\bot},\pair{\cWant}{s}} operation in \Line{getLock:ApplyBotWant}, and incurred \Order{1} RMRs in the process.
Then $p$ proceeds to execute a \ctr.\inc{} operation in \Line{getLock:IncCounter}, and stores the returned value in  $\Role[p]$.
A \ctr.\inc{} operation returns values in \Set{\cKing,\cQueen,\cPawn,\bot}.
If it returns $\bot$, $p$ repeats the role-loop, and executes another \ctr.\inc{} operation in \Line{getLock:IncCounter}.
From Claim~\ref{claim:CASCounterFailureProbability}, it follows that $p$ repeats the role-loop only a constant number of times before its \ctr.\inc{} operation returns a non-$\bot$ value.

\textbf{Case a - } $p$ executes a \ctr.\inc{} operation in \Line{getLock:IncCounter} that returns \cKing.
Then $p$ sets $\Role[p] = \cKing$ in \Line{getLock:IncCounter}.
Then from the code structure $p$ does not busy-wait on any variables, and proceeds to return $\infty$ in \Line{getLock:returninfty}, and thus incurs only \Order{1} RMRs.
Hence,~\refC{scl:pDoesNotStarve} and~\refC{scl:eRMRcomplexity} hold.

\textbf{Case b - } $p$ executes a \ctr.\inc{} operation in \Line{getLock:IncCounter} that returns \cQueen.
Then $p$ increments $\ctr$ from $1$ to $2$ in \Line{getLock:IncCounter} and sets $\Role[p] = \cQueen$ in \Line{getLock:IncCounter}.
Then from the code structure $p$ proceeds to busy-wait on \X\ in \Line{getLock:awaitX}.
Since $p$ increased $\ctr$ from $1$ to $2$, $\pt{p}{getLock:awaitX} = \Ib{2}$ for some \ctr-cycle interval $T$.
From \Claim{cl:@I0-} it follows that at $\Ib{0}$, $\X = \LSync = \bot$ and \PawnSet\ is candidate-empty, and $R(\Ib{0}) = \varnothing$.
Then from Claim~\ref{cl:proofProperties2}\refC{QsGetLock} it follows that $p$ does not starve in \Line{getLock:awaitX}.
Since $p$ does not abort, it follows from an inspection of \Figure{fig:QsGetLock} and the code, that $p$ returns a non-$\bot$ value in \Line{getLock:returnX}, and $p$ does not change \X.
Hence, we have shown that Part~\refC{scl:pDoesNotStarve} holds.
Apart from $p$, the only releasers of \L\ during $T$ are \Set{\K,\P_1, \ldots, \P_{\ell}}, where $\ell$ is the number of promotions during $T$.
From an inspection of Figures~\ref{fig:KsGetLock},~\ref{fig:KsRelease},~\ref{fig:PisGetLock},~\ref{fig:PisRelease} and the code, it follows that only $\K$ possibly writes a non-$\bot$ value to \X\ during $T$ in \Line{release:setX}.
Since \X\ is written to only be a releaser of \L, and $\pt{p}{getLock:awaitX} \in T$, it then follows that \X\ is changed to a non-$\bot$ value at most once during $T$.
Then $p$ incurs at most one RMR while busy-waiting on \X.
Hence, we have shown that Part~\refC{scl:eRMRcomplexity} holds.

\textbf{Case c - } $p$ executes a \ctr.\inc{} operation in \Line{getLock:IncCounter} that returns \cPawn.
Then $p$ found $\ctr$ to be $2$ in \Line{getLock:IncCounter} and set $\Role[p] = \cPawn$ in \Line{getLock:IncCounter}.
Then from the code structure $p$ proceeds to busy-wait on \Apply$[p]$ and \ctr\ in \Line{getLock:awaitAckOrCtrDecrease}.

We now show that $p$ does not starve while busy-waiting in \Line{getLock:awaitAckOrCtrDecrease}.
Since $\ctr = 2$ at $\pt{p}{getLock:IncCounter}$, it follows that $\pt{p}{getLock:IncCounter} \in T$ for some \ctr-cycle interval $T$.

\textbf{Subcase (i) - } $\Apply[p] =\pair{\cWant}{s}$  at \Ib{0} during $T$, for some $s \in \N$.
Then from \Claim{cl:DoggedPGetsPromoted}\refC{scl:DoggedPGetsPromoted2}, $p$ is notified during $T$.
Since $p$ is notified during $T$ and $p$ does not abort, it follows that $p$ does not change \Apply$[p]$, and thus \Apply$[p]$ is changed from \pair{\cWant}{s} to \pair{\cOk}{s} when $p$ is notified.
Since \Apply$[p]$ is changed from \pair{\cOk}{s} to some other value only by $p$, it then follows that \Apply$[p]$ remains \pair{\cOk}{s} when $p$ reads \Apply$[p]$ for the first time after $p$ was notified.
Then $p$ incurs one RMR when it reads \Apply$[p]$ in \Line{getLock:awaitAckOrCtrDecrease} after its notification, breaks out of the spin loop of \Line{getLock:awaitAckOrCtrDecrease}, proceeds to satisfy the if-condition of \Line{getLock:ifBackpacked}, and sets $\Role[p] = \cPPawn$ in \Line{getLock:RolePPawn}, and proceeds to return $\infty$ in \Line{getLock:returninfty}.
Then we have shown Parts~\refC{scl:pDoesNotStarve} and~\refC{scl:eRMRcomplexity} hold.

\textbf{Subcase (ii) - } $\Apply[p] \neq \pair{\cWant}{s}$  at \Ib{0} during $T$, for some $s \in \N$.
Consider the only call to \doCollect{} during $T$ by $q \in \Set{\K,\Q}$.
If $p$ registered itself (i.e., executed its \Apply$[p]$.\CAS{\pair{\bot}{\bot},\pair{\cWant}{s}} operation in \Line{getLock:ApplyBotWant}) before $q$ reads \Apply$[p]$ in \Line{collect:ApplyRead} during \doCollect{q}), then $q$ collects $p$ during $T$.
Then from \Claim{cl:DoggedPGetsPromoted}\refC{scl:DoggedPGetsPromoted1}, $p$ is collected and promoted during $T$, and eventually notified.
Then Parts~\refC{scl:pDoesNotStarve} and~\refC{scl:eRMRcomplexity} hold as argued in \textbf{Subcase (i)}.

If $p$ registers itself after $q$ attempts to acknowledge $p$ during $T$, then no process changes \Apply$[p]$ during $T$.
Then $p$ continues to busy-wait in \Line{getLock:awaitAckOrCtrDecrease}, until the \ctr-cycle interval $T$ ends and \ctr\ is reset to $0$.

If \ctr\ is increased to $2$ before $p$ reads \ctr\ again in \Line{getLock:awaitAckOrCtrDecrease}, then let $T'$ be the \ctr-cycle interval that starts when \ctr\ was reset to $0$ at the end of $T$.
Since $\Apply[p]$ was changed to a non-\pair{\cWant}{s} value before the start of $T'$, it follows that $\Apply[p] = \pair{\cWant}{s}$ at the start of $T'$.
Then from \Claim{cl:DoggedPGetsPromoted}\refC{scl:DoggedPGetsPromoted2}, $p$ is acknowledged, collected, promoted during $T'$, and eventually notified.
Then Parts~\refC{scl:pDoesNotStarve} and~\refC{scl:eRMRcomplexity} hold as argued in \textbf{Subcase (i)}.

If $\ctr \neq 2$ when $p$ reads \ctr\ again in \Line{getLock:awaitAckOrCtrDecrease}, then $p$ incurs one RMR in \Line{getLock:awaitAckOrCtrDecrease}, breaks out of the spin loop, and proceeds to execute \Line{getLock:ifBackpacked}.
If $p$ satisfies the if-condition of \Line{getLock:ifBackpacked}, then $p$ has been acknowledged during some \ctr-cycle interval $T''$.
Then from \Claim{cl:DoggedPGetsPromoted}\refC{scl:DoggedPGetsPromoted1}, $p$ is collected, promoted during $T''$, and eventually notified.
Then Parts~\refC{scl:pDoesNotStarve} and~\refC{scl:eRMRcomplexity} hold as argued in \textbf{Subcase (i)}.
If $p$ fails the if-condition of \Line{getLock:ifBackpacked}, then $p$ proceeds to repeat the role-loop.
Consider $p$'s second iteration of the role-loop.
If $p$ sets $\Role[p] = \Set{\cKing,\cQueen}$ in \Line{getLock:IncCounter}, then Parts~\refC{scl:pDoesNotStarve} and~\refC{scl:eRMRcomplexity} hold as argued in \textbf{Case a} and \textbf{Case b}.
If $p$ sets $\Role[p] = \cPawn$ in \Line{getLock:IncCounter}, then it follows that $\pt{p}{getLock:IncCounter} \in T'''$, for some \ctr-cycle interval $T'''$, such that $\Apply[p] = \pair{\cWant}{s}$ at \Ib{0} for $T'''$.
Parts~\refC{scl:pDoesNotStarve} and~\refC{scl:eRMRcomplexity} hold as argued in \textbf{Case c(i)}.
\end{proof}

\begin{lemma} \label{cl:getLock}
If all processes in the system continue to take steps and $p$ does not abort, then
\begin{enumerate}[(a)]
 \item $p$ finishes its call to \lock{p} and returns a non-$\bot$ value.  \label{scl:pDoesNotStarve}
 \item $p$ incurs \Order{1} RMRs in expectation during its call to \lock{p}.
 \label{scl:eRMRcomplexity}
\end{enumerate}
\end{lemma}
\begin{proof}
From an inspection of \lock{p}, $p$ incurs a constant number of RMRs while executing all other lines of \lock{p} except while busy-waiting in lines~\ref{getLock:ApplyBotWant}, ~\ref{getLock:awaitAckOrCtrDecrease} and~\ref{getLock:awaitX}.

Consider $p$'s call to \lock{p}.
Process $p$ first attempts to register itself in \Line{getLock:ApplyBotWant}, by attempting to execute an \Apply$[p]$.\CAS{\pair{\bot}{\bot},\pair{\cWant}{s}} operation.
Now, \Apply$[p]$ is changed from \pair{\bot}{\bot} to a non-\pair{\bot}{\bot} value only by $p$ (Claim~\ref{cl:basic:Apply}\refC{scl:ApplyRegister}).
If $\Apply[p] = \pair{\bot}{\bot}$ at \ptB{p}{getLock:ApplyBotWant}, then $p$ executes a successful \Apply$[p]$.\CAS{$\pair{\bot}{\bot},\pair{\cWant}{s}$} operation in \Line{getLock:ApplyBotWant} and incurs only one RMR.
Then our claims follow immediately from Claims~\ref{cl:conditional:getLock}\refC{scl:conditional:pDoesNotStarve} and~\ref{cl:conditional:getLock}\refC{scl:conditional:eRMRcomplexity}.

If $\Apply[p] \neq \pair{\bot}{\bot}$ at \ptB{p}{getLock:ApplyBotWant}, it follows that some process $p'$ executed a successful \Apply$[p]$.\CAS{$\pair{\bot}{\bot},\pair{\cWant}{s'}$} in \Line{getLock:ApplyBotWant} during \lock{p}, and $\Apply[p] \neq \pair{\bot}{\bot}$ throughout $[\pt{p'}{getLock:ApplyBotWant},\ptB{p}{getLock:ApplyBotWant}]$.
Since calls to \lock{p} are not executed concurrently, it follows that $p'$ has completed its call to \lock{p} during $[\pt{p'}{getLock:ApplyBotWant},\ptB{p}{getLock:ApplyBotWant}]$.

\textbf{Case 1 - } $p'$'s call to \lock{p} returned $\bot$.
Then it follows from the code structure that $p'$ executed a call to \abort{p} and returned from line~\ref{abort:returnbot} or~\ref{abort:returnr}.
Since $p$ executed a successful \Apply$[p]$.\CAS{$\pair{\bot}{\bot},\pair{\cWant}{s'}$} in \Line{getLock:ApplyBotWant}, $p'$ could not have aborted while busy-waiting on \Line{getLock:ApplyBotWant}, and thus $p'$ aborted while busy-waiting in line~\ref{getLock:awaitAckOrCtrDecrease} or~\ref{getLock:awaitX}.
Then $p'$ executed \Line{getLock:setFlag}, and set its local variable $p'.flag$ to \True, and thus $p$ could not have returned $\bot$ from \Line{abort:returnbot} during \abort{p}.
Then $p'$ returned $\bot$ in \Line{abort:returnr}, and thus $p'$ executed operations \Apply$[p]$.\CAS{\pair{\cWant}{s'},\pair{\cOk}{s'}} (in \Line{abort:ApplyWantOk}), and \Apply$[p]$.\CAS{\pair{\cOk}{s'},\pair{\bot}{\bot}} (in \Line{abort:ApplyOkBot}).
Since, \Apply$[p]$ can be changed from \pair{\cWant}{s'} only to \pair{\cOk}{s'}, and from \pair{\cOk}{s'} only to \pair{\bot}{\bot}, it then follows that $p'$ executes a successful \Apply$[p]$.\CAS{\pair{\cOk}{s'},\pair{\bot}{\bot}} (in \Line{abort:ApplyOkBot}).
Then $p'$ eventually resets \Apply$[p]$ during its \lock{p} call.
Since $\Apply[p] \neq \pair{\bot}{\bot}$ throughout $[\pt{p'}{getLock:ApplyBotWant},\ptB{p}{getLock:ApplyBotWant}]$ and $p'$ completed its call to \lock{p} during $[\pt{p'}{getLock:ApplyBotWant},\ptB{p}{getLock:ApplyBotWant}]$, we have a contradiction.

\textbf{Case 2 - } $p'$'s call to \lock{p} returned a non-$\bot$ value.
Then from the code structure $p'$ executed operations \Apply$[p]$.\CAS{\pair{\cWant}{s'},\pair{\cOk}{s'}} (in \Line{getLock:ApplyWantOk} or \Line{abort:ApplyWantOk})  before returning from its call to \lock{p}.
Since \Apply$[p]$ can be changed from \pair{\cWant}{s'} only to \pair{\cOk}{s'}, and from \pair{\cOk}{s'} only to \pair{\bot}{\bot} and only by a process with pseudo-ID $p$, it then follows that $\Apply[p] = \pair{\cOk}{s'}$ when $p'$'s \lock{p} returns.
Then it also follows that $\Apply[p] = \pair{\cOk}{s'}$ until a process with pseudo-ID $p$ executes an \Apply$[p]$.\CAS{\pair{\cOk}{s'},\pair{\bot}{\bot}} operation.

Since $p'$ won the lock \L, it follows that some process, say $r$, eventually executes a call to \release{p}{j}, for some integer $j$.
Since a call to \release{p}{j} is wait-free and all processes continue to take steps, it follows that eventually $r$ executes lines~\ref{release:ApplyRead} and~\ref{release:ApplyOkBot} where it reads value \pair{\cOk}{s'} from \Apply$[p]$ in \Line{release:ApplyRead} and resets \Apply$[p]$ with a \Apply$[p]$.\CAS{\pair{\cOk}{s'},\pair{\bot}{\bot}} operation in \Line{release:ApplyOkBot}.
Since $p$ does not abort, and no other process calls \lock{p} concurrently, it then follows that eventually $p$ executes a successful \Apply$[p]$.\CAS{\pair{\bot}{\bot},\pair{\cWant}{s}} operation in \Line{getLock:ApplyBotWant}.
Since \Apply$[p]$ changed only once from \pair{\cOk}{s'} to \pair{\bot}{\bot} while $p$ busy-waited in \Line{getLock:ApplyBotWant}, it follows that $p$ incurs \Order{1} RMRs during the entire process.
Then our claims follow immediately from Claims~\ref{cl:conditional:getLock}\refC{scl:conditional:pDoesNotStarve} and~\ref{cl:conditional:getLock}\refC{scl:conditional:eRMRcomplexity}.
\end{proof}


\begin{lemma} \label{cl:ARMLock:exitway:waitfree}
The abort-way is wait- free.
\end{lemma}
\begin{proof}
The abort-way is defined to be all steps taken by a process (say $p$) after it receives a signal to abort and breaks out of one of the busy-wait cycles of lines~\ref{getLock:ApplyBotWant},~\ref{getLock:awaitAckOrCtrDecrease} or~\ref{getLock:awaitX}.
After $p$ breaks out of one of the busy-wait cycles of lines~\ref{getLock:ApplyBotWant},~\ref{getLock:awaitAckOrCtrDecrease} or~\ref{getLock:awaitX} $p$ executes a call to \abort{p}.
If $p$'s call to \abort{p} returns $\bot$, then $p$'s passage ends, or else $p$'s \lock{p} returns non-$\bot$ value and $p$ calls \release{p}{} and $p$'s passage ends when the \release{p}{} method returns.
Since \abort{p} and \release{p}{} are both wait-free (by Lemma~\ref{cl:MethodsWaitfree}), our claim follows.
\end{proof}

\begin{lemma} \label{cl:ARMLock:starvationfree}
The starvation freedom property holds during history $H$.
\end{lemma}
\begin{proof}
Consider a process $p$ that begins to execute its passage.
From Lemma~\ref{cl:getLock}\refC{scl:pDoesNotStarve}, it follows that if $p$ does not abort during \lock{p} and all processes continue to take steps then $p$ eventually returns from \lock{p} with a non-$\bot$ value.
Then $p$ eventually calls \release{p}{}, and since \release{p}{} is wait-free, $p$ eventually completes its passage.
If $p$ receives a signal to abort during \lock{p}, then $p$ executes its abort-way.
Since the abort-way is wait-free (by Lemma~\ref{cl:ARMLock:exitway:waitfree}), $p$ eventually completes its passage.
\end{proof}

\begin{lemma}
If a call to \release{p}{j} returns \True, then there exists a concurrent call to \lock{} that eventually returns $j$.
\label{lemma:ARMLockTree:transferable}
\end{lemma}
\begin{proof}

The only operations that write a value to \X\ are \X.\CAS{$\bot,\infty$} in \Line{abort:setX}, and \X.\CAS{$\bot,j$} in \Line{release:setX}.
From Claim~\ref{cl:variablesChangedByReleaserOnly}, \X\ is written to only by a releaser of \L.
From \Claim{cl:@I0-} it follows that at $\Ib{0}$, $\X = \LSync = \bot$ and \PawnSet\ is candidate-empty, and $R(\Ib{0}) = \varnothing$.
Then from Claims~\ref{cl:proofProperties1}\refC{I0},~\ref{cl:proofProperties1}\refC{R@I1},~\ref{cl:proofProperties3}\refC{I2-,L},~\ref{cl:proofProperties3}\refC{L,G},~\ref{cl:proofProperties5}\refC{[G,I2+]}, and ~\ref{cl:proofProperties5}\refC{@I2+}, the only releasers of \L\ during a \ctr-cycle interval $T$,
 are \Set{\K,\Q,\P_1,\ldots,\P_{\ell}}.
Then from an inspection of Figures~\ref{fig:KsGetLock},~\ref{fig:KsGetLock},~\ref{fig:QsGetLock},~\ref{fig:QsRelease},~\ref{fig:PisGetLock} and~\ref{fig:PisRelease}, it follows that only \K\ and \Q\ can write to \X\ during \ctr-cycle interval $T$.

Since $p$ returns \True, it then follows from an inspection of the code that $p$ executed a successful \X.\CAS{$\bot,j$} operation in \Line{release:setX}, and thus failed the \ctr.\CAS{$1,0$} operation in \Line{release:ctr10} and $\Role[p] = \cKing$ at \pt{p}{release:ctr10}.
Then $p = \K$ for some \ctr-cycle interval $T$.
Since \K\ failed the \ctr.\CAS{$1,0$} operation in \Line{release:ctr10}, it then follows that \ctr\ was increased to $1$ by process \Q\ during $T$, and $\Ib{2} = \pt{\Q}{getLock:IncCounter} < \pt{\K}{release:ctr10}$.
Since $\Ib{1} = \pt{\K}{getLock:IncCounter}$ and $\Ib{1} < \Ib{2}$, it then follows that \Q's \lock{\Q} call is concurrent to \K's \release{\K}{j} call.

From \Claim{cl:@I0-} it follows that at $\Ib{0}$, $\X = \LSync = \bot$ and \PawnSet\ is candidate-empty, and $R(\Ib{0}) = \varnothing$.
Then from Claim~\ref{cl:proofProperties3}\refC{@I2-}, $\X = \bot$ at \Ib{2}, and \K\ and \Q\ are the only two releasers of \L\ during $[\Ib{2},\lambda)$, where $\lambda$ is the first point in time when \T\ is changed to a non-$\bot$ value, and $\lambda = \minimum{\pt{\K}{hRelease:setT},\pt{\Q}{hRelease:setT}}$.

Now, \X\ is reset only in \Line{hRelease:resetX}, and since $\pt{\K}{hRelease:resetX} > \pt{\K}{hRelease:setT} \geq \lambda$ and $\pt{\Q}{hRelease:resetX} > \pt{\Q}{hRelease:setT} \geq \lambda$, it then follows that \K\ and \Q\ do not reset \X\ during $[\Ib{2},\lambda]$.
Since  \K\ and \Q\ are the only processes with write-access to \X, \X\ is not reset during $[\Ib{2},\lambda]$.

Consider \Q's \lock{} call (see Figure~\ref{fig:QsGetLock}).
Since \K\ executed a successful \X.\CAS{$\bot,j$} operation and \X\ is not reset during $[\Ib{2},\lambda]$, it then follows that if \Q\ executes the \X.\CAS{$\bot,\infty$} operation in \Line{abort:setX}, then the operation fails.
From an inspection of Figure~\ref{fig:QsGetLock}, \Q\ either returns from its \lock{} call in \Line{getLock:returnX} or \Line{abort:returnX}.
In both these lines, \Q\ returns the non-$\bot$ value stored in \X.
Since \K\ is the only process apart from \Q\ that can write to \X\, \Q\ returns the value $j$ that \K\ wrote during its \release{\K}{j} call.
\end{proof}


Now consider an implementation of object \ARMLockArray{N}, where instance \PawnSet\ is implemented using object \UC{\APArray{n}}, and the operations in lines~\ref{collect:updateAll}, ~\ref{hRelease:collectFixOther},~\ref{promote:FR12},~\ref{promote:collectFixSelf}, and~\ref{promote:resetBackpack} are executed using the \performFast{} method, while the operation in \Line{abort:ifHead} is executed using the \performSlow{}.

\begin{claim} \label{scl:opsNotConcurrent}
 Lines~\ref{promote:collectFixSelf},\ref{promote:FR12},~\ref{promote:resetBackpack} of \doPromote{}, all lines of \doCollect{}, and lines~\ref{hRelease:readX}-\ref{hRelease:callPromote} are not executed concurrently.
\end{claim}
\begin{proof}
From \Claim{cl:ifReleasinglockThenReleaser}\refC{scl:isReleaser}, it follows that only a releaser of \L\ can execute any of these lines.
From \Claim{cl:@I0-} it follows that at $\Ib{0}$, $\X = \LSync = \bot$ and \PawnSet\ is candidate-empty, and $R(\Ib{0}) = \varnothing$.
Then from Claims~\ref{cl:proofProperties1}\refC{I0},~\ref{cl:proofProperties1}\refC{R@I1},~\ref{cl:proofProperties3}\refC{I2-,L},~\ref{cl:proofProperties3}\refC{L,G},~\ref{cl:proofProperties5}\refC{[G,I2+]}, and ~\ref{cl:proofProperties5}\refC{@I2+} it follows that \L\ has more than one releaser only during $[\Ib{2},\lambda)$ for some \ctr-cycle interval $T$.
More specifically, there are two releasers of \L\ only during $[\Ib{2},\lambda)$, and the releasers are \K\ and \Q.
From \Claim{cl:proofProperties3}\refC{aOrbCollect} it follows that a \doCollect{} is executed only by \K\ or \Q\ but not both.
Then it follows immediately that lines of \doCollect{} are not executed concurrently.
Since $\lambda = \minimum{\pt{\K}{hRelease:setT}, \pt{\Q}{hRelease:setT}}$, it follows from an inspection of Figures~\ref{fig:KsGetLock},~\ref{fig:KsRelease},~\ref{fig:QsGetLock},~\ref{fig:QsRelease} and the code, that processes \K\ and \Q\ have not executed a call to \doPromote{} or lines~\ref{hRelease:readX}-\ref{hRelease:callPromote} of \helpRelease{}, before \pt{\K}{hRelease:setT} and \pt{\Q}{hRelease:setT} respectively.
Then none of the lines chosen in the claim are executed concurrently, and thus our claim holds.
\end{proof}

\begin{lemma} \label{lemma:ARMLockTree:complexity}
\begin{enumerate}[(a)]
 \item Both \helpRelease{p} and \doPromote{p} have \Order{1} RMR complexity. \label{scl:hReleaseAndPromote:complexity}
 \item \doCollect{p} has \Order{n} RMR complexity. \label{scl:doCollect:complexity}
 \item \abort{p} has \Order{n} RMR complexity. \label{scl:abort:complexity}
 \item If a call to \release{p}{j} returns \True, then $p$ incurs \Order{n} RMRs during \release{p}{j}.\label{scl:releaseFast:complexity}
 \item If a call to \release{p}{j} returns \False, then $p$ incurs \Order{1} RMRs during \release{p}{j}.\label{scl:releaseSlow:complexity}
\end{enumerate}
\end{lemma}
\begin{proof}
\textbf{Proof of~\refC{scl:hReleaseAndPromote:complexity} and~\refC{scl:doCollect:complexity}: }
As per the properties of object \UC{\APArray{n}} (Lemma~\ref{theorem:UC}), an operation performed using the \performFast{} method has \Order{1} RMR complexity, as long as it is not executed concurrently with another \performFast{} method call.
Since \PawnSet\ is an instance of object \UC{\APArray{n}}, where operations in lines~\ref{collect:updateAll}, ~\ref{hRelease:collectFixOther},~\ref{promote:FR12},~\ref{promote:collectFixSelf}, and~\ref{promote:resetBackpack} are executed using the \performFast{} method, and each of these operations are not executed concurrently (by Claim~\refC{scl:opsNotConcurrent}), it then follows that all of these operations have \Order{1} RMR complexity.
Then Part~\refC{scl:hReleaseAndPromote:complexity} follows immediately from an inspection of methods \helpRelease{} and \doPromote{}.
Since method \doCollect{} has a loop of size $n$ that incurs a constant number of RMRs in each iteration, Part~\refC{scl:doCollect:complexity} follows.

\textbf{Proof of~\refC{scl:abort:complexity},~\refC{scl:releaseFast:complexity} and~\refC{scl:releaseSlow:complexity}: }
As per the properties of object \UC{\APArray{n}} (Lemma~\ref{theorem:UC}), an operation performed using the \performSlow{} method has \Order{n} RMR complexity, where $n$ is the maximum number of processes that can access the object concurrently.
Since the operation in \Line{abort:ifHead} is executed using the \performSlow{} method, the operation has \Order{n} RMR complexity.
Since \helpRelease{} and \doPromote{} have an RMR complexity of \Order{1} (by Part~\refC{scl:hReleaseAndPromote:complexity}), and \doCollect{} has an RMR complexity of \Order{n}  (by Part~\refC{scl:doCollect:complexity}), it then follows from an inspection of \abort{}, that a call to \abort{} has an RMR complexity of \Order{n}.
Thus Part~\refC{scl:doCollect:complexity} follows.

If a call to \release{p}{j} returns \True, then $p$ does execute a call to \doCollect{p} in \Line{release:doCollect}, else it does not.
Then from an inspection of \release{p}{j}, Parts~\refC{scl:releaseFast:complexity} and~\refC{scl:releaseSlow:complexity} follow immediately.
\end{proof}

Lemma~\ref{theorem:ARMLockArray} follows from Lemmas~\ref{cl:MethodsWaitfree}, \ref{cor:ARMLock:mutualExclusion}, \ref{cl:getLock}, \ref{cl:ARMLock:exitway:waitfree}, \ref{cl:ARMLock:starvationfree}, \ref{lemma:ARMLockTree:transferable}, and \ref{lemma:ARMLockTree:complexity}.



\section{The Tree Based Randomized Abortable Lock}
\label{sec:appendix:ARLockTree}

\subsection{Implementation / Low Level Description}\label{sec:ARLockTree:Implementation}

We assume that the tree structure \tree\ provides a function \nodeOnPath{}, such that, for a leaf node \myleaf{} and integer $\ell$, the function \nodeOnPath{$\myleaf{},\ell$} returns a pair \pair{u}{i}, where $u$ is the $\ell$-th node on the path from \myleaf{} to the root node, and $i$ is the index of the child node of $u$ that lies on the path.

We now describe the implementation of the abortable lock (see Figure~\ref{fig:ARMEAlgorithm}).

\begin{classfigure}[!htbp]
\setcounter{AlgoLine}{0}
\begin{algo}{Implementation of the abortable lock} \label{Algo:AbortRandomMutEx}
  \Define \Node: struct \{                              \label{Node_Struct_Begin}
     $\L$: \ARMLockArray{\Delta}
    \}                                          \label{Node_Struct_End}\;
  \shared
    \qquad \tree: complete $\Delta$-ary tree of height $\Delta$ and node type \Node\ \;
  \local                                      \label{Local_Types_Begin}
    \qquad $v$: \Node\ \Init $\bot$; 
    \qquad $i,\ell,k$: \Int\ \Init $0$;
    \qquad $abort\_signal$: \Bool\ \Init $\False$;\;\;
    \Define function \tree.\nodeOnPath{\Node\ $\myleaf{},\Int\ \ell$}: returns a pair \pair{u}{i}, where $u$ is the $\ell$-th node on the path from \myleaf{} to the root node of \tree, and $i$ is the index of the child node of $u$ that lies on the path.\;
\end{algo}
 \begin{minipage}{0.57\textwidth}
 \LinesNumbered
 \begin{algorithm}[H]\Method{lock$_p$()}
 %
  \label{Func:lock}
     \While{$\ell < \tree.height $}{     \label{lock:BeginLoop}
        \tuple{v,i} \la\ \tree.\nodeOnPath{$\leaf_p,\ell+1$}\; \label{lock:getNode}
        $val$ \la\ $v$.\L.\lock{i}\;                      \label{lock:callGetLock}
        \lIf {$val = \infty$}  {                           \label{lock:ifCaptured}
               $\ell \la\ \ell +1$\;                   \label{lock:incrementL}
        }
        \lIf {$val \notin \Set{\bot,\infty}$}  {           \label{lock:ifTransfered}
          $\ell$ \la\ $val$ \;          \label{lock:receiveTransfer}
        }
        \If{$abort\_signal$ = \True} { \label{lock:ifAbortSignal}
              \release{p}{} \;                            \label{lock:callRelease}
              \return $\bot$ \;                                 \label{lock:returnbot}
        }
    }                                                           \label{lock:EndLoop}
   \return $\infty$ \;                                       \label{lock:returninfty}
 \end{algorithm}
 \end{minipage}
 \begin{minipage}{0.42\textwidth}
 \LinesNumbered
 \begin{algorithm}[H]\Method{release$_p$()}
 \label{Func:release}
     \While {$k \leq \ell$} {                   \label{release:BeginLoop}
         \tuple{v,i} \la\ \tree.\nodeOnPath{$\leaf_{p},k$} \; \label{release:getNode}
         \lIf{$v$.\L.\release{i}{$\ell$}} {                   \label{release:callRelease}
           \breakLoop \;                                        \label{release:ifTransfered}
         }
         $k$ \la\ $k+1$ \;                                 \label{release:incrementK}
       }                                       \label{release:EndLoop}
 \end{algorithm}
 \end{minipage}
\caption{Implementation of the abortable lock}\label{fig:ARMEAlgorithm}
\end{classfigure}

\bparagraph{Description of the \lock{p} method}
Suppose process $p$ executes a call to \lock{p}.
With every iteration of the while-loop, process $p$ captures at least one node on its path from \myleaf{p} to \tree.\root.
Suppose $p$ executes an iteration of while-loop (lines~\jref{lock:BeginLoop}-\jref{lock:EndLoop}) and $\ell_p = k$ at \Line{lock:BeginLoop} for some arbitrary integer $k$.
In \Line{lock:getNode}, process $p$ determines the $k$-th node (say $u$) on \mypath{p}  and the index (say $r$) of $u$'s child node that lies on \mypath{p}, and stores them in local variables $v_p$ and $i_p$.
The variables $v_p$ and $i_p$ are unchanged during the rest of the iteration.
In \Line{lock:callGetLock}, process $p$ attempts to capture $u.\L$, and thus node $u$ by executing a call to $u$.\L.\lock{} with pseudo-ID $r$.
If $p$'s $u$.\L.\lock{r} returns an integer value (say $j$) then $p$ has been transferred all nodes on its path up to height $j$
(we ensure $j \geq \h{u}$).
If $p$'s $u$.\L.\lock{} returns $\infty$ then $p$ has captured lock $u$.\L.
In lines~\ref{lock:incrementL} and~\ref{lock:receiveTransfer}, $p$ stores the height of the highest captured node in its local variable $\ell_p$.
In line~\ref{lock:ifAbortSignal}, $p$ checks whether it has received a signal to abort.
In this case $p$ releases all its captured nodes by executing a call to \release{p}{} in line~\ref{lock:callRelease} and then returns from its call to \lock{p} in line~\ref{lock:returnbot} with value $\bot$.
Otherwise $p$ continues its while-loop.
On completing its while-loop, $p$ owns the root node, and thus returns with value $\infty$ in \Line{lock:returninfty} to indicate a successful \lock{} call.

\bparagraph{Description of the \release{p}{} method}
Suppose process $p$ executes a call to \release{p}{}.
Let $s$ be the highest node $p$ owns at the beginning of \release{p}{}.
We later prove that $\h{s} = \ell_p$.
During an iteration of the while-loop (lines~\jref{release:BeginLoop}-\jref{release:EndLoop}), process $p$ either releases a node on its path from \myleaf{p} to $s$, or $p$ hands over all remaining nodes that it owns to some process.

Consider the execution of an iteration of the while-loop where $k_p = t$ at \Line{release:BeginLoop} for some integer $t \leq \h{s}$.
In \Line{release:getNode}, process $p$ determines the $t$-th node (say $u$) on \mypath{p} and the index (say $r$) of $u$'s child node that lies on \mypath{p}, and stores them in local variables $v_p$ and $i_p$.
In \Line{release:callRelease}, process $p$ releases $u.\L$, and thus node $u$, by executing a call to $u$.\L.\release{}{\h{s}} with pseudo-ID $r$.
If $p$'s $u$.\L.\release{r}{\h{s}} returns \False\ then $p$ has successfully released lock $u$.\L, and thus node $u$.
If $p$'s $u$.\L.\release{r}{\h{s}} returns \True\ then $p$ has successfully handed over all nodes from $u$ to $s$ on \mypath{p} to some process that is executing a concurrent call to $u$.\L.\lock{}.
If $p$ has handed over all its nodes, then $p$ breaks out of the while-loop in \Line{release:ifTransfered}, and returns from its call to \release{p}{}.
If $p$ has not handed over all its nodes then $p$ increases $k_p$ in \Line{release:incrementK} and continues its while-loop.

Notice that our strategy to release node locks is to climb up the tree until all node locks are released or a hand over of remaining locks is made.
Climbing up the tree is necessary (as opposed to climbing down) in order to hand over node locks to a process, say $q$, such that the handed over nodes lie on $\mypath{q}$.
There is however a side effect of this strategy which is as follows:
Suppose $p$ owns nodes $v$ and $u$ on \mypath{p} such that $\pair{u}{i} = \nodeOnPath{\myleaf{p},\h{u}}$ and $v$ is the $i$-th child on node $u$.
Now suppose $p$ releases lock $v.\L$ at node $v$.
Since the lock at node $v$ is now released, some process $r \neq p$ may now capture lock $v.\L$ and then proceed to call $u.\L.\lock{i}$.
If process $p$ has not yet released $u$.\L\ by completing its call to $u$.\L.\release{i}{}, then we have a situation where a call to $u.\L.\lock{i}$ is made before a call to $u$.\L.\release{i}{} is completed.
Since there can be at most one owner of lock $v$.\L\ there can be at most one such call to $u.\L.\lock{i}$ concurrent to $u$.\L.\release{i}{}.
This is precisely the reason why we designed object \ARMLockArray{n} to be accessed by at most $n+1$ processes concurrently.

\subsection{Analysis and Proofs of Correctness}\label{sec:ARMLockStateless:Analysis}
In this section, we formally prove all properties of our abortable lock for the CC model.
We first, establish the safety conditions on the usage of the object.
\begin{condition}
\label{cond:safety:ARMLockTree}
\begin{enumerate}[(a)]
  \item If process $p$ executes a successful \lock{p} call, then process $p$ eventually executes a \release{p}{} call.
  \item A process calls method \release{}{} if and only if its last access of the lock object was a successful \lock{} call.
  \item Methods \lock{p} and \release{p}{} are called only by process $p$, where $p \in \Set{0,\ldots,N-1}$.
  \item For every \release{p}{} call, there must exist a unique successful \lock{p} call that has been executed. \label{condition:ifReleaseThenExistsLock}
\end{enumerate}
\end{condition}

\bparagraph{Notations and Definitions}
Let $H$ be an arbitrary history of an algorithm that accesses an instance $\L$ of our abortable lock where Condition~\ref{cond:safety:ARMLockTree} is satisfied.
Consider an arbitrary node $u$ on the tree \tree.
Let $\h{u}$ denote the height of node $u$.

A node $u$ is said to be \emph{handed over} from process $p$ to process $q$, when $p$ executes a $v$.\L.\release{}{j} call that returns \True, where $j \geq \h{u} > \h{v}$ and $q$ executes a concurrent $v$.\L.\lock{} call that returns $j$.
Process $p$ is said to start to \textit{own} node $u$ when $p$ captures $u$.\L\ or when it is handed over node $u$ from the previous owner of node $u$.
Process $p$ \emph{ceases} to own node $u$ when $p$ releases $u$.\L, or when $p$ hands over node $u$ to some other process.


\begin{claim}
\label{claim:helpful:basic}
Consider an arbitrary process $p$ and some node $u$ on \mypath{p}.
\begin{enumerate}[(a)]
  \item If $p$ executes a $u$.\L.\lock{} operation that returns value $j \notin \Set{\bot,\infty}$, then $j \geq \h{u}$. \label{scl:TransferredLockOfGreaterNodesOnly}
  \item The value of $\ell_p$ is increased every time $p$ writes to it. \label{scl:LpIncreaseEveryTimeWrittenTo}
  \item If $\ell_p = k$, then process $p$ owns all nodes on \mypath{p} up to height $k$. \label{scl:OwnerOfAllNodesUptoLp}
\end{enumerate}

\end{claim}
\begin{proof}
%

\textbf{Proof of~\refC{scl:TransferredLockOfGreaterNodesOnly} : }
Then from the properties of object \ARMLockArray{\Delta} (Lemma~\ref{theorem:ARMLockArray}), it follows that some process (say $q$) executed a concurrent $u$.\L.\release{}{j} operation.
Then from the code structure, $q$ executed a $u$.\L.\release{}{j} in \Line{release:callRelease}, where $\ell_q = j$.
Then $q$ also executed a \tree.\nodeOnPath{$\myleaf{q},k$} operation in \Line{release:getNode} that returned \pair{u}{i}, for some $i$, such that $\h{u} = k_q$ (from the semantics of the \nodeOnPath{} method).
Since $j = \ell_q \geq k_q = \h{u}$, our claim follows.

\textbf{Proof of~\refC{scl:LpIncreaseEveryTimeWrittenTo}: }
Process $p$ writes to its local variable $\ell_p$ only in lines~\ref{lock:incrementL} and~\ref{lock:receiveTransfer}.
Clearly, $p$ increases $\ell_p$ every time it executes \Line{lock:incrementL}.
Now, suppose $p$ executes line~\ref{lock:receiveTransfer} where it writes the value of $val_p$ to $\ell_p$, where $v_p = u$, for some node $u$.
Since $p$ satisfies the if-condition of \Line{lock:receiveTransfer} and the \ARMLockArray{\Delta}\   method \lock{} only returns a value in $\Set{\bot,\infty} \cup \N$, it follows that $p$'s call to $u$.\L.\lock{} returned a non-$\Set{\bot,\infty}$ value.
Then from Part~\refC{scl:TransferredLockOfGreaterNodesOnly}, $val_p \geq \h{u}$.
Since $p$ also executed a \tree.\nodeOnPath{$\myleaf{p},b$} operation in \Line{lock:getNode}, where $b = \ell_p+1$ that returned \pair{u}{i}, for some $i$, such that $\h{u} = b$ (from the semantics of the \nodeOnPath{} method), it follows that $val_p \geq \h{u} = \ell_p + 1$.
Then, $p$ increases $\ell_p$ when $p$ writes $val_p$ to $\ell_p$ in \Line{lock:receiveTransfer}.

\textbf{Proof of~\refC{scl:OwnerOfAllNodesUptoLp}: }
Let $t^i$ be the point in time such that $p$ writes to its local variable $\ell_p$ for the $i$-th time.
We prove our claim by induction over $i$

\textbf{Basis ($i = 0$): }
Since the initial value of $\ell_p$ is $0$ and $\ell_p$ is written to for the first time only at $t^1 > t^0$, the claim holds.

\textbf{Induction step ($i > 0$): }
Let the value of $\ell_p$ be $j$ after the $(i-1)$-th write to it.
Then from the induction hypothesis, $p$ owns all nodes on \mypath{p} up to height $j$.
Consider the iteration of the while-loop during which $p$ writes to $\ell_p$ for the $i$-th time, and specifically the  \tree.\nodeOnPath{$\myleaf{p},\ell+1$} operation in \Line{lock:getNode}.
Since $\ell_p = j$, at the beginning of this while-loop iteration, it follows from the semantics of the \nodeOnPath{} operation, that the operation returned the pair \pair{u}{i}, for some $i$, where $\h{u} = j+1$.
Now, process $p$ writes to its local variable $\ell_p$ only in lines~\ref{lock:incrementL} and~\ref{lock:receiveTransfer}.

\textbf{Case a - } $p$ writes to $\ell_p$ in line~\ref{lock:incrementL}.
Then $p$ increased $\ell_p$ from $j$ to $j+1$ in \Line{lock:incrementL}.
Then, to prove our claim we need to show that $p$ owns the node with height $j+1$ on \mypath{p}.
Since $p$ satisfies the if-condition of \Line{lock:incrementL}, it follows from the code structure that $p$'s $u$.\L.\lock{} method in \Line{lock:callGetLock} returned the special value $\infty$, where $v_p = u$.
Since $\h{u} = j+1$, and $p$ successfully captured lock $u$.\L, it follows that $p$ owns the $j+1$-th node on \mypath{p}.

\textbf{Case b - } $p$ writes to $\ell_p$ in line~\ref{lock:receiveTransfer}.
Let $val_p = x$ when $p$ writes to $\ell_p$ in line~\ref{lock:receiveTransfer}.
From Part~\refC{scl:LpIncreaseEveryTimeWrittenTo}, it follows that $\ell_p$ is increased every time it is written to, and therefore $val_p = x > \ell_p$ when $p$ writes to $\ell_p$ in line~\ref{lock:receiveTransfer}.
Thus, to prove our claim we need to show that $p$ owns all nodes on \mypath{p} with heights in the range $\Set{j,\ldots,x}$.
Since $p$ satisfies the if-condition of \Line{lock:receiveTransfer} and the \ARMLockArray{\Delta}\   method \lock{} only returns a value in $\Set{\bot,\infty} \cup \N$, it follows that $p$'s call to $u$.\L.\lock{} returned a non-$\Set{\bot,\infty}$ value.
Thus, $p$ has captured $u$.\L\ and now owns node $u$.
It also follows that $p$ has been handed over all nodes on \mypath{p} with heights in the range $\Set{\h{u}+1,\ldots,x}$.
Since $\h{u} = j$, our claim follows.
\end{proof}

A process is said to \emph{attempt to capture node} $u$ if it executes a $u.\L$.\lock{} method in line~\ref{lock:callGetLock}.

\begin{claim}\label{clm:registered_at_same_node}
\begin{enumerate}[(a)]
 \item If two distinct processes $p$ and $q$ attempt to capture node $v$, then their local variables $i$ have different values. \label{scl:safetyForARMLockArray}
 \item A node has at most one owner at any point in time. \label{scl:NodeMutex}
\end{enumerate}
\end{claim}
\begin{proof}
 We prove our claims for all nodes of height at most $h$, by induction over integer $h$.

\textbf{Basis ($h = 1$)}
Consider an arbitrary node $u$ of height $1$, such that two distinct processes $p$ and $q$ attempt to capture node $u$.
Then processes $p$ and $q$ executed a \nodeOnPath{\pair{\myleaf{p}}{1}} and \nodeOnPath{\pair{\myleaf{q}}{1}} in \Line{lock:getNode}, and received pairs \pair{u}{i} and \pair{u}{j}, and set their local variables $i_p$ and $i_q$ to $i$ and $j$ respectively.
Since $p$ and $q$ are distinct, \myleaf{p} and \myleaf{q} are distinct leaf nodes of tree \T, and thus from the semantics of the \nodeOnPath{} method it follows that $i \neq j$, and thus Part~\refC{scl:safetyForARMLockArray} follows.


Consider an arbitrary node $u$ of height $1$.
From Part~\refC{scl:safetyForARMLockArray}, it follows that no two processes execute a concurrent call to $u$.\L.\lock{i} for the same $i$, and thus it follows from the mutual exclusion property of object \ARMLockArray{\Delta}, that at most one process captures $u$.\L.
By definition, a process can become an owner of node $u$ only if it captures $u$.\L\ or if it is handed over node $u$ from some other process $q$.
If a node $u$ is handed over from some other process $q$, then $q$ also ceases to be the owner of node $u$ at that point, and thus the number of owners of $u$ does not increase upon a hand over.
Thus it follows that node $u$ has at most one owner at any point in time, and thus Part~\refC{scl:NodeMutex} follows.

\textbf{Induction Step ($h > 1$)}
Consider an arbitrary node $u$ of height $h$, such that two distinct processes $p$ and $q$ attempt to capture node $u$.
Then processes $p$ and $q$ executed a \nodeOnPath{\pair{\myleaf{p}}{h}} and \nodeOnPath{\pair{\myleaf{q}}{h}} in \Line{lock:getNode}, and received pairs \pair{u}{i} and \pair{u}{j}, and set their local variables $i_p$ and $i_q$ to $i$ and $j$, respectively.
For the purpose of a contradiction, assume $i=j$.
From the semantics of \nodeOnPath{} method, $i = j$ only if the $(h-1)$-th nodes on \mypath{p} and \mypath{q} are the same (say $w$).
From the induction hypothesis of Part~\refC{scl:NodeMutex} for $h-1$, $w$ has at most one owner at any point in time.
Since $\ell_p = \ell_q = h -1$ when $p$ and $q$ attempt to capture node $u$, it follows from Claim~\ref{claim:helpful:basic}\refC{scl:OwnerOfAllNodesUptoLp}, that $p$ and $q$ own all nodes up to height $h-1$ on their individual paths \mypath{p} and \mypath{q}.
Then $p$ and $q$ are both the owners of $w$ -- a contradiction.
Thus, Part~\refC{scl:safetyForARMLockArray} follows.

Since Part~\refC{scl:safetyForARMLockArray} holds for $h$, Part~\refC{scl:NodeMutex} holds for $h$, as argued in the \textbf{Basis} case.
\end{proof}

\begin{lemma}
The mutual exclusion property is satisfied during history $H$.
\label{cl:ARMLockTree:mutualexclusion}
\end{lemma}
\begin{proof}
Assume two processes $p$ and $q$ are in their Critical Section at the same time, i.e., both processes returned a non-$\bot$ value from their last \lock{} call.
Then both processes executed \Line{lock:returninfty} and thus $\ell_p = \ell_q = \tree.height$ holds.
Then from Claim~\ref{claim:helpful:basic}\refC{scl:OwnerOfAllNodesUptoLp} it follows that both $p$ and $q$ own node \tree.\root.
But from Claim~\ref{clm:registered_at_same_node}\refC{scl:NodeMutex}, at most one process may own \tree.\root\ at any point in time -- a contradiction.
\end{proof}

\begin{claim}
Process $p$ repeats the while-loop in \lock{} at most $\Delta$ times. \label{cl:lockLoopRepeatedAtMostDeltaTimes}
\end{claim}
\begin{proof}
Consider an arbitrary process $p$ that calls \lock{}.
From the code structure of \lock{}, it follows that if $p$ repeats an iteration of the while-loop then $p$ either executed \Line{lock:incrementL} or \Line{lock:receiveTransfer} in its previous iteration.
Then it follows from \Claim{claim:helpful:basic}\refC{scl:LpIncreaseEveryTimeWrittenTo} that $p$ increases $\ell_p$ every time it repeats an iteration of the while-loop.
Since the height of the \tree\ is $\Delta$, our claim follows.
\end{proof}

\begin{lemma}
No process starves in history $H$.
\label{cl:ARMLockTree:starvationfree}
\end{lemma}
\begin{proof}
Since no two processes execute a concurrent call to $u$.\L.\lock{i} for the same $i$ (from Claim~\ref{clm:registered_at_same_node}~\refC{scl:safetyForARMLockArray}), it follows from the starvation-freedom property of object \ARMLockArray{\Delta}, that a process does not starve during a call to $u$.\L.\lock{} for some node $u$ on its path.

Consider an arbitrary process $p$ that calls \lock{}.
Since $p$ repeats the while-loop in \lock{} at most $\Delta$ times before returning from \Line{lock:returninfty} (follows from \Claim{cl:lockLoopRepeatedAtMostDeltaTimes}), it follows that $p$ starves only if $p$ starves during a call to $u$.\L.\lock{} in \Line{lock:callGetLock} for some node $u$.
As already argued, this cannot happen, and thus our claim follows.
\end{proof}

\begin{lemma}
Process $p$ incurs \Order{\Delta} RMRs during \release{p}{}. \label{cl:ARMLockTree:release:complexity}
\end{lemma}
\begin{proof}
Consider $p$'s call to \release{}{}.
Since $\ell_p \leq \tree.height = \Delta$, it follows from an inspection of the code that during \release{}{}, $p$ executes at most $\Delta$ calls to \L.\release{}{} (in \Line{release:callRelease}), and at most one of the \L.\release{}{} calls returns \True.
As per the properties of object \ARMLockArray{\Delta} (Lemma~\ref{theorem:ARMLockArray}), a process incurs \Order{\Delta} RMRs during a call to \L.\release{}{}, if the call returns \True, otherwise \Order{1} RMRs.
Then our claim follows immediately.
\end{proof}

\begin{lemma}
Process $p$ incurs \Order{\Delta} RMRs in expectation during \lock{p}. \label{cl:ARMLockTree:lock:complexity}
\end{lemma}
\begin{proof}
A process may or may not receive a signal to abort during \lock{p}.

\textbf{Case a - } $p$ does not receive a signal to abort during \lock{p}.
As per the properties of object \ARMLockArray{\Delta} (Lemma~\ref{theorem:ARMLockArray}), if a process does not receive a signal to abort during a call to \L.\lock{}, then the process incurs \Order{1} RMRs in expectation during the call.
Since $p$ repeats the while-loop in \lock{} at most $\Delta$ times (by \Claim{cl:lockLoopRepeatedAtMostDeltaTimes}), and $p$ does not receive a signal to abort during \lock{p}, it follows that $p$ incurs \Order{\Delta} RMRs in expectation during \lock{p}.

\textbf{Case b - } $p$ receives a signal to abort during \lock{p}.
As per the properties of object \ARMLockArray{\Delta} (Theorem~\ref{theorem:ARMLockArray}), if a process aborts during a call to \L.\lock{}, then the process incurs \Order{\Delta} RMRs in expectation during the call.
Since $p$ repeats the while-loop in \lock{} at most $\Delta$ times (by \Claim{cl:lockLoopRepeatedAtMostDeltaTimes}), and $p$ executes at most one call to $u$.\L.\lock{} after having received an abort signal, it follows that $p$ incurs \Order{\Delta} RMRs in expectation during \lock{p}.
\end{proof}

\begin{lemma}
Method \release{}{} is wait-free. \label{cl:ARMLockTree:release:waitfree}
\end{lemma}
\begin{proof}
As per the bounded exit property of object \ARMLockArray{\Delta}, method \release{}{} of the object is wait-free.
Then our claim follows immediately from an inspection of the code of \release{}{}.
\end{proof}

\begin{lemma}
The abort-way is wait-free and has \Order{\Delta} RMR complexity.
\label{cl:ARMLockTree:exitway}
\end{lemma}
\begin{proof}
The abort-way of a process $p$ consists of the steps executed by the process after receiving a signal to abort and before completing its passage.
From Lemma~\ref{cl:ARMLockTree:release:waitfree} and~\ref{cl:ARMLockTree:release:complexity}, method \release{p}{} is wait-free, and has \Order{\Delta} RMR complexity.
From \Claim{cl:lockLoopRepeatedAtMostDeltaTimes}, a process repeats the while-loop in \lock{p} at most $\Delta$ times.
Then from an inspection of the code it follows that a process executes all steps during its passage in a wait-free manner, except the call to $u$.\L.\lock{} in \Line{lock:callGetLock}, and that a process incurs at most \Order{\Delta} RMRs during all these steps.

To complete our proof we now show that if a process has received a signal to abort and it executes a call to $u$.\L.\lock{} in \Line{lock:callGetLock}, for some node $u$, then the process executes $u$.\L.\lock{} in a wait-free manner and incurs \Order{\Delta} RMR during the call, and does not call $v$.\L.\lock{} for any other node $v$.

Suppose that $p$ has received a signal to abort, and $p$ executes a call to $u$.\L.\lock{} call in \Line{lock:callGetLock}.
Since $p$ has received a signal to abort, it follows that $p$ executes the abort-way of the node lock $u$.\L.
As per the properties of object \ARMLockArray{\Delta} (Lemma~\ref{theorem:ARMLockArray}), its abort-way is wait-free and has \Order{\Delta} RMR complexity.
Then $p$ executes the $u$.\L.\lock{} call in \Line{lock:callGetLock} in a wait-free manner and incurs \Order{\Delta} RMR complexity.
It then goes on to satisfy the if-condition of \Line{lock:ifAbortSignal}, and executes a call to \release{}{} in \Line{lock:callRelease} and returns $\bot$ in \Line{lock:returnbot}, thereby completing its abort-way.
Thus, our claim holds.
\end{proof}


Theorem \ref{thm:main} follows from Lemmas~\ref{cl:ARMLockTree:mutualexclusion}, \ref{cl:ARMLockTree:starvationfree}, \ref{cl:ARMLockTree:release:complexity}, \ref{cl:ARMLockTree:lock:complexity}, \ref{cl:ARMLockTree:release:waitfree} and \ref{cl:ARMLockTree:exitway}.

\section{Remaining Proofs of Properties of \ARMLockArray{n}}
\label{sec:Appendix:remainingProofs}
\begin{claim}
Suppose a process $p$ executes a call to \lock{p} during a passage.
The value of $\Role[p]$ at various times is as follows.\

\begin{tabular}{ l | l}
    \hline
    Points in time & Value of $\Role[p]$  \\ \hline
    $\pt{p}{getLock:IncCounter}$ & \Set{\infty,\cKing,\cQueen,\cPawn} \\ \hline
    $[\pt{p}{getLock:awaitAckOrCtrDecrease},\pt{p}{getLock:ifBackpacked}]$ &  \cPawn \\ \hline
    $\pt{p}{getLock:RolePPawn}$ &  \cPPawn\ \\ \hline
    $\ptB{p}{getLock:ifQueen}$ & \Set{\cKing,\cQueen,\cPPawn}  \\ \hline
    $\pt{p}{getLock:awaitX}$ & \cQueen \\ \hline
    $[\pt{p}{getLock:ApplyWantOk},\pt{p}{getLock:ifRolePQueen}]$ & \Set{\cKing,\cQueen,\cPPawn} \\ \hline
\end{tabular}
\label{cl:TableOfRoles:getLock}
\end{claim}

\begin{proof}
Since the values returned by a \ctr.\inc{} operation are in \Set{\infty,0,1,2} =  \Set{\infty,\cKing,\cQueen,\cPawn}, $\Role[p]$ is set to one of these values in \Line{getLock:IncCounter}.
Hence, $\Role[p] \in \Set{\infty,\cKing,\cQueen,\cPawn}$ at \pt{p}{getLock:IncCounter}.
If $p$ satisfies the if-condition of \Line{getLock:ifSoldier}, then $\Role[p]$ = \cPawn, and $p$ changes $\Role[p]$ next only in \Line{getLock:RolePPawn}.
Hence, $\Role[p] = \cPawn$ during $[\pt{p}{getLock:awaitAckOrCtrDecrease},\pt{p}{getLock:ifBackpacked}]$.
In \Line{getLock:RolePPawn} $p$ changes $\Role[p]$ to \cPPawn\ and does not change $\Role[p]$ thereafter.
Hence, $\Role[p] = \cPPawn$ at $\pt{p}{getLock:RolePPawn}$.

Process $p$ does not change $\Role[p]$ after \Line{getLock:RolePPawn}.
To break out of the getLock loop, $\Role[p] \in \Set{\cKing,\cQueen,\cPPawn}$ must be satisfied when $p$ executes \Line{getLock:EndInnerLoop}.
Hence, $\Role[p] = \Set{\cKing,\cQueen,\cPPawn}$ during $[\pt{p}{getLock:ApplyWantOk},\pt{p}{getLock:ifRolePQueen}]$.
Since $p$ executes \Line{getLock:ifQueen} only after breaking out of the getLock loop, $\Role[p] \in \Set{\cKing,\cQueen,\cPPawn}$ at \ptB{p}{getLock:ifQueen}.
If $p$ satisfies the if-condition of \Line{getLock:ifQueen}, then $\Role[p]$ = \cQueen, and since $p$ does not change $\Role[p]$ thereafter, $\Role[p] = \cQueen$ at \pt{p}{getLock:awaitX}.
\end{proof}

\begin{claim}
Suppose a process $p$ executes a call to \abort{p}.
The value of $\Role[p]$ at various points in time is as follows.\

\begin{tabular}{ l | l}
    \hline
    Points in time & Value of $\Role[p]$  \\ \hline
    $[\pt{p}{abort:ApplyWantOk},\ptB{p}{abort:ifPawn}]$ & \Set{\cQueen,\cPawn}   \\ \hline
    $\pt{p}{abort:ifHead}$ & \cPawn  \\ \hline
    $[\pt{p}{abort:RolePPawn},\pt{p}{abort:returninfty}]$ & \cPPawn\ \\ \hline
    $[\ptB{p}{abort:setX},\pt{p}{abort:callhRelease}]$ &  \cQueen \\ \hline
\end{tabular}
\label{cl:TableOfRoles:abort}
\end{claim}
\begin{proof}
Process $p$ calls \abort{p} only if $p$ has received a signal to abort and $p$ is busy waiting in one of lines~\ref{getLock:ApplyBotWant},~\ref{getLock:awaitAckOrCtrDecrease},  or~\ref{getLock:awaitX}.
Then, the last line executed by $p$ before calling \abort{p} is line~\ref{getLock:ApplyBotWant},~\ref{getLock:awaitAckOrCtrDecrease}, or line~\ref{getLock:awaitX}.
From \Claim{cl:TableOfRoles:getLock}, it follows that $\Role[p]$ = \cPawn\ at \pt{p}{getLock:awaitAckOrCtrDecrease}, and $\Role[p]$ = \cQueen\ at \pt{p}{getLock:awaitX}.

Now, $p$'s local variable $flag$ is set to value \true\ for the first time in \Line{getLock:setFlag}.
If $p$ fails the if-condition of \Line{abort:ifNotRegistered}, then $p$ must have executed \Line{getLock:setFlag}, and thus $p$ broke out of the busy-wait loop of \Line{getLock:ApplyBotWant}.
Then, $p$ last executed line~\ref{getLock:awaitAckOrCtrDecrease} or line~\ref{getLock:awaitX} before calling \abort{p}.
Hence, $\Role[p] \in \Set{\cPawn,\cQueen}$ in $[\pt{p}{abort:ApplyWantOk},\pt{p}{abort:ifPawn}]$, since $p$ changes $\Role[p]$ next only in \Line{abort:RolePPawn}.

If $p$ satisfies the if-condition of \Line{abort:ifPawn}, then $\Role[p]$ = \cPawn, and $p$ changes $\Role[p]$ next only in \Line{abort:RolePPawn}.
Hence, $\Role[p] = \cPawn$ at $\pt{p}{abort:ifHead}$.
In \Line{abort:RolePPawn} $p$ changes $\Role[p]$ to \cPPawn\ and $p$ does not change $\Role[p]$ after that.
Hence, $\Role[p] = \cPPawn$ during $[\pt{p}{abort:RolePPawn},\pt{p}{abort:returninfty}]$.
If $p$ does not satisfy the if-condition of \Line{abort:ifPawn}, then $\Role[p] = \cQueen$ at  $[\ptB{p}{abort:setX},\pt{p}{abort:callhRelease}]$ follows.
\end{proof}

\begin{claim}
Suppose a process $p$ executes a call to \release{p}{j} during a passage.
The value of $\Role[p]$ at various points in time is as follows.\

\begin{tabular}{ l | l}
    \hline
    Points in time & Value of $\Role[p]$  \\ \hline
    $[\ptB{p}{release:safetyCheck},\ptB{p}{release:ifKing}]$ & \Set{\cKing,\cQueen,\cPPawn} \\ \hline
    $[\ptB{p}{release:ctr10},\pt{p}{release:callhRelease:King}]$ & \cKing\ \\ \hline
    $\ptB{p}{release:callhRelease:Queen}$ & \cQueen\ \\ \hline
    $\ptB{p}{release:callPromote}$ & \cPPawn\ \\ \hline
    $[\ptB{p}{release:ApplyOkBot},\pt{p}{release:return}]$ &  \Set{\cKing,\cQueen,\cPPawn} \\ \hline
\end{tabular}
\label{cl:TableOfRoles:release}
\end{claim}
\begin{proof}
Suppose the point in time \ptB{p}{release:safetyCheck}.
Then, $p$ is is executing a call to \release{p}{j}, and $p$ last executed a call to \lock{p} that returned a non-$\bot$ value.
Then, $p$'s call to \lock{p} either returned from \Line{getLock:end} in \lock{p} or from \Line{abort:returninfty} or \Line{abort:returnX} in \abort{p}.
From \Claim{cl:TableOfRoles:getLock}, $\Role[p] \in \Set{\cKing,\cQueen,\cPPawn}$ at time \ptB{p}{getLock:end} and from \Claim{cl:TableOfRoles:abort}, $\Role[p]  = \cPPawn$ at \ptB{p}{abort:returninfty} and $\Role[p] = \cQueen$ at \ptB{p}{abort:returnX}.
Therefore, $\Role[p] \in \Set{\cKing,\cQueen,\cPPawn}$ at time \ptB{p}{release:safetyCheck}.

From \Claim{cl:basic:Role}\refC{scl:RoleUnchanged}, $\Role[p]$ is unchanged during \release{p}{}.
Therefore, $\Role[p] \in \Set{\cKing,\cQueen,\cPPawn}$ during $[\ptB{p}{release:safetyCheck},\ptB{p}{release:ifKing}]$ and $[\ptB{p}{release:ApplyOkBot},\pt{p}{release:return}]$.
Then, from the if-conditions of lines~\ref{release:ifKing},~\ref{release:ifQueen} and~\ref{release:ifPPawn}, it follows immediately that $\Role[p] = \cKing$ during $[\ptB{p}{release:ctr10},\pt{p}{release:callhRelease:King}]$, and $\Role[p]$ = \cQueen\ at \ptB{p}{release:callhRelease:Queen}, and $\Role[p]$ = \cPPawn\ at \ptB{p}{release:callPromote}.

\end{proof}

\begin{claim}
Suppose a process $p$ executes a call to \doCollect{p}, \helpRelease{p} or \doPromote{p} during a passage.
The value of $\Role[p]$ at various points in time is as follows.\

\begin{tabular}{ l | l}
    \hline
    Points in time & Value of $\Role[p]$  \\ \hline
    $[\ptB{p}{collect:collectLoop},\pt{p}{collect:updateAll}]$ &  \Set{\cKing,\cQueen} \\ \hline
    $[\ptB{p}{hRelease:setT},\pt{p}{hRelease:end}]$ &  \Set{\cKing,\cQueen} \\ \hline
    $[\ptB{p}{promote:FR12},\pt{p}{promote:end}]$ & \Set{\cKing,\cQueen,\cPPawn}  \\ \hline
\end{tabular}
\label{cl:TableOfRoles:rest}
\end{claim}
\begin{proof}
From the code structure, $p$ does not change $\Role[p]$ during \doPromote{}, \doCollect{p} and \helpRelease{}.

From a code inspection, \doCollect{p} is called by $p$ only in lines~\ref{abort:doCollect}, and~\ref{release:doCollect}.
From Claim~\ref{cl:TableOfRoles:abort}, $\Role[p] = \cQueen$ at \ptB{p}{abort:doCollect} and from \Claim{cl:TableOfRoles:release}, $\Role[p] = \cKing$ at \ptB{p}{release:doCollect}.
Since $\Role[p]$ is unchanged during \doCollect{p}, it follows that $\Role[p] \in \Set{\cKing,\cQueen}$ during $[\ptB{p}{collect:collectLoop},\pt{p}{collect:updateAll}]$.

Now, suppose $p$ executes a call \helpRelease{p}.
From a code inspection, \helpRelease{p} is called by $p$ only in lines~\ref{abort:callhRelease},~\ref{release:callhRelease:King} and~\ref{release:callhRelease:Queen}.
From Claim~\ref{cl:TableOfRoles:abort}, $\Role[p] = \cQueen$ at \ptB{p}{abort:callhRelease} and from \Claim{cl:TableOfRoles:release}, $\Role[p] = \cKing$ at \ptB{p}{release:callhRelease:King} and $\Role[p] = \cQueen$ at \ptB{p}{release:callhRelease:Queen} .
Since $\Role[p]$ is unchanged during \helpRelease{}, it follows that $\Role[p] \in \Set{\cKing,\cQueen}$ during $[\ptB{p}{hRelease:setT},\pt{p}{hRelease:end}]$.

Now, suppose $p$ executes a call \doPromote{p}.
From a code inspection, \doPromote{p} is called by $p$ only in lines~\ref{release:callPromote} and~\ref{hRelease:callPromote}.
From Claim~\ref{cl:TableOfRoles:release}, $\Role[p] = \cPPawn$ at \ptB{p}{release:callPromote} and from earlier in this claim, $\Role[p] \in \Set{\cKing,\cQueen}$ at \ptB{p}{hRelease:callPromote}.
Since $\Role[p]$ is unchanged during \doPromote{}, it follows that $\Role[p] \in \Set{\cKing,\cQueen,\cPPawn}$ during $[\ptB{p}{promote:FR12},\pt{p}{promote:end}]$.
\end{proof}


\end{document}